\PassOptionsToPackage{dvipsnames,table}{xcolor}
\documentclass[a4paper,UKenglish,cleveref,autoref,thm-restate,numberwithinsect]{lipics-v2021}
\hideLIPIcs
\nolinenumbers
\csappto{ps@headings}{\csdef{@oddfoot}{}}
\usepackage{tikz}
\usetikzlibrary{calc,fit,decorations.pathmorphing,shapes,positioning,arrows.meta,shapes.geometric,patterns}

\usepackage[linesnumbered]{algorithm2e}

\RestyleAlgo{ruled}

\usepackage{enumitem}
\usepackage{booktabs}
\usepackage{capt-of}
\usepackage{xparse}
\usepackage{float}
\usepackage[braket]{mathtools}
\usepackage{mathrsfs}
\usepackage{mdframed}
\usepackage{stmaryrd}
\usepackage{csquotes}
\usepackage{etoolbox}

\sethlcolor{yellow}

\tikzset{
  >=Latex,
  rednode/.style={circle, draw=red, text=red, fill=white, minimum size=6mm, inner sep=0},
  bluenode/.style={circle, draw=blue, text=blue, fill=white, minimum size=6mm, inner sep=0},
  upnode/.style={circle, draw=black, text=black, fill=white, minimum size=8mm, inner sep=0},}

\AtEndEnvironment{example}{\hfill$\qed$}

\theoremstyle{plain}

\newcommand*{\sectionwiseHanchor}[1]{\csxdef{theH#1}{\noexpand\thesection.\noexpand\arabic{#1}}}
\forcsvlist{\sectionwiseHanchor}{theorem,lemma,corollary,proposition,exercise,definition,conjecture,observation,example,note,remark,claim}

\newcommand{\vars}{\mathsf{vars}}
\newcommand{\join}{\mathsf{join}}
\newcommand{\njoin}{\mathsf{njoin}}
\newcommand{\head}{\mathsf{head}}
\newcommand{\Q}{\mathbb{Q}}

\newcommand{\varout}{\mathsf{incomp}}
\newcommand{\finvals}[1]{\mathsf{finVals}(#1)}

\newcommand{\POS}[2]{\mathsf{pos}({#1,#2})}
\newcommand{\cqclass}{\mathsf{CQ}}
\newcommand{\nullclass}{\mathsf{CQ^{\perp}}}
\newcommand{\vcclass}{\mathsf{CQ^{vc}}}
\newcommand{\cmpclass}{\mathsf{CQ^{cmp}}}
\newcommand{\combclass}{\mathsf{CQ^{vc,\perp}}}
\newcommand{\cmpcombclass}{\mathsf{CQ^{cmp,\perp}}}
\newcommand{\Dcomb}{\mathcal{D}_{\mathrm{comb}}}
\newcommand{\nonnull}{\mathsf{nonnull}}
\newcommand{\valcnt}{\mathsf{val}}

\newcommand{\tox}[1]{\mathsf{match}(#1)}

\newcommand{\fulltox}[3]{\mathsf{match}(#1,#2,#3)}
\newcommand{\witset}{{\mathcal W}}
\newcommand{\feasible}[1]{\mathsf{eDom}(#1)}
\newcommand{\fullfeasible}[2]{\mathsf{eDom}_{#1}(#2)}

\newcommand{\body}{\mathsf{body}}
\newcommand{\compvc}[1]{\mathsf{comp}_{vc}(#1)}
\newcommand{\compvv}[1]{\mathsf{comp}_{vv}(#1)}

\newcommand{\allintervals}[1]{\mathsf{Inter}(#1)}

\newcommand{\Vval}[1]{V_{#1}}
\newcommand{\RG}{\mathcal{RG}}
\newcommand{\cmtox}[1]{\mathsf{cmatch}(#1)}
\newcommand{\cmctox}[1]{\mathsf{cmatch}^{-1}(#1)}

\newcommand{\compCG}{\mathsf{Comp}(\CG)}

\DeclarePairedDelimiterX{\set}[1]{\{}{\}}{#1}

\newcommand{\compCons}[1]{\mathsf{K}(#1)}

\newcommand{\CG}{\mathcal{CG}}
\newcommand{\OG}{\mathcal{OG}}
\newcommand{\actvars}[2]{\mathsf{vars}_{#1,#2}}

\newcommand{\repval}[2]{\mathsf{repvals}_{#1,#2}}
\newcommand{\rankmax}[2]{\mathsf{rmax}_{#2}(#1)}
\newcommand{\repvalx}[3]{\mathsf{repvals}_{#1,#2}(#3)}
\newcommand{\ceqord}{\prec}

\newcommand{\fullmcomp}[2]{\mathsf{matchComp}(#1,#2)}
\newcommand{\mcomp}{\mathsf{ matchComp}}
\newcommand{\fullrelcomp}[2]{\mathsf{relComp}(#1,#2)}

\newcommand{\forbidden}[2]{\mathsf{forbidden}_{#1}(#2)}

\newcommand{\cedge}{\mathsf{E_{cycle}}}
\newcommand{\mfac}{\mathsf{MFES}}

\newcommand{\myparagraph}[1]{\medskip\noindent\textbf{#1}}

\newcommand{\comp}{\mathrel{\lessdot}}
\newcommand{\altcomp}{\mathrel{\lessdot'}}
\renewcommand{\POS}[2]{\mathsf{POS}_{#1}(#2)}
\newcommand{\adom}{\mathrm{adom}}
\newcommand{\compvars}{\mathsf{compVars}}

\title{How Can We Shrink the Family of Test Databases? Query Containment with Nulls and Comparisons}
\titlerunning{Shrinking the Family of Test Databases}

\author{Helen Sternbach}{School of Computer Science and Engineering, The Hebrew University of Jerusalem, Israel}{helen.sternbach@mail.huji.ac.il}{https://orcid.org/0000-0002-6022-7982}{}
\author{Sara Cohen}{School of Computer Science and Engineering, The Hebrew University of Jerusalem, Israel}{scohen@mail.huji.ac.il}{https://orcid.org/0000-0002-8482-9435}{}

\authorrunning{H. Sternbach and S. Cohen}

\Copyright{Helen Sternbach and Sara Cohen}



\funding{The authors were partially funded by the Israel Science Foundation (ISF), grant no.~359/21. H.~Sternbach was partially funded by the Ariane de Rothschild Women Doctoral Program.}

\AtBeginDocument{\hypersetup{pdfsubject={}}}

\EventEditors{}
\EventNoEds{2}
\EventLongTitle{}
\EventShortTitle{}
\EventAcronym{}
\EventYear{}
\EventDate{}
\EventLocation{}
\EventLogo{}
\SeriesVolume{}
\ArticleNo{}

\begin{document}

\maketitle

\begin{abstract}

Query containment and equivalence drive database query optimization and rewriting. For plain conjunctive queries, both are decided by evaluating one query over a single canonical database of the other. This classical test breaks down in two settings that pervade real queries: databases with null values under SQL's three-valued semantics, and queries with order comparisons. In both, deciding containment is $\Pi_2^p$-complete, and the known characterizations replace it by an exponential family of test databases, leaving no practical route to certifying equivalence.

\looseness=-1 We ask how the family of test databases can be shrunk. For conjunctive queries over databases with nulls, we shrink the family to one that is exponential only in a special set of variables, and place containment in NP when that set has constant size. For queries with comparisons, we construct canonical values that decide containment, and then shrink the family by decomposing the test into independent components and by splitting it along the order conflicts. Finally, we combine the two features, and prove that the number of test databases is fixed-parameter tractable in three \emph{local} parameters of the two queries, with the null-only and comparison-only tests as special cases. Each test evaluates the containing query over a family of databases, an operation native to any database system.
\end{abstract}

\section{Introduction}\label{sec:intro}

Query containment and equivalence are among the most fundamental problems in database
theory. Given queries $q,q'$, \emph{containment} asks whether $q(D) \subseteq q'(D)$ for
every database $D$, and \emph{equivalence} is containment in both directions. The problems drive query
optimization~\cite{chandra1977optimal,chaudhuri1993optimization} and query rewriting
over materialized views and integrated data
sources~\cite{levy1993queries,ullman1997information}. They have also found new urgency: in a recent
industrial study roughly a third of over $3{,}100$ LLM-proposed rewrites of production
queries changed their results~\cite{narasayya2026leveraging}, and text-to-SQL benchmarks
rank models by deciding equivalence to a reference query~\cite{kim2025flex, klopfenstein2026spotit,zhong2020semantic}.

\looseness=-1 For the core class of \emph{conjunctive queries} (CQs), Chandra and
Merlin~\cite{chandra1977optimal} characterized containment by the existence of a
containment mapping, or, equivalently: $q \subseteq q'$ if and only if $q'$ returns the
frozen head tuple over the \emph{canonical database} $D_q$, obtained from $q$ by freezing
each variable into a constant. Such a test is attractive in practice: it amounts to
evaluating $q'$ over concrete databases, an operation every database system performs
natively, so equivalence is certified without separate machinery for finding
homomorphisms.

\looseness=-1 This simple picture breaks down in two settings that pervade real queries. Databases
routinely contain \emph{null values}, which SQL interprets under a three-valued logic,
and queries routinely contain \emph{comparisons}. In both
settings, deciding containment is
$\Pi_2^p$-complete~\cite{farre2007containment,klug1988conjunctive,van1992complexity},
and a single canonical database no longer suffices. The known characterizations replace $D_q$ by
an exponential family: all \emph{null versions} of $D_q$ in the first setting, and
canonical databases for all orderings of the variables and constants in the second.
$\Pi_2^p$-completeness, tested through an exponential family, leaves no practical route
to certifying equivalence.

\myparagraph{Contributions.}
The central question of this paper is: \emph{how can we shrink the family of test
databases that decides containment?} For CQs over databases with nulls (Section~\ref{NullCQ}), we replace the
known exponential family by one that is exponential only in a special set of variables, and place containment in NP when that set has constant size. For queries with
comparisons (Section~\ref{IneCQ}), we give a family of canonical databases that decides
containment, and then shrink it twice. The test decomposes into independent parts, and it
splits along the order conflicts between the two queries. Finally, we combine the two features (Section~\ref{sec:combined}). The
constructions compose, because a comparison is never satisfied by a null, and the size of
the combined family is \emph{fixed-parameter tractable} in three local parameters of the
two queries. For each result we give the idea of its proof, deferring the full proofs to
Appendix~\ref{app:proofs}.

\section{Formal Framework}\label{sec:prelim}

\myparagraph{Databases and Queries.}
A \emph{schema} $\sigma$ is a finite set of relation symbols, each with a fixed
\emph{arity}.
A \emph{database} $D$ over $\sigma$ assigns to each  $R \in \sigma$
of arity $k$ a finite relation $R^D \subseteq (\mathbb{Q} \cup \{\perp\})^k$, where
$\mathbb{Q}$ denotes the rationals and $\perp$ is a distinguished \emph{null} symbol
representing an unknown value.
A database is \emph{ordinary} (or \emph{null-free}) if $R^D \subseteq \mathbb{Q}^k$
for every $R \in \sigma$.

A \emph{term} is either a variable or a constant from $\mathbb{Q}$.
A \emph{relational atom} over $\sigma$ has the form $R(t_1,\ldots,t_k)$ where $R \in \sigma$
has arity $k$ and each $t_i$ is a term.
A \emph{comparison atom} has the form $s \comp s'$ where $s, s'$ are terms and
$\comp \in \{<,\leq\}$. This is no restriction, since $>$ and $\geq$ are the mirror images
of $<$ and $\leq$, and an equality $s = s'$ is the pair $s \le s'$, $s' \le s$.

A \emph{query} $q$ has the form
$q(\bar{x}) \leftarrow R_1(\bar{t}_1),\ldots,R_n(\bar{t}_n),\;
  C_1,\ldots,C_k$
where each $R_i(\bar{t}_i)$ is a relational atom, each $C_j$ is a comparison atom,
and $\bar{x}$ is a tuple of variables called the \emph{head}.
The set $\{R_1(\bar{t}_1),\ldots,R_n(\bar{t}_n)\}$ is the \emph{relational body} of
$q$, and together with the comparison atoms it forms $\body(q)$.
We say $q$ is \emph{Boolean} when $\bar{x}$ is empty.
We write $\compvc{q}$ for the comparison atoms of $q$ with a
constant argument, and $\compvv{q}$ for those with two variables.

Throughout, without loss of generality, queries are \emph{normalized}: \textbf{(N1)}~no two distinct variables are
made equal by comparisons, and \textbf{(N2)}~every argument of a relational atom is a
variable, so constants occur only in comparison atoms. 
We use  $\cqclass$ for the class of  \emph{conjunctive queries}
(\emph{CQs}), which have no comparisons and are evaluated over ordinary
databases, $\nullclass$ for \emph{CQs with null values} without comparisons evaluated over databases that may contain $\perp$ and $\cmpclass$, which
allows arbitrary comparison atoms over ordinary databases and its subclass
$\vcclass \subseteq \cmpclass$, in which every comparison is between a variable and a constant.

 \myparagraph{Variable Classification.}
The \emph{variables} of $q$, written $\vars(q)$, are all variables appearing in the
query. Variables appearing in the head are the \emph{output variables}, $\mathsf{head}(q)$.
A variable $v \in \vars(q)$ is a \emph{join variable} if it occurs at least twice in
the relational atoms of $q$ (possibly twice in a single atom), and otherwise it is a
\emph{non-join variable}. We write $\join(q)$ and $\njoin(q)$ for these two sets.
The \emph{position set} $\POS{q}{v}$ of $v$ in $q$ is the set of pairs $(R,j)$ such that $v$ appears in position $j$ of some relational atom of $q$ over the relation $R$.
\begin{example}
\looseness=-1 For $q(x) \leftarrow R(x,y),\ R(z,y),\ S(z,w)$, the position sets are
$\POS{q}{x}=\{(R,1)\}$, $\POS{q}{y}=\{(R,2)\}$, $\POS{q}{z}=\{(R,1),(S,1)\}$, and
$\POS{q}{w}=\{(S,2)\}$. The variables $y$ and $z$ occur twice and are join variables,
with $y$ recurring in a single position and $z$ across two, while $x,w\in\njoin(q)$.
\end{example}

\myparagraph{Query Evaluation.}
A \emph{valuation} for $q$ over database $D$ is a function
$\mu\colon \vars(q) \to \mathbb{Q} \cup \{\perp\}$,
extended to terms by $\mu(c) = c$ for every constant $c \in \mathbb{Q}$.
The valuation $\mu$ \emph{satisfies} a relational atom $R(t_1,\ldots,t_k)$ in $D$ if
$(\mu(t_1),\ldots,\mu(t_k)) \in R^D$, and it \emph{satisfies} a comparison atom
$s \comp s'$ if $\mu(s), \mu(s') \in \mathbb{Q}$ and $\mu(s) \comp \mu(s')$ holds. If either argument is $\perp$ the comparison atom is not satisfied, following SQL's
three-valued logic. We say $\mu$ is a
\emph{satisfying valuation} for $q$ in $D$ if it satisfies every atom in $\body(q)$
and maps every join variable of $q$ to a non-$\perp$ value. The last requirement reflects that two occurrences of a null are never equal under SQL
semantics, and is vacuous over null-free databases.
The \emph{result} of $q(\bar{x})$ on $D$ is
$q(D) := \{\,\mu(\bar{x}) \mid \mu \text{ is a satisfying valuation for } q \text{ in } D\,\}$.
The output tuple $\mu(\bar{x})$ may itself contain nulls if a head variable is mapped
to $\perp$ by a relational atom.

\myparagraph{Canonical Database and Containment.}
\looseness=-1 The \emph{canonical database} $D_q$ of a query $q$ is the ordinary database obtained
by treating each variable as a distinct constant (``freezing'' it): $D_q$ contains
the tuple $\iota(\bar{t}_i)$ in relation $R_i$ for each relational atom
$R_i(\bar{t}_i)$ in $\body(q)$, where $\iota$ is the identity on variables viewed
as constants.

Let $q$ and $q'$ be queries of the same output arity evaluated over a class
$\mathcal{D}$ of databases. We say $q$ is \emph{contained} in $q'$ w.r.t.\
$\mathcal{D}$, written $q \subseteq_{\mathcal{D}} q'$, if $q(D) \subseteq q'(D)$ for
all $D \in \mathcal{D}$, and the queries are \emph{equivalent}, $q \equiv_{\mathcal{D}} q'$,
if containment holds in both directions.
When $\mathcal{D}$ is all ordinary databases we write $q \subseteq q'$. When $\mathcal{D}$ includes databases with nulls we write
$q \subseteq_\perp q'$.

For plain CQs, containment is characterized by \emph{containment mappings}~\cite{chandra1977optimal}:
$q \subseteq q'$ iff there exists a homomorphism from the relational body of $q'$ to
that of $q$ mapping $\mathsf{head}(q')$ to $\mathsf{head}(q)$, or equivalently, iff
$q'$ returns $\bar{x}$ on $D_q$.
Both null-containment and containment for $\cmpclass$ are $\Pi_2^P$-complete
\cite{farre2007containment,klug1988conjunctive,van1992complexity}
and cannot be decided by evaluating $q'$ on $D_q$ alone.

\section{Null Conjunctive Queries}\label{NullCQ}

Null conjunctive queries ($\nullclass$) are conjunctive queries evaluated, under the SQL
three-valued semantics of Section~\ref{sec:prelim}, over databases that may contain null
values.

\begin{example}\label{example:null-containment}
A genealogy stores each person's (possibly \emph{unknown}) parent and birth and death years in
$\mathrm{Person}(\mathit{parent},\mathit{child},\mathit{childBirthYr},\mathit{childDeathYr})$. Consider the queries
\[
  \begin{aligned}
    q (p_1,p_2) &\leftarrow \mathrm{Person}(p_1,p_2,b_2,d_2),\ \mathrm{Person}(p_2,p_3,b_3,d_3),\\
    q'(p_1,p_2) &\leftarrow \mathrm{Person}(p_1,p_2,b_2,d_2),\ \mathrm{Person}(p_1,p_3,b_3,d_3).
  \end{aligned}
\]
So $q$ returns grandparent--parent pairs, while $q'$ pairs a parent with a child when it has
another. Over complete databases every pair returned by $q$ is also returned by $q'$, so
$q \subseteq q'$. With nulls it fails: $q$ returns $(\perp,\mathrm{Alice})$
over $\{$\mbox{$\mathrm{Person}(\perp,\mathrm{Alice},1920,1980)$}, \mbox{$\mathrm{Person}(\mathrm{Alice},\mathrm{Bob},1940,2000)$}$\}$,
where Alice's parent is unknown, while $q'$ cannot return this tuple.
\end{example}

Deciding \emph{null-containment} is $\Pi_2^P$-complete in general and
NP-complete for Boolean queries~\cite{farre2007containment}.
\looseness=-1 We focus on characterizing containment via canonical databases, rather than homomorphisms (see Section~\ref{sec:related} for homomorphism-based conditions). Recall $D_q$ from
Section~\ref{sec:prelim}. A \emph{null version} of $D_q$ replaces some of the frozen
non-join variables of $q$ by $\perp$. For a set $N \subseteq \njoin(q)$, let $\theta_N$ be the
substitution that maps the variables of $N$ to $\perp$ and fixes every other variables. We write
$\theta_N D_q$ for the resulting database. Farr\'e et al.\ showed the following result.

\begin{theorem}[\cite{farre2007containment}]\label{thm:null-farre}
Let $q(\bar{x}), q'(\bar{x})$ be $\nullclass$ queries. Then $q \subseteq_\perp q'$
if and only if, for every $N \subseteq \njoin(q)$, the query $q'$ returns the tuple
$\theta_N \bar{x}$ over $\theta_N D_q$.
\end{theorem}

Since a null version may set \emph{any} subset of the non-join variables to $\perp$,
Theorem~\ref{thm:null-farre} ranges over $2^{|\njoin(q)|}$ databases. We show
that it suffices to toggle only a small subset of them.

\looseness=-1 Throughout this section $q$ and $q'$ are fixed. Everything turns on how the non-join
variables of $q$ relate to the head of $q'$. Call a non-join variable $v$ of $q$
\emph{covered} when it occupies the same positions as some head variable of $q'$, that
is, $\POS{q}{v} = \POS{q'}{x'}$ for some $x' \in \head(q')$, and let $V_{\mathrm C}$ be
the set of covered variables. Among these, $V_{\mathrm{CJ}}$ contains the variables $v$
whose position set also satisfies $\POS{q}{v} = \POS{q'}{y'}$ for some join variable
$y'$ of $q'$, and $V_{\mathrm{CN}} = V_{\mathrm C} - V_{\mathrm{CJ}}$ contains the rest.
Note that since $v$ is a non-join variable, $|\POS{q}{v}| = 1$, so both equalities are
equivalent to the containments $\POS{q'}{x'} \subseteq \POS{q}{v}$ and
$\POS{q'}{y'} \subseteq \POS{q}{v}$.

When creating a canonical database from $q$, we must decide which
variables to freeze and which to null. Join variables must be frozen, so that $q$
returns a result over the canonical database. For a covered variable there is a
duality: nulling it may prevent $q'$ from returning a result, as join variables of $q'$
cannot be mapped to nulls, but nulling it may also enable $q'$ to return a result with
$\perp$ in the output, by mapping a head variable of $q'$ to it. Given this duality, a
variable of $V_{\mathrm{CJ}}$ must be toggled between frozen and null. A variable of
$V_{\mathrm{CN}}$ can safely be frozen, as it never prevents $q'$ from returning a
result, and the remaining non-join, non-head variables can safely be nulled, as $q'$
cannot use them to produce a result with $\perp$ in the output.

Formally, we freeze the variables of
$V_f = \join(q) \cup V_{\mathrm{CN}} \cup (\head(q)-V_{\mathrm{CJ}})$, toggle the
variables of $V_t = V_{\mathrm{CJ}}$, and null the remaining variables, denoted $V_n$.
For a subset $V\subseteq V_t$, let $\theta_{V\cup V_n}$ be the mapping that nulls all
variables in $V\cup V_n$ and freezes every other variable of $q$.

\begin{theorem}\label{thm:null-tight}
$q \subseteq_\perp q'$ if and only if $q'$ returns $\theta_{V\cup V_n} \bar{x}$ over $\theta_{V\cup V_n} D_q$ for every
$V \subseteq V_t$.
\end{theorem}

The proof recasts the test of Theorem~\ref{thm:null-farre} as the existence of a homomorphism from the relational body of $q'$ to that of $q$ that avoids the nulled variables on join variables and matches the head up to nulling. It then shows that the homomorphism witnessing the toggled database with $V=N\cap V_t$ also witnesses the null version of an arbitrary $N\subseteq\njoin(q)$, so the $2^{|V_t|}$ toggled databases subsume all $2^{|\njoin(q)|}$ null versions.

\begin{corollary}\label{cor:single-db}
If $|V_t|$ has constant size, null containment is in NP. Furthermore, if $|V_t|$ has constant size and 
$q'$ admits polynomial evaluation, then null containment is in P. 
\end{corollary}

For example, if $q$ and $q'$ are Boolean, then $V_t$ is empty, so
Theorem~\ref{thm:null-tight} extends the previously known NP upper bound for the
Boolean case to a wider class of queries. Moreover, if additionally $q'$ is
acyclic~\cite{yannakakis1981algorithms}, then null containment is in P.

\begin{example}\label{ex:null-savings}
Recall Example~\ref{example:null-containment}. The query $q$ has six non-join
variables, so Theorem~\ref{thm:null-farre} considers $2^6 = 64$ null versions of $D_q$.
The covered variables are $p_1$ and $p_3$, with $V_{\mathrm{CJ}} = \{p_1\}$ and
$V_{\mathrm{CN}} = \{p_3\}$, giving $V_f = \{p_2, p_3\}$, $V_t = \{p_1\}$, and
$V_n = \{b_2, d_2, b_3, d_3\}$. Theorem~\ref{thm:null-tight} tests only two databases:
$D_1$, which nulls $V_n$ and freezes all other variables ($V = \emptyset$), and $D_2$,
which also nulls $p_1$ ($V=\{p_1\}$). Indeed, $D_2$ shows non-containment.
\end{example}

\section{Conjunctive Queries with Inequalities}\label{IneCQ}
\looseness=-1 We now turn to conjunctive queries with comparison atoms, evaluated over null-free databases. Sections~\ref{sec:matching-witness} and~\ref{sec:vc-components} develop the containment test for the variable--constant class $\vcclass$, and Sections~\ref{sec:cmp} and~\ref{sec:addable} extend it to the full class $\cmpclass$.

\subsection{Witness Sets and the Containment Criterion}\label{sec:matching-witness}
Recall the classes $\cmpclass$ (queries
with comparison atoms) and $\vcclass\subseteq\cmpclass$ (whose comparisons are all
variable--constant), and the position set $\POS q v$. The \emph{effective domain}
$\fullfeasible q v$ of a variable $v$ is the set of values in $\Q$ satisfying all
comparison atoms of $q$ involving $v$, written $\feasible v$ when $q$ is clear. For a
variable with no comparisons, $\fullfeasible q v = \Q$. Since the comparison operators
are $<$ and $\le$, every effective domain is an interval, possibly unbounded.
We call such subsets of $\Q$ \emph{subdomains}.

\looseness=-1 Given $q$ and $q'$, we seek the variables of $q'$ that can correspond to
those of $q$. We say that a variable $y\in \vars (q')$ {\em matches\/} $x\in \vars(q)$ if (1)~$\POS {q'} y \subseteq \POS q x$, and (2)~$\feasible y\cap \feasible x \neq \emptyset$. We write
$\fulltox{q'}q x$, or simply $\tox x$, for the set of variables of $q'$ that match
$x$.

Given a subdomain $E$ of $\Q$ and a set of variables $V\subseteq \vars(q)$, we will say that $V$ is {\em incompatible\/}  with $E$ written $V\nvDash E$ if $\feasible v\cap E = \emptyset$, for all $v\in V$. Given a set of variables $V\subseteq \vars(q)$ such that some are incompatible with $E$, while others are not, we use $\varout(E,V)$ to denote the maximal subset $V'$ of $V$ such that $V'\nvDash E$.

\begin{example}[Running example]\label{exm:vc-running}
Throughout this subsection, consider a bank database with the relations
$\mathrm{Trans}(\mathit{acct},\mathit{amount},\mathit{hour},\mathit{fee})$,
$\mathrm{Account}(\mathit{acct},\mathit{balance},\mathit{rate})$, and
$\mathrm{Flagged}(\mathit{acct})$, which lists the accounts flagged for review.
Account numbers up to $50$ are reserved for internal accounts, and first-branch accounts lie strictly between $100$ and $200$. Consider the Boolean $\vcclass$ queries
\[
\begin{aligned}
q() \leftarrow\ & \mathrm{Trans}(a,m,h,e),\ \mathrm{Account}(a,b,r),\ \mathrm{Flagged}(a),\ 100<a<200,\ m\ge 1000\\
q'() \leftarrow\ & \mathrm{Trans}(a',m',h',e'),\ \mathrm{Account}(a',b',r'),\ \mathrm{Flagged}(a'),\ \mathrm{Flagged}(u),\\
& 150\le a'\le 400,\ 500<m'<2000,\ u\le 50
\end{aligned}
\]
Query $q$ asks whether some flagged account of the first branch made a transaction of at least $1000$. Query $q'$ asks whether some flagged account numbered between $150$ and $400$ made a transaction of an amount between $500$ and $2000$, and some internal account is flagged. Neither query constrains the hour, fee, balance, and rate variables, so their effective domains are $\Q$. The remaining effective domains are $\fullfeasible q a=(100,200)$ and $\fullfeasible q m=[1000,\infty)$ in $q$, and $\fullfeasible{q'}{a'}=[150,400]$, $\fullfeasible{q'}{m'}=(500,2000)$, and $\fullfeasible{q'}{u}=(-\infty,50]$ in $q'$.
Both $\POS{q'}{a'}$ and $\POS{q'}{u}$ are subsets of $\POS q a=\set{(\mathrm{Trans},1),(\mathrm{Account},1),(\mathrm{Flagged},1)}$. Since $\fullfeasible{q'}{a'}\cap \fullfeasible q a=[150,200)\neq\emptyset$, the variable $a'$ matches $a$. In contrast, $\fullfeasible{q'}{u}\cap \fullfeasible q a=\emptyset$: an internal account cannot play the role of the branch account of $q$, so $u$ does not match $a$. Hence $\tox a=\set{a'}$, and each other variable of $q$ is matched exactly by its primed counterpart. For $\varout$, take $E=[200,\infty)$. Only $\feasible a$ is disjoint from $E$, so $\varout(E,\vars(q))=\set{a}$.
\end{example}

\looseness=-1 For any set $S\subseteq \tox x$, write
$E_x(S) = \bigcap_{y\in S}\bigl(\fullfeasible q x \setminus \fullfeasible {q'} y \bigr)$
for the subdomain it induces. We say that $S$ is {\em contradictable by\/} $x$ if $E_x(S)\neq\emptyset$, and {\em infinite-contradictable by\/} $x$ if $E_x(S)$ is infinite. If $S$ is maximal among the infinite-contradictable subsets of $\tox x$, then $S$ is {\em maximally infinite-contradictable by\/} $x$. Several sets may be maximally infinite-contradictable by $x$, as Example~\ref{exm:witness-running} shows.
Let $\finvals x$ be the set of values of $\feasible x$ that lie in $E_x(S)$ for some $S\subseteq\tox x$ with $E_x(S)$ finite. A finite $E_x(S)$ may contain more than one value, for example the two endpoints of $\feasible x$.

\looseness=-1 The {\em witness set\/} $\witset_x$ of $x$ collects the subdomains $E_x(S)$ for the sets $S\subseteq\tox x$ that are maximally infinite-contradictable by $x$, together with the singleton $\set c$ for every $c\in\finvals x$.

\begin{example}[Witness sets]\label{exm:witness-running}
Over the bank schema, consider the queries
\[
\begin{aligned}
q() &\leftarrow \mathrm{Flagged}(a),\ 100 \le a \le 200\\
q'() &\leftarrow \mathrm{Flagged}(u_1),\ \mathrm{Flagged}(u_2),\ \mathrm{Flagged}(u_3),\ u_1 < 150,\ u_2 > 150,\ u_3 \ge 120
\end{aligned}
\]
We have $\tox a = \set{u_1,u_2,u_3}$, $\fullfeasible q a = [100, 200]$, and $E_a(\set{u_1}) = [150, 200]$, $E_a(\set{u_2}) = [100, 150]$, and $E_a(\set{u_3}) = [100, 120)$. Exactly two sets are maximally infinite-contradictable by $a$: the set $\set{u_1}$, since adding $u_2$ leaves the finite set $\set{150}$ and adding $u_3$ leaves the empty set, and the set $\set{u_2,u_3}$, with witness $[100,120)$. The only finite $E_a(S)$ is $E_a(\set{u_1,u_2})=\set{150}$, so $\finvals a=\set{150}$ and
$\witset_a = \set{\, [150,200],\ [100, 120),\ \set{150} \,}.$
\end{example}

A larger witness-set computation is given in Example~\ref{exm:witness-large} in Appendix~\ref{app:proofs-matching-witness}.
We now show a bound on the size of the witness set for any given variable.
\begin{proposition}[Upper bound on witness-set size]\label{prop:witness_set_size}
Let $q$ and $q'$ be queries, and let $x$ be a variable in $\vars(q)$.
Then, $|\witset_x| \ \le\ 2\,|\tox{x}|+1\,.$
\end{proposition}

\looseness=-1 The proof splits $\feasible x$ into maximal segments on which the set of incompatible matched variables is constant. There are at most $2|\tox{x}|+1$ such segments, since each variable of $\tox x$ contributes two interval endpoints, and every witness occupies a segment of its own.

To determine containment of $q$ in $q'$, we create a set of {\em canonical databases\/} out of {\em canonical values\/} for the variables of $q$.
For each $x\in\vars(q)$ and each witness $E\in\witset_x$, we define a constant $c(x,E)$ as follows. If $E=\{c\}$ is a singleton, then $c(x,E)=c$. Otherwise $E$ is infinite, and $c(x,E)$ is a constant chosen from $E$, distinct from all constants of $q$ and $q'$, from all values in $\bigcup_{x\in\vars(q)}\finvals x$, and from the constants chosen for other infinite witnesses. The set of {\em canonical values\/} of $x$ is $c(x)=\set{c(x,E) : E\in \witset_x}$.

\looseness=-1 A {\em canonical assignment\/} $\theta$ maps each variable $x$ of $q$ to a canonical value in $c(x)$, and $\Theta_q$ denotes the set of all canonical assignments. The {\em canonical database\/} $\theta(D_q)$ is obtained from $D_q$ (Section~\ref{sec:prelim}) by replacing each frozen variable $v$ with the value $\theta(v)$. By construction, $q$ outputs $\theta(\bar{x})$ over $\theta(D_q)$ for every $\theta \in \Theta_q$. We write $\mathcal D(\Theta_q)$ for the set of all canonical databases for~$q$.

\begin{example}[Running example, cont.]
\looseness=-1 For the account variable $a$ of Example~\ref{exm:vc-running}, the only maximally infinite-contradictable set is $\set{a'}$, with the infinite witness $(100,200)\setminus[150,400]=(100,150)$, so $c(a)$ is a single value, say $120$. Likewise $\witset_m=\set{[2000,\infty)}$, and the witness set of each of $h$, $e$, $b$, and $r$ is $\set{\Q}$, as no nonempty subset of its match set is contradictable. 
Thus, every variable in $q$ has exactly one canonical value, yielding a \emph{single} canonical database in $\mathcal D(\Theta_q)$, even though $q$ has six variables and the two queries mention eight constants.
In contrast, the construction of Klug~\cite{klug1988conjunctive} enumerates one canonical database for every ordering of these variables and constants.
\end{example}

\begin{theorem}\label{thm:vc-criterion}
Let \(q(\bar{x}),q'(\bar{x})\) be $\vcclass$ queries. Then,
\(q\) is contained in \(q'\) if and only if $q(D)\subseteq q'(D)$ for all $D\in \mathcal D(\Theta_q)$.
\end{theorem}

The nontrivial direction turns a valuation $\nu$ that witnesses $\bar a \in q(D)$ into a valuation of $q'$ over $D$. A canonical assignment $\theta$ is chosen to \emph{mimic} $\nu$: for every $x\in\vars(q)$, the value $\theta(x)$ avoids the effective domain of a matched variable in $\tox x$ exactly when $\nu(x)$ does. By assumption, some valuation $\theta'$ of $q'$ produces $\theta(\bar x)$ over $\theta(D_q)$, and the composition $\nu\circ\theta^{-1}\circ\theta'$ is the required valuation. The composition is well defined although $\theta$ is not injective, since variables that share a $\theta$-value also share their $\nu$-value, and it satisfies the comparison atoms of $q'$ precisely because $\theta$ mimics $\nu$.

\subsection{Reducing to Independent Components}\label{sec:vc-components}

The number of canonical values of each variable $x$ is at most $2|\tox{x}|+1$, so by
Theorem~\ref{thm:vc-criterion} deciding containment amounts to evaluating $q'$ over all
$\prod_{x\in\vars(q)}|c(x)|$ canonical databases. Often, far fewer suffice. Only variables
that share an atom of $q$ or realize a common atom of $q'$ interact, so recombining the
values of unrelated parts of $q$ tests nothing new. Example~\ref{exm:components-running}
shows the redundancy concretely.

\begin{example}[Running example for Section~\ref{sec:vc-components}]\label{exm:components-running}
We illustrate over four relations recording the transactions on each account: $\mathrm{Deposit}(\mathit{amt},\mathit{acct})$, $\mathrm{Withdraw}(\mathit{amt},\mathit{acct})$, $\mathrm{Fee}(\mathit{amt},\mathit{acct})$, and $\mathrm{Interest}(\mathit{amt},\mathit{acct})$:
\[
\begin{aligned}
q() \leftarrow\ & \mathrm{Deposit}(d,a),\ \mathrm{Withdraw}(w,a),\ \mathrm{Fee}(f,a),\ \mathrm{Interest}(g,a),\\
   & 0\le d,w\le 1000,\ 10\le f,g\le 50\\
q'() \leftarrow\ & \mathrm{Deposit}(t,c_1),\ \mathrm{Withdraw}(t,c_2),\ \mathrm{Deposit}(v_1,c_1),\ \mathrm{Withdraw}(v_2,c_2),\\
   & \mathrm{Fee}(v_3,c_3),\ \mathrm{Interest}(v_4,c_4),\ v_1,v_2>0,\ v_3,v_4>10
\end{aligned}
\]
In $q'$, the same amount $t$ is deposited into one account and withdrawn from another, and each account also has a transaction with a positive amount.

\looseness=-1 Only the variable $a$ of $q$ occupies more than one position, namely the second position of every atom, and $\POS{q'}{t}=\{(\mathrm{Deposit},1),(\mathrm{Withdraw},1)\}$. Here $\tox{d}=\{v_1\}$, and inside $\feasible{d}=[0,1000]$ the variable $v_1$ induces the finite witness $E_{d}(\{v_1\})=[0,1000]\setminus(0,\infty)=\{0\}$, so $\witset_{d}=\{[0,1000],\{0\}\}$ and $|c(d)|=2$. The same computation for $w,f,g$ gives two canonical values each, where the finite-witness value is the \emph{boundary constant} $0$ for $d,w$ and $10$ for $f,g$. These boundary constants are \emph{shared}: $0$ is a canonical value of both $d$ and $w$, and $10$ of both $f$ and $g$. The account variable $a$ is a \emph{hub}: it appears in every atom, so the amount variables meet only through it. Matched only by the unconstrained $c_1,\ldots,c_4$, it has a single canonical value. Once a value for $a$ is fixed, whether the Fee atom of $q'$ is realized depends on the value of $f$ alone, and likewise for the other atoms. Joint variation of $d,w,f,g$ is therefore largely redundant, and it should suffice to vary each amount variable around the hub separately.
\end{example}

\providecommand{\RG}{\mathcal{RG}}
\providecommand{\cmtox}[1]{\mathsf{cmatch}(#1)}
\providecommand{\cmctox}[1]{\mathsf{cmatch}^{-1}(#1)}
\providecommand{\Vval}[1]{V_{#1}}

Two subtleties make this plan delicate. First, a single variable of $q'$ may be realized
through \emph{several} variables of $q$ that share a value: in
Example~\ref{exm:components-running}, no variable of $q$ matches $t$, yet $q'$ realizes
$t$ on the databases where $d$ and $w$ both carry the shared boundary constant $0$.
Second, a single relational atom of $q'$ may then depend on two parts at once: realizing
$\mathrm{Deposit}(t,c_1)$ and $\mathrm{Withdraw}(t,c_2)$ requires $d$ and $w$ to take the
value $0$ together, so these two variables cannot be varied separately after all. We first
capture the sharing phenomenon by a relaxation of the matching relation, and then build it
into a graph whose components can be varied safely.

Recall (Section~\ref{sec:matching-witness}) that boundary (finite-witness) values are constants
of the queries and may be \emph{shared} by several variables, while infinite-witness values
are fresh and distinct. For a value $a$, the \emph{cohort}
$\Vval a=\set{x\in\vars(q)\mid a\in c(x)}$ collects the variables of $q$ that carry $a$ as
a canonical value, and sharing is exactly the situation $|\Vval a|\ge 2$. We say that $a$
is a \emph{witness value} for $y\in\vars(q')$ if (1)~$a\in\feasible y$, and (2)~the cohort
of $a$ collectively covers every position of $y$, that is,
$\bigcup_{x\in\Vval a}\POS q x\supseteq\POS{q'}y$. We say that $y$ \emph{covering-matches}
$x$, written $x\in\cmctox y$ (equivalently $y\in\cmtox x$), if some witness value for $y$
is a canonical value of $x$. Collective covering is exactly what a homomorphism guarantees
when it sends $y$ to a shared value.
Covering-match serves only to group the variables of $q$ in the decomposition below. The canonical values are unchanged, so Theorem~\ref{thm:vc-criterion} is unaffected.

\begin{example}[Running example, cont.]
The value $0$ is a witness value for $t$: it lies in $\feasible{t}=\Q$, and
$d$ and $w$ collectively cover $\POS{q'}{t}$, although neither does alone. Hence
$\cmctox{t}\supseteq\Vval 0=\{d,w\}$, even though $t$ has no ordinary match. All other
variables of $q'$ have singleton cohorts, on which the two notions of match coincide.
\end{example}

\looseness=-1 For a relational atom $A=R(y_1,\dots,y_k)$ of $q'$, write
$M_A=\bigcup_{t=1}^{k}\cmctox{y_t}$ for the variables of $q$ covering-matched by some
variable of $A$. The \emph{relational graph} $\RG$ has vertex set $\vars(q)$, an
\emph{atom edge} between distinct $x,x'$ that occur together in a relational atom of $q$,
and a \emph{co-match edge} between distinct $x,x'\in M_A$ for a relational atom $A$ of
$q'$. The co-match edge is what keeps $d$ and $w$ of the running example together. A
{\em separator\/} $S\subseteq\vars(q)$ is \emph{legal} if no variable of $q'$
covering-matches two variables of $S$, that is, $\cmtox x\cap\cmtox{x'}=\emptyset$ for all
distinct $x,x'\in S$. Given a legal separator $S$, let $C_1,\dots,C_m$ be the connected
components of $\RG-S$.

A subset $\Psi_q\subseteq\Theta_q$ of the canonical assignments of
Section~\ref{sec:matching-witness} is \emph{$S$-exhaustive} if, for every $\theta\in\Theta_q$
and every $i\in[m]$, some $\psi\in\Psi_q$ has $\psi|_{S\cup C_i}=\theta|_{S\cup C_i}$.
Intuitively, $\Psi_q$ varies each component, together with the separator, independently of
the other components.

\begin{example}[Running example, cont.]
\looseness=-1 The relational graph $\RG$ (Figure~\ref{fig:components-relational-graph}) has an atom edge
between $a$ and each of $d,w,f,g$, and the single co-match edge $d\!-\!w$, since
$d,w\in M_{\mathrm{Deposit}(t,c_1)}$. Although $f$ and $g$ also share a boundary constant, no
variable of $q'$ is collectively covered by them, so no co-match edge arises. The separator
$S=\{a\}$ is legal, and deleting it leaves $C_1=\{d,w\}$, $C_2=\{f\}$, and $C_3=\{g\}$.
Without the co-match edge, $d$ and $w$ would fall in different components, and the atom
$\mathrm{Deposit}(t,c_1)$ of $q'$ would no longer be realized inside a single one.
\end{example}

\begin{figure}[t]
\centering
\begin{subfigure}[b]{0.35\textwidth}
\centering
\begin{tikzpicture}[
    scale=0.52,
    every node/.style={font=\small},
    v/.style={circle, draw=black, line width=1pt, minimum size=0.72cm, inner sep=0pt},
    sep/.style={circle, draw=black, fill=black!12, line width=1pt, minimum size=0.72cm, inner sep=0pt},
    blackedge/.style={draw=black, line width=1pt},
    rededge/.style={draw=red, line width=1.2pt}
]
\node[sep] (a) at (0,0) {$a$};
\node[v] (d) at (-3.2,1.6) {$d$};
\node[v] (w) at (-3.2,-1.6) {$w$};
\node[v] (f) at (3.2,1.6) {$f$};
\node[v] (g) at (3.2,-1.6) {$g$};

\draw[blackedge] (a)--(d);
\draw[blackedge] (a)--(w);
\draw[blackedge] (a)--(f);
\draw[blackedge] (a)--(g);
\draw[rededge] (d)--(w);

\node[draw=black,dashed,line width=1pt,rectangle,fit=(d)(w),inner sep=0.3cm,label=above:{$C_1$}] {};
\node[draw=black,dashed,line width=1pt,rectangle,fit=(f),inner sep=0.25cm,label=above:{$C_2$}] {};
\node[draw=black,dashed,line width=1pt,rectangle,fit=(g),inner sep=0.25cm,label=below:{$C_3$}] {};
\end{tikzpicture}
\caption{The relational graph $\RG$.}
\label{fig:components-relational-graph}
\end{subfigure}\hfill
\begin{subfigure}[b]{0.29\textwidth}
\centering
\begin{tikzpicture}[
    scale=0.42,
    every node/.style={font=\small},
    v/.style={circle, draw=black, line width=1pt, minimum size=0.72cm, inner sep=0pt},
    blackedge/.style={draw=black, line width=1pt},
    rededge/.style={draw=red, line width=1.2pt},
    blob/.style={
        draw=black,
        dashed,
        line width=1pt,
        decorate,
        decoration={random steps,segment length=7pt,amplitude=1.6pt},
        rounded corners=8pt
    },
    bigblob/.style={
        draw=black,
        dashed,
        line width=1.4pt,
        decorate,
        decoration={random steps,segment length=8pt,amplitude=2pt},
        rounded corners=10pt
    }
]

\node[v] (x1) at (10,2.2) {$x_1$};
\node[v] (x2) at (15,2.2) {$x_2$};
\node[v] (x3) at (10,-1.7) {$x_3$};
\node[v] (x4) at (15,-1.7) {$x_4$};

\draw[blackedge] (x1) -- (x2);
\draw[blackedge] (x1) -- (x3);
\draw[rededge] (x1) -- (x4);
\draw[rededge] (x3) -- (x2);
\draw[blackedge] (x3) -- (x4);

\node[
    draw=black,
    dashed,
    line width=1pt,
    rectangle,
    fit=(x1)(x2)(x3)(x4),
    inner xsep=0.2cm,
    inner ysep=0.3cm
] {};
\end{tikzpicture}\par\vspace{3mm}
\caption{The comparison graph $\CG$.}
\label{fig:queries_graph__running_example}
\end{subfigure}\hfill
\begin{subfigure}[b]{0.32\textwidth}
\centering
\begin{tikzpicture}[
    scale=0.72,
    >=Latex,
    every node/.style={font=\small},
    v/.style={circle, draw=black, line width=0.9pt, minimum size=0.62cm, inner sep=1pt},
    k/.style={rectangle, draw=black, line width=0.9pt, minimum size=0.5cm, inner sep=2pt},
    qedge/.style={draw=black, line width=0.9pt, ->},
    tedge/.style={draw=red, line width=0.9pt, -{Latex[length=2.1mm, width=1.7mm]}}
]
\node[k] (c0) at (0,0) {$0$};
\node[v] (x) at (2.1,1.35) {$x$};
\node[v] (y) at (2.1,0) {$y$};
\node[v] (z) at (2.1,-1.35) {$z$};
\node[k] (c10) at (4.2,0) {$10$};

\draw[qedge] (c0)--(x);
\draw[qedge] (c0)--(y);
\draw[qedge] (c0)--(z);
\draw[qedge] (x)--(c10);
\draw[qedge] (y)--(c10);
\draw[qedge] (z)--(c10);
\draw[qedge] (c0) .. controls (1.1,-2.6) and (3.1,-2.6) .. (c10);
\draw[qedge] (x)--(y) node[midway,right,font=\small] {$\le$};
\draw[tedge] (z) to[bend left=28] (y);
\draw[tedge] (y) to[bend left=28] (z);
\end{tikzpicture}\par\vspace{2mm}
\caption{The opposite graph $\OG$.}
\label{fig:og-example}
\end{subfigure}
\caption{The graphs of the running examples. In (a), the red edge is a co-match edge and the separator is shaded. In (b), red edges come from comparisons induced from $q'$. In (c), unlabeled edges are labeled $<$, and the red edges are the $E_{\tox q}$ edges.}
\label{fig:three-graphs}
\end{figure}
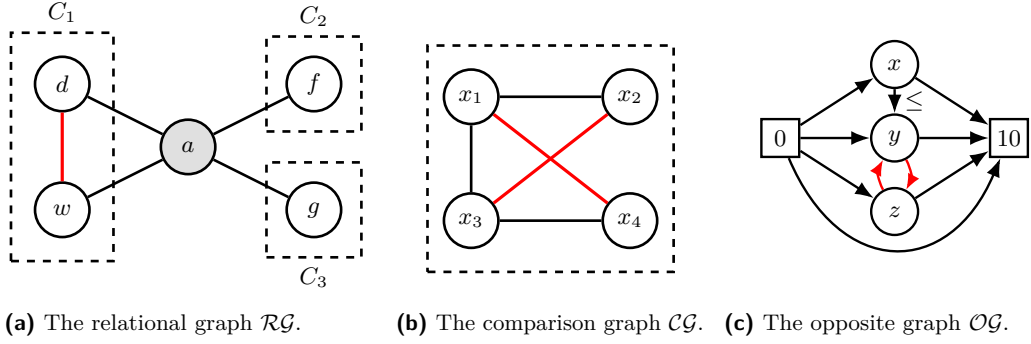

\begin{theorem}\label{thm:s-exhaustive-fixed}
Let $q(\bar x)$ and $q'(\bar x)$ be $\vcclass$ queries, let $S\subsetneq\vars(q)$ be a
\emph{legal} separator, and let $\Psi_q\subseteq\Theta_q$ be $S$-exhaustive with respect to
the components of $\RG$. If $\psi(\bar x)\in q'\bigl(\psi(D_q)\bigr)$ for all
$\psi\in\Psi_q$, then $\theta(\bar x)\in q'\bigl(\theta(D_q)\bigr)$ for all
$\theta\in\Theta_q$.
\end{theorem}

The theorem assumes nothing about the canonical values: boundary constants may be shared
and canonical assignments need not be injective.

The proof stitches per-component witnesses into one. For $\theta\in\Theta_q$,
$S$-exhaustiveness supplies for each component $C_i$ a passing assignment $\psi_i$ agreeing
with $\theta$ on $S\cup C_i$, and hence a valuation of $q'$ into $\psi_i(D_q)$. The
co-match edges confine all variables covering-matched by one atom of $q'$ to a single
component, so every atom of $q'$ is realized inside one $\psi_i(D_q)$, and legality makes
the component serving a variable that matches only the separator unambiguous.
Covering-match keeps the stitching sound without injectivity: whatever value a valuation
places at a variable $y$ of $q'$ is a witness value for $y$, so its whole cohort is
confined to the one component serving $y$.

\looseness=-1 Since $\Theta_q$ itself is $S$-exhaustive, the quantity of interest is the size of the
smallest $S$-exhaustive family, i.e., the number of canonical databases tested by the decomposition when $\Psi_q$ is chosen optimally. We record it next, with the complexity of choosing an optimal $S$.

\begin{lemma}\label{lem:s-exhaustive-size}
Let $q,q'$ be $\vcclass$ queries, let $S\subseteq\vars(q)$ be a separator, and let
$C_1,\dots,C_m$ be the connected components of $\RG-S$. The smallest $S$-exhaustive set
$\Psi_q\subseteq\Theta_q$ has size $\Bigl(\prod_{v\in S}|c(v)|\Bigr)\cdot \max_{i\in[m]}\ \prod_{u\in C_i}|c(u)|$.
\end{lemma}

\begin{example}[Running example, cont.]
For the running example, the full criterion tests $\prod_{x}|c(x)|=2^4=16$ canonical
databases, while the decomposition for $S=\{a\}$ tests only
$|c(a)|\cdot\max\bigl(|c(d)||c(w)|,\ |c(f)|,\ |c(g)|\bigr)=4$. Adding further transfer
pairs off the hub scales this gap geometrically, while the separator and largest component
stay fixed.
\end{example}

\looseness=-1 By Lemma~\ref{lem:s-exhaustive-size}, the size of the test is governed by the separator and the largest component of $\RG-S$, so we seek the legal separator that minimizes $|\Psi_q|$. This optimization is intractable.

\begin{theorem}\label{thm:separator-nphard}
Given $\vcclass$ queries $q,q'$ and an integer $t$, deciding whether $\RG$ has a legal separator whose smallest $S$-exhaustive family has size at most $t$ is \textnormal{NP}-complete, even when every variable of $q$ has exactly two canonical values.
\end{theorem}

The hardness is by reduction from \textsc{Vertex Cover}, through an intermediate graph
problem defined in the appendix. Intractability does not affect correctness: every legal
separator yields a sound and complete containment test, and only the family size depends on
the choice, so in practice one fixes a legal separator heuristically.

\subsection{Extending to Variable--Variable Comparisons}\label{sec:cmp}

We now extend our results from $\vcclass$ to $\cmpclass$, which also allows comparisons between variables.
For this class, the witness sets alone no longer determine adequate canonical values, in either direction. When $q$ contains a comparison between two variables, the witnesses of those variables may leave no values satisfying it. When only $q'$ contains one, the constraint it induces between \textcolor{black}{the variables of $q$ that can realize its sides} is invisible to the witness sets, and the test passes although containment fails. Example~\ref{exm:witness_set_not_feasible} exhibits both.

To capture these additional constraints, recall from Section~\ref{sec:prelim} the variable--constant comparison atoms $\compvc{q}$ and the variable--variable comparison atoms $\compvv{q}$.

\begin{example}[Running example]\label{exm:cmp-running}
An online-banking system logs sessions and maintenance windows in $\mathrm{Session}(\mathit{login},\mathit{logout})$ and $\mathrm{Maint}(\mathit{start},\mathit{end})$, with hours measured from midnight.
\[
\begin{aligned}
q() \leftarrow {} & \mathrm{Session}(x_1,x_2),\ \mathrm{Session}(x_1,x_4),\ \mathrm{Maint}(x_3,x_2),\ \mathrm{Maint}(x_3,x_4),
\\ &x_1 \leq x_2,\ x_1 \leq x_3,\ x_3 < x_4,\ 0<x_1,x_2<10,\ 0<x_3<20,\ 10<x_4<20 \\
q'() \leftarrow {} & \mathrm{Session}(y_1,y_2),\ \mathrm{Maint}(y_3,y_2),\ y_1 \leq y_2,\ y_2 > y_3,\ 0<y_1,y_2<20
\end{aligned}
\]
\looseness=-1 Here $q$ asks for two sessions that log in together at hour $x_1$, and two maintenance windows that start together at hour $x_3$ and end at the two logout hours. Its variable--variable atoms are $\compvv{q} = \set{x_1 \leq x_2,\ x_1 \leq x_3,\  x_3 < x_4}$, and $\compvv{q'} = \set{y_1 \leq y_2,\ y_2 > y_3}$. The effective domains in $q$ are $\feasible{x_1}=\feasible{x_2}=(0,10)$, $\feasible{x_3}=(0,20)$, and $\feasible{x_4}=(10,20)$, and the match sets are $\tox{x_1}=\set{y_1}$, $\tox{x_2}=\tox{x_4}=\set{y_2}$, and $\tox{x_3}=\set{y_3}$.
\end{example}

We now focus on pairs of variables in $q$ whose relative order may affect containment with respect to $q'$. Some pairs arise directly from variable--variable comparison atoms in $q$. Others are induced by variable--variable comparison atoms in $q'$ \textcolor{black}{under a relaxed form of matching}. \textcolor{black}{A valuation of $q'$ over a canonical database may send $y$ to a value shared by several variables of $q$, so we say that $y\in\vars(q')$ {\em partially matches\/} $x\in\vars(q)$ if $\POS{q'}{y}\cap\POS{q}{x}\neq\emptyset$ and $\feasible y\cap\feasible x\neq\emptyset$. Formally, $\fullmcomp{q}{q'}$ contains $x \comp x'$ whenever some atom $y \comp y'$ of $\compvv{q'}$ has $y$ partially matching $x$ and $y'$ partially matching $x'$, and $\set{x\comp x'}\cup \compvv q$ is satisfiable.} When $q$ and $q'$ are clear from the context, we simply write $\mcomp$. The complete set of relevant inequalities is then $\fullrelcomp q {q'} = \compvv{q} \cup \fullmcomp{q}{q'}$.

The \emph{comparison graph} $\CG$ is the undirected graph on the vertex set $\vars(q)$ with an edge between the two variables of each inequality in $\fullrelcomp q {q'}$. We denote by $\compCG$ the set of connected components of $\CG$. We emphasize that $\CG$ is not the relational graph $\RG$ of Section~\ref{sec:vc-components}: $\RG$ groups variables whose canonical values must be varied jointly, whereas $\CG$ records the order relations among the variables, to reduce the number of orderings considered.

\begin{example}[Running example, cont.]
Since $\compvv{q'}=\set{y_1 \leq y_2,\ y_2 > y_3}$, the match sets above yield $\fullmcomp{q}{q'} = \set{x_1 \leq x_2,\ x_1\leq x_4,\ x_3 < x_2,\ x_3<x_4}$, so $\fullrelcomp{q}{q'} = \set{x_1 \leq x_2,\ x_1 \leq x_3,\  x_3 < x_4,\ x_1 \leq x_4,\ x_3 <x_2}$. The comparison graph $\CG$, shown in Figure~\ref{fig:queries_graph__running_example}, consists of a single connected component $C_1=\set{x_1,x_2,x_3,x_4}$.
\end{example}

For every nontrivial connected component $C \in \compCG$, let $\compCons{C}=\set{k_1,\ldots,k_p}$, where $k_1<\cdots<k_p$, be the set of all constants that appear in comparison atoms of $q$ involving a variable of $C$, or in comparison atoms of $q'$ involving a variable that \textcolor{black}{partially matches} some variable of $C$. These constants partition the rationals into the {\em singleton intervals} $\set{k_i}$ and the {\em open intervals} between consecutive constants, and we let $\allintervals{C} = \set{(-\infty, k_1), \set{k_1}, (k_1, k_2), \set{k_2}, \ldots, \set{k_p}, (k_p, \infty)}$ denote the set of all these intervals.

For every interval $I \in \allintervals C$, let $\actvars{C}{I} = \set{x \in C \mid I \subseteq \feasible{x}}$ be the set of {\em active} variables of $C$ on $I$, that is, the variables that may be assigned a value anywhere in $I$. Note that, by construction, for every $x \in C$ and every $I \in \allintervals{C}$, the intersection $\feasible{x} \cap I$ is either empty or equal to $I$.

Next, we select {\em representative values} $\repval{C}{I}$ for each interval $I\in\allintervals{C}$. If $I$ is an open interval, let $\ell = |\actvars{C}{I}|$. We choose values $\alpha_1^{C,I}<\alpha_2^{C,I}<\ldots<\alpha_\ell^{C,I}$ in $I$, and set $\repval{C}{I}=\set{\alpha_1^{C,I},\ldots,\alpha_\ell^{C,I}}$. We choose these values to be distinct from all constants appearing in $q$ and $q'$ and from all canonical values chosen for other components, which is possible since $I$ contains infinitely many values. If $I=\set{k}$ is a singleton interval, then each variable in $\actvars{C}{I}$ can take only the value $k$ in this interval, and we set $\repval C I =\set{k}$.

Not every active variable may take every representative value of an open interval $I$. Think of the representatives $\alpha_1^{C,I}<\cdots<\alpha_\ell^{C,I}$ as ordered positions, occupied by the active variables in increasing order of their values, where variables with equal values share a position. If $q$ entails $x\le y$ then $y$ never lies below $x$, so fewer positions are open to $x$. Accordingly, let $\rankmax{x}{I}$ be $\ell$ minus the number of active variables $y\neq x$ on $I$ for which the comparison atoms of $q$ entail $x\le y$. Only the lowest $\rankmax{x}{I}$ representatives need be available to $x$, and we set $\repvalx{C}{I}{x}=\set{\alpha_j^{C,I}\mid 1\le j\le \rankmax{x}{I}}$. For a singleton interval $I=\set{k}$, we set $\repvalx{C}{I}{x}=\set{k}$ for every $x\in\actvars{C}{I}$.

\looseness=-1 We now define, for every variable $x \in \vars(q)$, its set of canonical values, denoted by $c(x)$. If $x$ is an isolated vertex, then $c(x)$ is defined as in the construction for queries in the class $\vcclass$, namely $c(x)=\set{c(x,E) : E\in \witset_x}$.
 If $x$ belongs to a connected component $C \in \compCG$ with more than one vertex, then we collect, over the intervals in which $x$ is active, the representative values consistent with its rank,
$c(x)=\bigcup_{\substack{I \in \allintervals{C} \\ x \in \actvars{C}{I}}} \repvalx{C}{I}{x}$.
Whenever a canonical value is chosen arbitrarily, either from an infinite witness subdomain or as a representative of an open interval, we choose it to be distinct from all canonical values that were already fixed, in particular from all values that come from finite witness subdomains and singleton intervals.

A {\em canonical assignment} $\theta$ is a mapping from the variables of $q$ to canonical values such that $\theta(x) \in c(x)$ for every $x \in \vars(q)$, and $\theta$ satisfies all comparison atoms of $q$. We denote by $\Theta_q$ the set of all canonical assignments. For each $\theta \in \Theta_q$, let $D_\theta := \theta(D_q)$, and define $\mathcal{D}(\Theta_q) = \set{D_\theta \mid \theta \in \Theta_q}$ to be the set of all canonical databases.

\begin{example}[Running example, cont.]
Here $\compCons{C_1}=\set{0,10,20}$. The active sets are $\actvars{C_1}{(0,10)}=\set{x_1,x_2,x_3}$, $\actvars{C_1}{\set{10}}=\set{x_3}$, and $\actvars{C_1}{(10,20)}=\set{x_3,x_4}$, with representatives $\set{2,6,8}$, $\set{10}$, and $\set{12,17}$, respectively. Rank trimming trims two variables: $q$ entails $x_1\le x_2$ and $x_1\le x_3$, so $\rankmax{x_1}{(0,10)}=3-2=1$ and $x_1$ keeps only $\set{2}$, and the atom $x_3<x_4$ leaves $x_3$ only $\set{12}$ in $(10,20)$. 
We have 
$c(x_1) = \set{2}$, $c(x_2) = \set{2,6,8}$, $c(x_3) = \set{2,6,8,10,12}$, $c(x_4) = \set{12,17}$.
Every canonical assignment fixes $x_1\mapsto 2$, so $\Theta_q$ consists of $|c(x_2)|\cdot|\set{(x_3,x_4)\in c(x_3)\times c(x_4)\mid x_3<x_4}| = 3\cdot 9 = 27$ assignments.
\end{example}

\begin{theorem}\label{thm:cmp-criterion}
Let $q(\bar{x}),q'(\bar{x}) $ be $\cmpclass$ queries. Then
$q$ is contained in $q'$ if and only if $q(D)\subseteq q'(D)$ for all $D\in \mathcal D(\Theta_q)$.
\end{theorem}

\looseness=-1 The proof parallels that of Theorem~\ref{thm:vc-criterion}, with the canonical assignment chosen componentwise. What is new is the relation it must meet: within each connected component of $\CG$, $\theta$ reproduces the order and equalities of a satisfying valuation $\nu$ and keeps every variable in the same interval, while isolated variables are treated as before. A witness for $q'$ over $\theta(D_q)$ then transports back, the variable--variable atoms holding because $\theta$ orders each component exactly as $\nu$ does.

In~\cite{klug1988conjunctive}, containment is checked over canonical databases for all possible orderings of the variables of $q$ together with all constants of the two queries. Theorem~\ref{thm:cmp-criterion} reduces this construction twice over: only the orderings \emph{within} each connected component of $\CG$ are considered, and each component is ordered only against the constants relevant to it.

\subsection{Splitting on Order Conflicts}\label{sec:addable}

\looseness=-1 The test of Theorem~\ref{thm:cmp-criterion} still enumerates, within each
connected component of $\CG$, every ordering of the variables consistent with $q$. Most of
them carry no information. The test only probes whether the order of $q$ can be arranged
so as to \emph{violate} a comparison induced from $q'$. This section isolates the pairs on
which the two queries disagree. An order that the induced comparisons do not contest may
be imposed on $q$ outright. What resists is a \emph{cyclic} conflict. Splitting on the
three possible orders of the pair it relates settles such a conflict, and each case pins
the whole component to a single assignment.

Comparison atoms may be added independently in each component of $\CG$, so we fix one
component $C$ throughout. All the order information relevant to $C$ fits into a single
directed graph, defined as follows. The
{\em opposite graph} $\OG=(V,E_{q}\cup E_{\tox{q}})$ has the vertex set
$V=\vars(C)\cup\compCons{C}$. Each of its edges points upward in the order and is labeled
$<$ or $\le$ by the strictness of the relation it records. The set $E_q$ records the order
that $q$ forces. It has a {\em constant edge} between consecutive constants of
$\compCons{C}$, an {\em endpoint edge} between each $x\in\vars(C)$ and each finite endpoint
of $\feasible{x}$, and an {\em atom edge} for every comparison atom $u\comp v$ of $q$.
Every canonical assignment satisfies all of $E_q$.

The set $E_{\tox q}$ records what the test is hunting for. A comparison atom
$y_1\comp y_2$ of $q'$ with \textcolor{black}{$y_1$ partially matching $x_1$ and $y_2$ partially matching $x_2$}, where
$x_1,x_2\in\vars(C)$, induces the relation $x_1\comp x_2$. If the comparison atoms and constants of $q$
already decide this relation, it receives the same verdict
in every canonical database, and no edge is added. Otherwise we add the {\em reverse edge} $x_2\to x_1$, labeled with the operator
$\altcomp$ opposite in strictness to $\comp$. Satisfying that edge is exactly violating
the induced relation. Induced comparisons against a constant of $\compCons{C}$ are handled
identically, with the constant as an endpoint.

An order that no reverse edge contests may be imposed on $q$ for free. Call $x_1<x_2$ a
{\em candidate comparison} if the comparison atoms of $q$ imply no strict order between
$x_1$ and $x_2$, so that the pair is unordered or ordered only non-strictly. Call it
{\em addable} if, in addition, $\OG$ has no path from $x_2$ to $x_1$. The path condition
keeps the addition from forcing further relations by transitivity.

\begin{example}[Running example]\label{exm:og-running}
Consider the hourly events $\mathrm{Login}$, $\mathrm{Backup}$ and $\mathrm{Audit}$, and
the pair
\[
\begin{aligned}
q() \leftarrow\ & \mathrm{Login}(x),\ \mathrm{Backup}(y),\ \mathrm{Audit}(z), && 0<x,y,z<10,\ x\le y,\\
q'() \leftarrow\ & \mathrm{Backup}(y'),\ \mathrm{Audit}(z'),\ \mathrm{Backup}(y''),\
\mathrm{Audit}(z''), && y'\le z',\ z''\le y''.
\end{aligned}
\]
\looseness=-1 The atoms of $q'$ induce $y\le z$ and $z\le y$, neither decided by $q$, so $\CG$ has the
single nontrivial component $C=\set{x,y,z}$, with $\compCons{C}=\set{0,10}$. Each induced
comparison contributes a reverse edge, labeled $<$
(Figure~\ref{fig:og-example}), and the two form the only directed cycle of $\OG$. The pair
$x<z$ is a candidate comparison, since $q$ implies no order between $x$ and $z$, and it is
addable: $\OG$ has no path from $z$ to $x$.
\end{example}

\begin{theorem} \label{thm:addable-edge-containment}
Let $x_1<x_2$ be an addable comparison, where $x_1,x_2 \in \vars(C)$, and let $q^*$ be
obtained from $q$ by adding it to the body of $q$. Then,
$q \subseteq q' \iff q^* \subseteq q'$.
\end{theorem}

The nontrivial direction reorders a canonical database witnessing $q\not\subseteq q'$
along a topological order of $\OG$, so that it witnesses $q^*\not\subseteq q'$ as well.
Only the absence of a reverse path is used, so the theorem also covers
variable--constant candidate comparisons, with the constant fixed.

Applied repeatedly, the theorem orders more and more pairs. What it cannot order is a
pair on a directed cycle of $\OG$. In
either direction, the cycle supplies the forbidden reverse path. A cycle inside $E_q$
merely reflects an equality that $q$ already forces. The genuine obstruction is a cycle
through a reverse edge, where a comparison induced from $q'$ conflicts cyclically with the
order constraints of $q$. We call the reverse edges lying on a directed cycle of $\OG$ the
{\em cycle reverse edges}, and write $\cedge\subseteq E_{\tox q}$ for their set. We split
$q$ on them.

\looseness=-1 Each cycle reverse edge $e=(u\to v)$ admits three cases, namely $u<v$, $u=v$ and $v<u$. A
choice $\sigma\in\set{<,=,>}^{\cedge}$ picks one case for every $e\in\cedge$, and the
{\em trichotomy query} $q_\sigma$ adds the picked cases to $q$. It thereby fixes the full
order relation between the endpoints of every cycle reverse edge. We call the restrictions
to $\vars(C)$ of the canonical assignments of $q$ the {\em local assignments} of $C$, and
write $\Theta_q\!\!\upharpoonright_C$ for their set.
The test of Theorem~\ref{thm:cmp-criterion} enumerates these inside $C$.

\begin{theorem}\label{thm:trichotomy}
$q\subseteq q'$ if and only if $q_\sigma\subseteq q'$ for every trichotomy query
$q_\sigma$. Moreover, each $q_\sigma$ is tested over at most one local assignment.
\end{theorem}

\looseness=-1 The case split holds because every satisfying valuation of $q$ realizes exactly one of the
three cases of each cycle reverse edge. For the single-assignment claim, merging the
variables that $q_\sigma$ equates leaves an opposite graph with no reverse edge on a
cycle, so the theorem above strictly orders every pair, and rank trimming then leaves
each variable a single representative. So $C$ contributes at most $3^{|\cedge|}$ local
assignments, and fewer when some $q_\sigma$ is unsatisfiable.

\looseness=-1 The trichotomy family fixes the order of every cycle reverse edge separately, although
breaking each cycle once already removes the obstruction. Call a directed cycle of $\OG$
{\em non-strict} when every edge on it is labeled $\le$, and suppose $\OG$ has none.
Then the minimal sets of reverse edges that break all cycles suffice. A set $E\subseteq\cedge$ is a {\em match feedback edge set} ($\mfac$) if the graph
$(V,E_q\cup(E_{\tox q}\setminus E))$ is acyclic, and a {\em minimal} $\mfac$ if no proper
subset of $E$ is one. For each minimal $\mfac$ $E_i$ of $\OG$, the {\em feedback query}
$q_i$ adds to $q$ the relation $u\comp v$ asserted by every cycle reverse edge
$u\overset{\comp}{\to}v$ outside $E_i$, and the negation $v\altcomp u$ of every
such edge in $E_i$.
\begin{example}[Running example, cont.]\label{exm:feedback-running}
\looseness=-1 Both reverse edges of Example~\ref{exm:og-running} lie on its directed cycle, so
$\cedge=\set{y\to z,\ z\to y}$ and the minimal $\mfac$s are $E_1=\set{y\to z}$ and $E_2=\set{z\to y}$.
For $E_1$, the kept edge $z\to y$ contributes $z<y$ and the removed edge its negation
$z\le y$, so $q_1$ is equivalent to $q\wedge z<y$. Symmetrically, $q_2$ is equivalent to
$q\wedge y<z$. With the representatives $\set{2,6,8}$, the test of
Theorem~\ref{thm:cmp-criterion} runs over the $15$ local assignments of
$\Theta_q\!\!\upharpoonright_C$, the trichotomy split over $3^2=9$ cases, and each feedback
query over a single one. Since $C$ holds every variable of $q$, two canonical databases
decide the containment.
\end{example}

\begin{theorem}\label{thm:feedback-query-containment}
Assume that $\OG$ has no non-strict cycle. Then $q\subseteq q'$ if and only if $q_i\subseteq q'$ for every minimal $\mfac$ $E_i$ of $\OG$. Moreover, these feedback queries, each together with its local assignment, can be generated with polynomial delay in the sizes of $q$ and $q'$.
\end{theorem}

The forward direction is immediate. For the
converse, the cycle reverse edges violated by a witnessing assignment contain a minimal
$\mfac$ $E_i$, and rebuilding along a topological order of the rest, as in
Theorem~\ref{thm:addable-edge-containment}, witnesses $q_i\not\subseteq q'$. The opposite
graph of each pair $(q_i,q')$ is acyclic, so saturation pins one local assignment, and the
minimal $\mfac$s are exactly the minimal feedback vertex sets of a digraph on $\cedge$,
enumerable with polynomial delay~\cite{schwikowski2002enumerating}.

The assumption cannot be dropped. Take $q() \leftarrow R(x,y)$, $R(y,x)$ and
$q'() \leftarrow R(u,v)$, $R(u',v')$, $u<v$, $v'<u'$, whose opposite graph is the
non-strict cycle $x\to y\to x$. Containment fails exactly on the
canonical databases with $\theta(x)=\theta(y)$, so the two feedback queries, $q\wedge y<x$
and $q\wedge x<y$, both miss it and are contained in $q'$. The trichotomy split still
decides the pair, through the case sending both cycle edges to $=$.

\section{Combining Nulls and Comparisons}\label{sec:combined}

\looseness=-1 A query in $\cmpcombclass$ has comparison atoms and is evaluated under the three-valued
semantics of Section~\ref{sec:prelim}, over databases that may contain $\perp$. Its
subclass $\combclass$ allows only variable--constant comparisons. The two features meet in
exactly one place. A comparison with a $\perp$ argument is unsatisfied, so a variable
occurring in a comparison atom is never nulled. Nulling and ordering therefore act on
disjoint variables, and the constructions of Sections~\ref{NullCQ} and~\ref{IneCQ} compose
instead of interfering.

\looseness=-1 For a query $q$, let $\compvars(q)$ be the variables of its comparison atoms, and let
$\nonnull(q)=\join(q)\cup\compvars(q)$ be its variables that cannot take the value
$\perp$. Let $V_{\mathrm{CJ}}$ consist of the covered variables $v$ of
Section~\ref{NullCQ} with $\POS q v=\POS{q'}{w'}$ for some $w'\in\nonnull(q')$, and set
\[
V_t=V_{\mathrm{CJ}}\setminus\nonnull(q),\quad
V_n=\vars(q)\setminus\bigl(\nonnull(q)\cup V_{\mathrm C}\cup\head(q)\bigr),\quad
V_f=\vars(q)\setminus(V_t\cup V_n).
\]
\looseness=-1 This is the partition of Section~\ref{NullCQ} with $\nonnull$ in place of $\join$, in both
queries.
A compared variable is non-null for the same reason a join variable is, so it plays the same role on each side.
On the side of $q'$ it forces a covered variable of $q$ to be toggled, and on the side of
$q$ it is frozen.

\looseness=-1 A \emph{combined assignment} $\theta$ sends each $x$ to a value in $c(x)$ if $x\in V_f$,
to $\perp$ if $x\in V_n$, and to either if $x\in V_t$. Its non-$\perp$ part satisfies the
comparison atoms of $q$. The set of variables that $\theta$ sends to $\perp$ is its
\emph{null pattern}, and the \emph{combined canonical family} is
$\Dcomb(q,q')=\set{\theta(D_q)\mid\theta\text{ combined}}$.

\begin{theorem}[Combined characterization]\label{thm:comb-criterion}
For $q,q'\in\cmpcombclass$, $q\subseteq_\perp q'$ if and only if $q(D)\subseteq
q'(D)$ for every $D\in\Dcomb(q,q')$.
\end{theorem}

\looseness=-1 The proof follows those of Theorems~\ref{thm:null-tight} and~\ref{thm:cmp-criterion}. A
combined assignment mimics a satisfying valuation of $q$ in both features at once, and a
witness for $q'$ then transports back. What is new is that the two halves must agree on
which variables may be nulled, and $\nonnull(q')$ in place of $\join(q')$ is what makes
them agree. A variable of $\nonnull(q')$ landing on a covered variable that is frozen
rather than toggled would place it in $V_{\mathrm{CJ}}$, contradicting its membership in
$V_{\mathrm{CN}}$, and this one argument settles join and compared variables alike.

\enlargethispage{2\baselineskip}

\looseness=-1 Let the \emph{value count} $\valcnt(u)$ be the number of values a combined assignment may
give $u$, namely $1$, $|c(u)|$, or $|c(u)|+1$ according as $u\in V_n$, $u\in V_f$, or
$u\in V_t$. For the trivial components of $\CG$ these choices are independent, so
they contribute the product of the counts. Inside a nontrivial component they are not. Its
variables draw their values from the representatives of Section~\ref{sec:cmp} and must
realize a common ordering, so the component contributes its orderings rather than a
product. The trichotomy split of Section~\ref{sec:addable} removes the exception. For every nontrivial component $C$ we
build its opposite graph on the vertices outside $V_n$, write $\cedge(C)$ for its cycle
reverse edges, and let $\cedge(q)=\bigcup_C\cedge(C)$. Nulling a vertex removes every cycle
through it, so a null pattern can only shrink the case split, and each case still forces
$C$ to a single local assignment. Within a case a variable of a nontrivial component
therefore has a single value, or two if it lies in $V_t$, and we use these counts for such
variables and the counts above for all others.

\looseness=-1 Two points remain about the decomposition of Section~\ref{sec:vc-components}. First, we
treat $\perp$ as a canonical value of every variable of $V_n\cup V_t$, so that the
covering-matches account for it, and we read $\RG$ accordingly. Second, $\RG$ is now too
coarse. A component of $\CG$ is pinned as a whole, so its variables can no longer be
varied independently, and a variable--variable comparison atom of $q'$ constrains the
values realized for its two sides together. We therefore obtain $\RG^{+}$ from $\RG$ by
adding a clique on every nontrivial component of $\CG$ and, for each variable--variable
comparison atom $y\comp y'$ of $q'$, a clique on $\cmctox y\cup\cmctox{y'}$. Deleting a
legal separator of $\RG^{+}$ leaves each of these sets inside a single piece, as the
stitching of Theorem~\ref{thm:s-exhaustive-fixed} requires. On $\combclass$ every component
of $\CG$ is trivial, so $\cedge(q)=\emptyset$ and $\RG^{+}=\RG$. Write
$\Delta=\max_{x}|\tox x|$ for the largest match set. The width of $\RG^{+}$ is
$w=\min_S(|S|+\max_i|C_i|)$, over its legal separators $S$.

\begin{theorem}[Combined multiplicity]\label{thm:comb-full-mult}
Let $q,q'\in\cmpcombclass$. For a legal separator $S$ of
$\RG^{+}$ with components $C_1,\dots,C_m$ of $\RG^{+}-S$, containment
$q\subseteq_\perp q'$ is decided by
\[
3^{\,|\cedge(q)|}\cdot\Bigl(\prod_{v\in S}\valcnt(v)\Bigr)\cdot
\max_{i\in[m]}\prod_{u\in C_i}\valcnt(u)
\ \le\
3^{\,|\cedge(q)|}\,(2\Delta+2)^{\,|S|+\max_i|C_i|}
\]
canonical databases. The size of the family is therefore fixed-parameter
tractable in $(\Delta,w,|\cedge(q)|)$, and on $\combclass$ in $(\Delta,w)$.
\end{theorem}

\looseness=-1 Replacing a component by its forced local assignment leaves the null pattern untouched, so
the covering-match decomposition applies with the value counts above. The numeric bound
holds because $\valcnt(u)\le 2|\tox u|+2\le 2\Delta+2$ by
Proposition~\ref{prop:witness_set_size}.

\begin{corollary}\label{cor:fpt-full}
If $\Delta$, $w$ and $|\cedge(q)|$ are constant, containment is in
\textnormal{NP}, and in \textnormal{P} when $q'$ admits polynomial evaluation.
If no opposite graph has a non-strict cycle, then $3^{\,|\cedge(q)|}$
improves to $\prod_C\ell_C$, where $\ell_C$ is the number of minimal
$\mfac$s of the opposite graph of $C$.
\end{corollary}

\looseness=-1 All three parameters are local. The largest match set $\Delta$ measures how many variables
of $q'$ compete for a single variable of $q$, the width $w$ measures the separator
structure of $\RG^{+}$, and $|\cedge(q)|$ counts the comparisons induced from $q'$ that conflict
cyclically with the order of $q$. None of them grows with the size of the queries on its
own.
Finding a legal separator that
minimizes $w$ is \textnormal{NP}-complete even for $\RG$
(Theorem~\ref{thm:separator-nphard}), and $\cedge(q)$ is empty whenever the induced comparisons
create no cyclic conflict. The bound accordingly specializes to that of
Lemma~\ref{lem:s-exhaustive-size} in the absence of nulls and of variable--variable
comparisons, and to the bound $2^{|V_t|}$ of Theorem~\ref{thm:null-tight} in the absence
of comparisons. Over null-free databases the feedback queries and their local assignments
are moreover generated with polynomial delay
(Theorem~\ref{thm:feedback-query-containment}). A worked example, computing the three
parameters and the resulting family, is given in Example~\ref{exm:combined-worked}.

\section{Related Work}\label{sec:related}

\looseness=-1 {\em Containment and equivalence\/} of conjunctive queries were characterized by Chandra and
Merlin~\cite{chandra1977optimal} through containment mappings, and the theory was
extended to relational expressions with union and difference by Aho, Sagiv, and
Ullman~\cite{aho1979equivalences} and by Sagiv and
Yannakakis~\cite{sagiv1980equivalences}. Containment of CQs is NP-complete, with
tractable fragments including acyclic queries~\cite{yannakakis1981algorithms}, restricted
fanout~\cite{saraiya1991subtree,sagiv1992minimizing} and bounded
treewidth~\cite{chekuri2000conjunctive}, and Barcel\'o et al.~\cite{barcelo2014does} study
when tractability of evaluation transfers to containment.

\looseness=-1 Klug~\cite{klug1988conjunctive} showed that the homomorphism criterion does not extend
to {\em conjunctive queries with inequalities}, although it remains complete for
semi-interval queries, and gave a $\Pi_2^p$ upper bound for containment through the
canonical databases of all orderings of the variables. Van der
Meyden~\cite{van1992complexity} established the matching lower bound. Kolaitis, Martin, and
Thakur~\cite{kolaitis1998complexity} located the boundary between the NP-complete and
$\Pi_2^p$-complete cases of inequality predicates. Afrati, Li, and Mitra~\cite{afrati2004containment} identify
further classes of queries with arithmetic comparisons for which homomorphism-based
tests remain complete. Afrati and Damigos~\cite{afrati2025semiinterval}
chart the semi-interval case, separating the classes in which a single containment mapping
certifies containment from those that remain $\Pi_2^p$-complete, and bounding the number
of mappings required. Their axis is the size of the \emph{mapping} disjunction. Ours is
orthogonal, counting \emph{canonical databases}, of which the homomorphism-complete cases
are the single-database endpoint.

\looseness=-1 Farr\'e et al.~\cite{farre2007containment} showed that {\em null-containment\/} is
$\Pi_2^P$-complete in general and NP-complete for Boolean queries, and characterized it
through the null versions of the canonical database (Theorem~\ref{thm:null-farre}).
They also give a sufficient but not necessary condition, a homomorphism
mapping every join variable to a join variable or a constant, which is itself NP-complete
to decide.

\looseness=-1 On the practical side,
{\em SQL equivalence checkers\/} are one-sided. Symbolic provers certify equivalence on
restricted syntactic fragments and return no
counterexamples~\cite{zhou2019automated,zhou2022spes,ding2023proving,wang2024qed}. Bounded checkers
search for a counterexample database up to a size bound and certify nothing beyond
it~\cite{chu2017cosette,he2024verieql,zhao2025polygon,klopfenstein2026spotit}, and
test-generation tools guarantee only coverage of a criterion or a mutation
space~\cite{chandra2015data,chen2025parseval,miao2019explaining}. We instead delimit families of databases
whose passing certifies containment outright.

\section{Conclusion}\label{sec:conclusion}
\enlargethispage{2\baselineskip}

\looseness=-1 This paper asked how few test databases suffice to decide containment, for two classes in
which one database does not: CQs over databases with nulls, and CQs with order
comparisons. For nulls, only the toggled variables need vary, giving $2^{|V_t|}$ test
databases in place of $2^{|\njoin(q)|}$. For comparisons, canonical values from witness
sets decide containment, and the family shrinks by decomposition and by splitting along
the cyclic order conflicts. Combined, the number of test databases is fixed-parameter
tractable in three local parameters, small for many query pairs.

\looseness=-1 Several directions remain open. The nearest is to admit disequality atoms ($s\neq s'$),
whose effective domains are unions of intervals rather than intervals. Further out lie
other semantics and richer languages. Under bag semantics two CQs are equivalent exactly
when they are isomorphic, while containment remains
open, and slight extensions already make it
undecidable~\cite{chaudhuri1993optimization,jayram2006containment,marcinkowski2024bag}, so it is not clear
what a canonical-database test should even be. Aggregation, unions and negation are beyond our constructions as well. Finally, our
families are generated in no particular order, while a checker hoping to
refute~\cite{he2024verieql,zhao2025polygon} wants a distinguishing database early.

\bibliography{references}
\newpage
\appendix
\section{Use of Artificial Intelligence}

\looseness=-1 We disclose the following uses of Anthropic’s Claude Code. During the exploratory stages of the research, we used the tool to search for potential counterexamples to tentative hypotheses and to perform preliminary sanity checks.

\looseness=-1 All proof strategies, constructions, mathematical arguments, and case analyses presented in the paper were developed by the authors. We additionally used Claude Code to draft portions of the appendices from detailed author-prepared outlines. The resulting prose was thoroughly reviewed and revised by the authors, and all mathematical statements, proofs, and citations in the final paper were independently checked and validated. The authors bear full responsibility for the content of the paper.

\section{Full Proofs}\label{app:proofs}

\subsection{Proofs for Section~\ref{NullCQ}}\label{app:proofs-nulls}

\begin{proof}[Proof of \cref{thm:null-tight}]
Recall that $\nullclass$ queries contain only relational atoms, so a satisfying valuation of
$q'$ over a database $D$ is a map $h\colon\vars(q')\to\adom(D)$, where $\adom(D)$ is the
set of values occurring in the tuples of $D$, that places every relational
atom of $q'$ into $D$ and sends every join variable of $q'$ to a non-$\perp$ value. We use two
elementary facts about position sets. First, every non-join variable of $q$ has a single occurrence, 
hence a singleton position
set. Second, being \emph{covered}, and lying in $V_{\mathrm{CJ}}$ or $V_{\mathrm{CN}}$, depends
only on a variable's position set: if $v,v'\in\njoin(q)$ satisfy $\POS{q}{v}=\POS{q}{v'}$, then
$v$ is covered (respectively, in $V_{\mathrm{CJ}}$, in $V_{\mathrm{CN}}$) iff $v'$ is, since each
of these conditions is phrased entirely in terms of the position set. 

\looseness=-1 \emph{The matching homomorphism.}
Let $M\subseteq\njoin(q)$, and suppose $q'$ returns $\theta_M\bar x$ over $\theta_M D_q$,
witnessed by a valuation $h$. Every value of $\adom(\theta_M D_q)$ is either $\perp$ or a frozen
variable $u\notin M$ (identified with its constant), and every tuple of $\theta_M D_q$ is the
$\theta_M$-image of a relational atom of $q$. Hence for each atom $R(\bar y)$ of $q'$ we may fix
an atom $R(\bar u)$ of $q$ with $h(\bar y)=\theta_M(\bar u)$. Reading off, for each occurrence of
a variable $y$ of $q'$ at a position $p$, the variable of the matched $q$-atom at position $p$,
defines a map $\phi\colon\vars(q')\to\vars(q)$. This is well defined. A non-join variable of $q'$
has a single occurrence. If $y$ is a join variable then $h(y)\neq\perp$ is a frozen variable
$w\notin M$, so $\theta_M$ maps the matched variable at \emph{every} occurrence of $y$ to the
constant $w$, forcing that variable to be $w$. Thus $\phi(y)=w\notin M$ for join $y$, independently
of the occurrence chosen.

By construction $\phi$ maps each atom $R(\bar y)$ of $q'$ to the atom $R(\phi(\bar y))$ of $q$,
i.e.\ it is a homomorphism from the relational body of $q'$ to that of $q$, and it satisfies
\begin{enumerate}
  \item[(h1)] $\phi(u)\notin M$ for every join variable $u$ of $q'$;
  \item[(h2)] $\theta_M\bigl(\phi(x'_i)\bigr)=\theta_M(x_i)$ for every head position $i$,
\end{enumerate}
where (h2) restates $h(\bar x')=\theta_M\bar x$ via $\theta_M\circ\phi=h$. Conversely, for any
$N\subseteq\njoin(q)$, a homomorphism $\phi$ from the relational body of $q'$ to that of $q$
satisfying (h1) and (h2) with $M:=N$ gives the valuation $\theta_N\circ\phi$, which witnesses that
$q'$ returns $\theta_N\bar x$ over $\theta_N D_q$: relational atoms are satisfied because $\phi$ is
a homomorphism, each join variable $u$ has $\theta_N(\phi(u))\neq\perp$ by (h1), and the output is
$\theta_N\bar x$ by (h2). We use this criterion in both directions.

For the forward direction, if $q\subseteq_\perp q'$, then by \cref{thm:null-farre}, $q'$ returns the head of $q$ over every null version of $D_q$. Since each member of the family is a null version $\theta_{V\cup V_n}D_q$ with $V\cup V_n\subseteq\njoin(q)$, $q'$ returns the head of $q$ over every member of the family.

For the converse, assume that $q'$ returns $\theta_{V\cup V_n}\bar x$ over $\theta_{V\cup V_n}D_q$ for every $V\subseteq V_t$. By \cref{thm:null-farre}, it suffices to fix an arbitrary $N\subseteq\njoin(q)$ and show that $q'$ returns $\theta_N\bar x$ over $\theta_ND_q$. We will use the identity $V_n=\vars(q)\setminus\bigl(\join(q)\cup V_{\mathrm C}\cup\head(q)\bigr)$,
which follows immediately from $V_f\cup V_t=\join(q)\cup V_{\mathrm C}\cup\head(q)$ and the definitions of $V_f,V_t,$ and $V_n$.

We first observe that $\POS{q'}{x'_i}\subseteq\POS{q}{x_i}$ for every head position $i$. Indeed, take $V=\emptyset$, so the corresponding member of the family is $\theta_{V_n}D_q$. Since $\head(q)\cap V_n=\emptyset$, every head variable is frozen. Let $\psi$ be a matching homomorphism witnessing that $q'$ returns over this database. By (h2), $\theta_{V_n}(\psi(x'_i))=\theta_{V_n}(x_i)=x_i$,
hence $\psi(x'_i)=x_i$. Since $\psi$ is a homomorphism, every position of $x'_i$ in $q'$ is therefore a position of $x_i$ in $q$.

Now set $V:=V_t\cap N$ and $M:=V\cup V_n$. By hypothesis, $q'$ returns $\theta_M\bar x$ over $\theta_MD_q$. Let $\phi$ be a matching homomorphism witnessing this, so $\phi$ satisfies (h1) and (h2) for $M$. We show that the same $\phi$ satisfies both conditions for $N$.
\emph{Condition (h2).}
Fix a head position $i$. Since $x_i\in\head(q)$, we have $x_i\notin V_n$, and hence $x_i\in M$ iff $x_i\in V$.

\begin{itemize}
    \item If $x_i\notin V$, then $x_i\notin M$, so $\theta_M(x_i)=x_i$. By (h2) for $M$, $\phi(x'_i)=x_i$, and therefore $\theta_N(\phi(x'_i))=\theta_N(x_i)$.
    \item \looseness=-1 If $x_i\in V=V_t\cap N$, then $\theta_M(x_i)=\perp$, so (h2) for $M$ gives $\phi(x'_i)\in M$. We claim that $\phi(x'_i)\in N$. Since $M=V\cup V_n$ and $V\subseteq N$, it remains only to rule out $\phi(x'_i)\in V_n$.

    Suppose $e:=\phi(x'_i)\in V_n$. Then $e\in\njoin(q)$, so $\POS{q}{e}$ is a singleton, and since $\phi$ is a homomorphism,
    $\POS{q'}{x'_i}\subseteq\POS{q}{e}$.
    
  \looseness=-1 We also have $\POS{q'}{x'_i}\subseteq\POS{q}{x_i}$ by the observation above, while $\POS{q}{x_i}$ is a singleton because $x_i\in N\subseteq\njoin(q)$. Since $\POS{q'}{x'_i}$ is nonempty, $ \POS{q}{e}=\POS{q'}{x'_i}=\POS{q}{x_i}$.
    But $x_i\in V_t\subseteq V_{\mathrm C}$ is covered, so $e$ is covered as well, contradicting $e\in V_n$. Thus $\phi(x'_i)\in V\subseteq N$, and hence $\theta_N(\phi(x'_i))=\perp=\theta_N(x_i)$.
    
\end{itemize}

\emph{Condition (h1).}
\looseness=-1 Let $u$ be a join variable of $q'$. By (h1) for $M$,
$\phi(u)\notin M=V\cup V_n$.
Suppose towards a contradiction that $\phi(u)\in N$. Since $N\subseteq\njoin(q)$, $\POS{q}{\phi(u)}$ is a singleton, say $\{(R,c)\}$. As $\phi$ is a homomorphism and $u$ occurs in $q'$,
$\POS{q'}{u}\subseteq\POS{q}{\phi(u)}=\{(R,c)\}$,
and hence $\POS{q'}{u}=\POS{q}{\phi(u)}$.

\begin{itemize}
    \item If $\phi(u)$ is covered, then $\POS{q}{\phi(u)}=\POS{q'}{x''}$ for some $x''\in\head(q')$. Since $u$ is a join variable and $\POS{q'}{u}=\POS{q}{\phi(u)}$, the variable $\phi(u)$ satisfies the defining condition of $V_{\mathrm{CJ}}$. Thus
    $\phi(u)\in V_{\mathrm{CJ}}=V_t$,
    and therefore $\phi(u)\in V_t\cap N=V\subseteq M$, a contradiction.

    \item If $\phi(u)$ is not covered, then $\phi(u)\notin V_{\mathrm C}$. Moreover, $\phi(u)\notin\join(q)$ because $\phi(u)\in N\subseteq\njoin(q)$, and $\phi(u)\notin V_n$ because $\phi(u)\notin M$. By the identity above, $\phi(u)\in\head(q)$, so we write $\phi(u)=x_i$. The observation above then gives $\POS{q'}{x'_i}\subseteq\POS{q}{x_i}=\set{(R,c)}$.

    Since $x'_i$ occurs in $q'$, equality follows. Hence $x_i$ is covered by the head variable $x'_i$ of $q'$, contradicting the assumption that $\phi(u)=x_i$ is not covered.
\end{itemize}
Thus $\phi(u)\notin N$, so (h1) also holds for $N$.

Therefore (h1) and (h2) hold for $N$, so $\theta_N\circ\phi$ witnesses that $q'$ returns
$\theta_N\bar x$ over $\theta_N D_q$. As $N\subseteq\njoin(q)$ was arbitrary, \cref{thm:null-farre}
yields $q\subseteq_\perp q'$.
\end{proof}

\subsection{Proofs for Section~\ref{sec:matching-witness}}\label{app:proofs-matching-witness}

\looseness=-1 The following example, referenced in Section~\ref{sec:matching-witness}, illustrates the witness-set construction on a larger instance, where the maximally infinite-contradictable sets can be read directly from the diagram of the effective domains in Figure~\ref{fig:witness_interval}.

\begin{example}\label{exm:witness-large}
Consider the queries $q$ and $q'$ over $\mathbb{Q}$:
\[
\begin{aligned}
q() &\leftarrow  R(x_1),\ x_1 \ge 0,\ x_1 \le 13 \\
q'() &\leftarrow  R(y_1),\ R(y_2),\ R(y_3),\ R(y_4),\ R(y_5), \\
& y_1 \le 3,\ y_2 \ge 8,\ y_3 \le 9,\ y_3 \ge 2,\ y_4 \le 6,\ y_5 \le 11,\ y_5 \ge 5
\end{aligned}
\]

We compute the witness set $\mathcal{W}_{x_1}$. The effective domain of $x_1$ is $\fullfeasible{q}{x_1} = [0, 13]$. For each $y_i \in \text{match}(x_1)$, the incompatibility set $E_i = \fullfeasible{q}{x_1} \setminus \fullfeasible{q'}{y_i}$ is as follows:
\[
\begin{aligned}
E_1 = (3, 13], \quad E_2 = [0, 8), \quad E_3 = [0, 2) \cup (9, 13], \\
E_4 = (6, 13], \quad E_5 = [0, 5) \cup (11, 13]
\end{aligned}
\]

\looseness=-1 Figure~\ref{fig:witness_interval} illustrates the interaction between these domains. This illustration shows how the set of incompatible variables changes along the effective domain of $x_1$ as we cross the endpoints of each $\fullfeasible{q'}{y_i}$.
The labels $-y_i$ and $+y_i$ show when a variable's interval starts or ends within the domain of $x_1$. Specifically, $-y_i$ marks the point where $y_i$'s interval begins, meaning it is no longer incompatible with $x1$. Conversely, $+y_i$ marks where $y_i$'s interval ends, making it incompatible from that point.

From Figure~\ref{fig:witness_interval}, we can see that there are four maximally infinite-contradictable sets of variables. These sets, highlighted in red boxes, define the following witnesses in $\mathcal{W}_{x_1}$:
\begin{enumerate}
    \item $S_1 = \{y_2, y_3, y_5\}$ giving the witness $E_1 = [0,2)$.
    \item $S_2 = \{y_1, y_2, y_5\}$ giving the witness $E_2 = (3,5)$.
    \item $S_3 = \{y_1, y_2, y_4\}$ giving the witness $E_3 = (6,8)$.
    \item $S_4 = \{y_1, y_3, y_4, y_5\}$ giving the witness $E_4 = (11,13]$.
\end{enumerate}

Here all comparisons are non-strict, so no finite witness arises and $\finvals{x_1}=\emptyset$. As a result, $\mathcal{W}_{x_1} = \{[0,2), (3,5), (6,8), (11,13]\}$. Note that each witness represents a region where no single $y_i$ from the corresponding maximal set can satisfy the query $q$.
\end{example}

\begin{figure}[t]
\centering
\begin{tikzpicture}[
  x=1.0cm, y=0.78cm,
  interval/.style={black, line width=0.8pt},
  axis/.style={black, line width=1.0pt},
  vline/.style={blue, dashed, line width=0.8pt},
  settxt/.style={font=\scriptsize, fill=white, inner sep=1.5pt},
  eventtxt/.style={font=\scriptsize},
  ylab/.style={font=\small, anchor=west, inner sep=2pt},
  hi/.style={draw=red, line width=0.8pt, fill=white, inner sep=1.5pt}
]

\foreach \x in {0,2,3,5,6,8,9,11,13} {
  \draw[vline] (\x,-1.75) -- (\x,-0.5);
  \draw[vline] (\x,0.05) -- (\x,3.7);
}

\draw[interval] (0,0.7) -- (3,0.7);
\node[ylab] at (3.05,0.7) {$y_1$};

\draw[interval] (8,0.7) -- (13,0.7);
\node[ylab] at (13.05,0.7) {$y_2$};

\draw[interval] (2,1.4) -- (9,1.4);
\node[ylab] at (9.05,1.4) {$y_3$};

\draw[interval] (0,2.1) -- (6,2.1);
\node[ylab] at (6.05,2.1) {$y_4$};

\draw[interval] (5,2.8) -- (11,2.8);
\node[ylab] at (11.05,2.8) {$y_5$};

\draw[axis] (0,0) -- (13,0);

\foreach \x/\lab in {0/0,2/2,3/3,5/5,6/6,8/8,9/9,11/11,13/13} {
  \node[anchor=north, inner sep=2pt, fill=white] at (\x,-0.08) {\scriptsize \lab};
}

\node[settxt,hi] at (1,-1.05) {$\left\{y_2,y_3,y_5\right\}$};
\node[settxt]    at (2.5,-1.05) {$\left\{y_2,y_5\right\}$};
\node[settxt,hi] at (4,-1.05) {$\left\{y_1,y_2,y_5\right\}$};
\node[settxt]    at (5.5,-1.05) {$\left\{y_1,y_2\right\}$};
\node[settxt,hi] at (7.0,-1.05) {$\left\{y_1,y_2,y_4\right\}$};
\node[settxt]    at (8.5,-1.05) {$\left\{y_1,y_4\right\}$};
\node[settxt]    at (10.0,-1.05) {$\left\{y_1,y_3,y_4\right\}$};
\node[settxt,hi] at (12.0,-1.05) {$\left\{y_1,y_3,y_4,y_5\right\}$};

\node[eventtxt, text=blue, fill=white, inner sep=1.5pt] at (2,-2){$-y_3$};
\node[eventtxt, text=blue, fill=white, inner sep=1.5pt] at (3,-2){$+y_1$};
\node[eventtxt, text=blue, fill=white, inner sep=1.5pt] at (5,-2){$-y_5$};
\node[eventtxt, text=blue, fill=white, inner sep=1.5pt] at (6,-2){$+y_4$};
\node[eventtxt, text=blue, fill=white, inner sep=1.5pt] at (8,-2){$-y_2$};
\node[eventtxt, text=blue, fill=white, inner sep=1.5pt] at (9,-2){$+y_3$};
\node[eventtxt, text=blue, fill=white, inner sep=1.5pt] at (11,-2){$+y_5$};

\end{tikzpicture}
\caption{Contradictable sets.}\label{fig:witness_interval}
\end{figure}
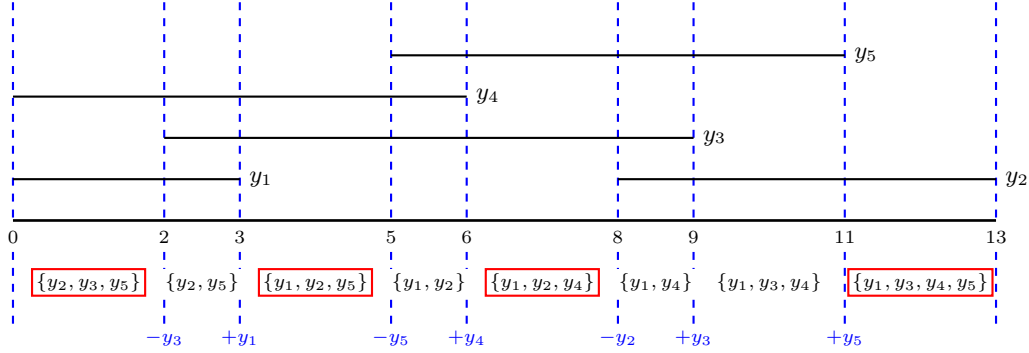

\begin{proof}[Proof of \cref{prop:witness_set_size}]
Since the comparison operators are $<$ and $\le$, the effective domain $\feasible x$ and
every $\fullfeasible{q'}{y}$ with $y\in\tox x$ are intervals. For a point $p\in\feasible x$, let $\varout(p)=\set{y\in\tox x\mid p\notin\fullfeasible{q'}{y}}$ be
the set of variables in $\tox x$ whose effective domain does not contain $p$. Grouping the points of $\feasible x$ by the set $\varout(p)$ splits $\feasible x$ into maximal segments, on which $\varout(p)$ does not change. We call each such segment a \emph{cell}.
Each cell is an interval or a single point. We first bound the number of cells, and then show that every
witness in $\witset_x$ picks out a cell of its own.

\smallskip\noindent \looseness=-1 There are at most $2|\tox x|+1$ cells.
As $p$ moves through $\feasible x$, the set $\varout(p)$ can change only when $p$ 
crosses an endpoint of some $\fullfeasible{q'}{y}$. 
Each $y\in\tox x$ has just two endpoints and can therefore change $\varout$ at most twice, once when $p$ enters $\fullfeasible{q'}{y}$ and once when it leaves. This gives at most $2|\tox x|$ changes in total. The cells are the maximal segments on which $\varout(p)$ does not change, so their number is one more than the number of changes, that is, at most $2|\tox x|+1$.

\smallskip\noindent\looseness=-1\emph{Singleton witnesses use point cells.}
Let $\set c$ be a witness with $c\in \finvals x$, so $c\in E_x(S)$ for some $S$ with $E_x(S)$ finite. Then $S\subseteq\varout(c)$, so $E_x(\varout(c))\subseteq E_x(S)$ is finite and still contains $c$. The cell of $c$ lies inside the finite set $E_x(\varout(c))$, which forces it to be the single point $\set c$. Therefore, distinct values $c$ occupy distinct point cells.

\smallskip\noindent\emph{Infinite witnesses use interval cells.}
Let $E_x(S)$ be a witness with $S$ maximally infinite-contradictable. As $E_x(S)$ is infinite, it contains a nondegenerate interval, and hence an interval cell $C$, on which $\varout(p)$ is the same set $T\supseteq S$. Were $T$ strictly larger than $S$, the inclusion $C\subseteq E_x(T)$ would make $T$ infinite-contradictable and larger than $S$, contradicting maximality. So $\varout$ equals $S$ on $C$, and distinct sets $S$ pick out distinct interval cells.

\smallskip
Point cells and interval cells are never the same cell, so the singleton and infinite witnesses are matched to disjoint families of cells. Therefore
\[
|\witset_x|\ \le\ \#\text{cells}\ \le\ 2|\tox{x}|+1 . \qedhere
\]
\end{proof}

\begin{proof}[Proof of \cref{thm:vc-criterion}]
Clearly, if $q$ is contained in $q'$, then  $q(D)\subseteq q'(D)$ for all $D\in \mathcal D(\Theta_q)$. We prove the other direction. 

\looseness=-1 Assume that $q(D)\subseteq q'(D)$ for all $D\in \mathcal D(\Theta_q)$.
Let $D$ be any database instance. We assume that $q(D) \neq \emptyset$, as otherwise containment is trivial. 
Therefore, there is a valuation $\nu: \vars(q) \to \adom(D)$ and a tuple $\bar{a} \in q(D)$ such that $\nu(\bar{x}) = \bar{a}$ and  $D\models \nu(\body(q))$. 

\looseness=-1 We construct a valuation $\nu'$ for $q'$ that also produces $\bar a$ over $D$. To find this valuation, we first identify a canonical database that maps the variables in $q$ to constants in a way that is similar to (and only more restrictive than) the valuation of $q$ over $D$. 
\looseness=-1 For each \(x \in \vars(q)\), consider the set \(\varout(\nu(x))\) of variables in \(\tox x\) whose effective domains do not contain \(\nu(x)\), and the subdomain \(E_x(\varout(\nu(x)))\) induced by this set. Note that this subdomain is non-empty, since \(\nu(x) \in \feasible{x}\).

We now choose, for each \(x\in\vars(q)\), a witness \(F_x\in\witset_x\), and use it to define a canonical assignment \(\theta\in\Theta_q\). Recall \(\nu(x)\in E_x(\varout(\nu(x)))\).

\smallskip
\noindent\emph{If \(E_x(\varout(\nu(x)))\) is finite:} since \(\nu(x)\in E_x(\varout(\nu(x)))\), the definition of \(\finvals x\) gives \(\nu(x)\in \finvals x\), so \(\{\nu(x)\}\in\witset_x\). We set \(F_x:=\{\nu(x)\}\) and \(\theta(x):=\nu(x)\), which is a canonical value since \(c(x,\{\nu(x)\})=\{\nu(x)\}\).

\smallskip
\noindent\emph{If \(E_x(\varout(\nu(x)))\) is infinite:} then \(\varout(\nu(x))\) is infinite-contradictable, so we may extend it to a maximally infinite-contradictable set \(S_x\supseteq \varout(\nu(x))\). We set \(F_x:=E_x(S_x)\), which is an infinite witness in \(\witset_x\), and \(\theta(x):=c(x,F_x)\), the constant chosen for $F_x$.

\smallskip
In both cases \(\theta(x)\in F_x\), and \(F_x\subseteq E_x(\varout(\nu(x)))\) (in the infinite case because \(\varout(\nu(x))\subseteq S_x\)). Hence \(\theta(x)\in E_x(\varout(\nu(x)))\).

Now, consider the canonical database $D_\theta = \theta(D_q)$. 
Since $\theta \in \Theta_q$, the assignment $\theta$ satisfies all atoms of $q$ over $D_\theta$, so $\theta(\bar x) \in q(D_\theta)$, and by assumption $\theta(\bar x) \in q'(D_\theta)$. Therefore, there exists a valuation
\(\theta' : \vars(q') \to \adom(D_\theta)\) such that $\theta'(\bar x) = \theta(\bar x)$ and $D_\theta \models \theta'(\body(q'))$. Observe that we have a mapping $\theta'$ from $q'$ to $D_\theta$, a mapping $\theta$ from $q$ to $D_\theta$ and a mapping $\nu$ from $q$ to $D$ (see Figure~\ref{fig:mappings}). If $\theta$ was invertible, composing these mapping would immediately yield a mapping from $q'$ to $D$ as is required. Unfortunately, $\theta$ is not invertible, as multiple variables in $q$ may be mapped to the same constant by $\theta$. However, as we show next, we can overcome this problem as all such variables in $q$ will also be mapped to the same value in $D$.

We define the mapping \(\nu':\vars(q')\to\adom(D)\) as follows. For each \(y\in\vars(q')\), choose some \(x\in\vars(q)\) with \(\theta'(y)=\theta(x)\), and set \(\nu'(y):=\nu(x)\). Such an \(x\) exists because \(q\) is in normalized form, so every value occurring in \(D_\theta\) is the \(\theta\)-image of a variable of \(q\).

\looseness=-1 We check that \(\nu'\) is well-defined, i.e., independent of the chosen \(x\). Suppose \(\theta(x_1)=\theta(x_2)\). If this common value was chosen from an infinite witness, then by construction it is distinct from all other canonical values, so \(x_1=x_2\). Otherwise it is a singleton value, and then \(\theta(x_i)=\nu(x_i)\) for \(i=1,2\), so \(\nu(x_1)=\theta(x_1)=\theta(x_2)=\nu(x_2)\). In either case the value \(\nu'(y)\) does not depend on the chosen preimage.

Now, it remains to show that $D\models \nu'(\body(q'))$. First, consider a relational atom
$R(\bar{y})$ in the body of $q'$. We must show that this atom is satisfied by $\nu'$ in $D$.
Since $\theta'$ satisfies $q'$ in $D_\theta$, we have $R(\theta'(\bar{y})) \in D_\theta$.
By the construction of $D_\theta$,  there exists an atom $R(\bar{z})$ in the body of $q$ such that
\(\theta(\bar{z}) = \theta'(\bar{y})\).
Moreover, since \(D \models \nu(\body(q))\), we get \(R(\nu(\bar{z})) \in D\).
By the definition of \(\nu'\) and well-definedness, we have \(\nu'(\bar{y}) = \nu(\bar{z})\), and therefore \(R(\nu'(\bar{y})) \in D\).
Hence, the relational atoms of \(q'\) are satisfied in \(D\) under \(\nu'\).

\looseness=-1 It remains to show that $\nu'$ satisfies all comparison atoms of $q'$. Let $y\in \vars(q')$, and let $x$ be the variable used in defining $\nu'(y)$, so $\theta'(y)=\theta(x)$ and $\nu'(y)=\nu(x)$. We show that $\nu'(y)\in \fullfeasible{q'}{y}$.

If $F_x=\{\nu(x)\}$, then $\nu'(y)=\nu(x)=\theta(x)=\theta'(y)$. Since $\theta'$ satisfies the comparison atoms of $q'$, we have $\theta'(y)\in \fullfeasible{q'}{y}$, and hence $\nu'(y)\in \fullfeasible{q'}{y}$.

Now suppose $F_x$ is infinite, so $\theta(x)$ is the constant chosen for  $F_x$. We first show that $y\in\tox{x}$. Since $\theta'(y)=\theta(x)\in\fullfeasible{q'}{y}$ (as $\theta'$ is a valuation) and $\theta(x)\in F_x\subseteq\feasible{x}$, we have $\feasible{x}\cap\fullfeasible{q'}{y}\neq\emptyset$. For the positional condition, let $(R,j)\in\POS{q'}{y}$. Then $y$ occurs in position $j$ of some atom $R(\bar y)$ of $q'$, so $R(\theta'(\bar y))\in D_\theta$, and the value in position $j$ is $\theta'(y)=\theta(x)$. Since $\theta(x)$ is the canonical constant assigned to the infinite witness $F_x$, it is distinct from every other canonical value. Thus, this tuple can only have been produced from an atom of $q$ in which $x$ occurs in position $j$. Hence $(R,j)\in\POS{q}{x}$.

Assume, towards a contradiction, that $\nu'(y)=\nu(x)\notin\fullfeasible{q'}{y}$. Since $y\in\tox{x}$, the definition of $\varout(\nu(x))$ gives $y\in \varout(\nu(x))$, whence $E_x(\varout(\nu(x)))\subseteq\feasible{x}\setminus\fullfeasible{q'}{y}$. As $F_x\subseteq E_x(\varout(\nu(x)))$, it follows that $F_x\cap\fullfeasible{q'}{y}=\emptyset$. But $\theta(x)\in F_x$ and $\theta(x)=\theta'(y)\in\fullfeasible{q'}{y}$, a contradiction. Hence $\nu'(y)\in\fullfeasible{q'}{y}$.

Therefore $D\models\nu'(\body(q'))$. Finally, $\theta'(\bar{x})=\theta(\bar{x})$ gives $\nu'(\bar{x})=\nu(\bar{x})=\bar a$, so $\bar a\in q'(D)$. Since $D$ and $\bar a\in q(D)$ were arbitrary, $q\subseteq q'$, as required.

\end{proof}
\begin{figure}[h]
\centering
\begin{tikzpicture}[
  font=\small,
  node distance=1.3cm and 2.6cm,
  every node/.style={inner sep=2pt},
  map/.style={-{Latex[length=4pt,width=3.5pt]}, line width=0.5pt, shorten <=1pt, shorten >=1pt},
  lab/.style={font=\small, inner sep=2.5pt},
]
\node (Qp) {$\vars(q')$};
\node (Dt) [right=of Qp] {$\adom(D_\theta)$};
\node (D)  [below=of Qp] {$\adom(D)$};
\node (Q)  [below=of Dt] {$\vars(q)$};

\draw[map]                (Qp) -- node[lab, above] {$\theta'$} (Dt);
\draw[map]                (Q)  -- node[lab, right] {$\theta$}  (Dt);
\draw[map]                (Q)  -- node[lab, below] {$\nu$}     (D);
\draw[map, densely dashed] (Qp) -- node[lab, left]  {$\nu'$}    (D);

\node[font=\footnotesize, align=center] at ($(Qp)!0.5!(Q)$) {$\theta'(y)=\theta(x)$\\$\Longrightarrow\ \nu'(y)=\nu(x)$};
\end{tikzpicture}
\caption{The mappings in the proof of Theorem~\ref{thm:vc-criterion}. Given a valuation
$\theta'$ of $q'$ over the canonical database $D_\theta$, we use $\theta$ and $\nu$ to build
a valuation $\nu'$ of $q'$ over the original database $D$.}
\label{fig:mappings}
\end{figure}
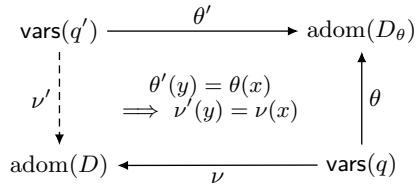

\subsection{Proofs for Section~\ref{sec:vc-components}}\label{app:proofs-vc-components}

\looseness=-1 The proof of \cref{thm:s-exhaustive-fixed} rests on two localization lemmas. We will use
the following observation, immediate from the definition of witness values: a witness
value pulls in its whole cohort.

\begin{observation}\label{obs:cohort}
If $a$ is a witness value for $y$, then $\Vval a\subseteq\cmctox y$.
\end{observation}
 Throughout,
recall that legality means $|\cmctox{y}\cap S|\le 1$ for every $y\in\vars(q')$.

\looseness=-1 For $y\in\vars(q')$ set $I(y)=\set{i\in[m]\mid \cmctox{y}\cap C_i\neq\emptyset}$ and, for a
relational atom $A$ of $q'$, $I_A=\bigcup_{y\in A}I(y)$. Thus $I(y)$ records the components
from which $y$ may take its value, and a match lying only in the separator contributes no
index. When $I_A$ is a single component, the atom $A$ is realized within it, and when $I_A$
has several, $A$ straddles them and only the separator links its images.

\begin{lemma}[Atom locality]\label{lem:atom-locality}
For every relational atom $A$ of $q'$, $|I_A|\le 1$. In particular $|I(y)|\le 1$ for every
$y\in\vars(q')$, and $\cmctox{y}\setminus S$ is contained in a single component.
\end{lemma}
\begin{proof}
\looseness=-1 Fix a relational atom $A=R(y_1,\dots,y_k)$ of $q'$, and recall that
$M_A=\bigcup_{t=1}^{k}\cmctox{y_t}$ is the set of all variables of $q$ that are
covering-matched by some variable occurring in $A$. By the co-match edges, every two
distinct variables of $M_A$ are adjacent in $\RG$, so $M_A$ is a clique.

Forming $\RG-S$ deletes exactly the separator vertices, so the vertices removed from this
clique are those in $M_A\cap S$. Any two surviving variables of $M_A\setminus S$ are still
joined by their co-match edge, because deleting a vertex removes only the edges incident to
it. Hence $M_A\setminus S$ remains a clique and lies in a single connected component of
$\RG-S$.

Since $\bigcup_{y\in A}(\cmctox{y}\setminus S)=M_A\setminus S$ sits in that one component,
the variables of $A$ reach at most one component, so $|I_A|\le 1$. Applying this to any atom
$A$ that contains a fixed $y\in\vars(q')$ gives $\cmctox{y}\subseteq M_A$, so
$\cmctox{y}\setminus S$ lies in a single component and $|I(y)|\le 1$. Note that
$I(y)=\emptyset$ if every variable covering-matched by $y$ lies in $S$.
\end{proof}

\begin{lemma}[Realizer localization]\label{lem:realizer}
Let $\theta\in\Theta_q$ and let $\nu$ be a valuation with
$\theta(D_q)\models\nu(\body(q'))$. Let $A=R(y_1,\dots,y_k)$ be a relational atom of $q'$
and let $R(z_1,\dots,z_k)$ be any relational atom of $q$ with $\theta(z_t)=\nu(y_t)$ for all
$t$ (one exists because $R(\nu(\bar y))\in\theta(D_q)$). Then the value $a_t:=\nu(y_t)$ is a
witness value for $y_t$ and $z_t\in\Vval{a_t}\subseteq\cmctox{y_t}$ for every $t$.
Consequently, if $I_A=\set{i}$ then $z_1,\dots,z_k\in S\cup C_i$, and if $I_A=\emptyset$ then
$z_1,\dots,z_k\in S$.
\end{lemma}
\begin{proof}
Fix $t$ and write $a=\nu(y_t)=\theta(z_t)$. As $a=\theta(z_t)$ is a canonical value of
$z_t$, we have $z_t\in\Vval a$. We check that $a$ is a witness value for $y_t$.

First, $a=\nu(y_t)\in\feasible{y_t}$, because $\nu$ satisfies the comparison atoms of $y_t$.

\looseness=-1 Second, let $(R',p')\in\POS{q'}{y_t}$ be arbitrary, say $y_t$ occurs at position $p'$ in an
atom $A'=R'(\dots)$ of $q'$. Since $\nu$ satisfies $A'$, the tuple $R'(\nu(\dots))$ lies in
$\theta(D_q)$, so there is an atom $R'(w_1,\dots)$ of $q$ with $\theta(w_{p'})=\nu(y_t)=a$.
Hence $w_{p'}\in\Vval a$ and $(R',p')\in\POS q{w_{p'}}$. As $(R',p')$ was arbitrary,
$\bigcup_{x\in\Vval a}\POS q x\supseteq\POS{q'}{y_t}$.

Thus $a$ is a witness value for $y_t$, and by Observation~\ref{obs:cohort}
$\Vval a\subseteq\cmctox{y_t}$, in particular $z_t\in\cmctox{y_t}$. Finally, by
Lemma~\ref{lem:atom-locality} we have $\cmctox{y_t}\setminus S\subseteq C_i$ when $I_A=\set
i$ (and $\cmctox{y_t}\subseteq S$ when $I_A=\emptyset$, since then $\cmctox{y_t}$ meets no
component). Hence $z_t\in S\cup C_i$ (resp.\ $z_t\in S$). The argument uses no injectivity:
the realizer $z_t$ is localized directly through the witness value $a$, regardless of how
many variables of $q$ also carry $a$.
\end{proof}

\begin{proof}[Proof of \cref{thm:s-exhaustive-fixed}]
\looseness=-1 Let $\theta\in\Theta_q$ be arbitrary. We construct a valuation witnessing
$\theta(\bar x)\in q'(\theta(D_q))$. Let $C_1,\dots,C_m$ be the components of $\RG-S$.
Since $\Psi_q$ is $S$-exhaustive, for each $i\in[m]$ pick $\psi_i\in\Psi_q$ with
$\psi_i|_{S\cup C_i}=\theta|_{S\cup C_i}$. By hypothesis each $\psi_i$ admits a valuation
$\mu_i:\vars(q')\to\adom(\psi_i(D_q))$ with $\psi_i(D_q)\models\mu_i(\body(q'))$ and
$\mu_i(\bar x)=\psi_i(\bar x)$. Lemma~\ref{lem:realizer} applies to each pair
$(\psi_i,\mu_i)$. The plan is to build $\mu$ from the $\mu_i$, taking the value of each
$y\in\vars(q')$ from the index $i$ in $I(y)$.

\looseness=-1 \emph{Step 1: variables with $I(y)=\emptyset$ get the same value from every $\mu_i$.}
Let $y\in\vars(q')$ with $I(y)=\emptyset$, and let $A$ be a relational atom of $q'$
containing $y$. By Lemma~\ref{lem:atom-locality}, $\cmctox{y}\subseteq S$. The set
$\cmctox{y}$ is not empty, since Lemma~\ref{lem:realizer} applied to $(\psi_1,\mu_1)$ and
$A$ places a variable of $\cmctox{y}$ at the position of $y$. Since $S$ is a legal separator, $|\cmctox{y}\cap S|\le 1$. Hence $\cmctox{y}=\set{s}$ for a single $s\in S$. Now fix any $i$.
Lemma~\ref{lem:realizer} applied to $(\psi_i,\mu_i)$ and $A$ places a variable
$z\in\cmctox{y}=\set{s}$ at the position of $y$, so $\mu_i(y)=\psi_i(s)=\theta(s)$, as
$\psi_i$ and $\theta$ agree on $S$. This value does not depend on $i$.

\emph{Step 2: the assembled valuation.} By Lemma~\ref{lem:atom-locality}, $|I(y)|\le 1$
for every $y\in\vars(q')$. Define $\mu(y):=\mu_i(y)$, where $i$ is the unique index with
$I(y)=\set i$, and where $i$ is arbitrary if $I(y)=\emptyset$. By Step~1 the choice of
$i$ does not matter in the second case. In both cases, $\mu(y)=\mu_i(y)$ for every $i$
with $I(y)\subseteq\set i$. We use this property in the remaining steps.

\emph{Step 3: relational atoms.} Let $A=R(y_1,\dots,y_k)$ be a relational atom of $q'$.
By Lemma~\ref{lem:atom-locality}, $I_A\subseteq\set i$ for some $i$. Every $y_t$
satisfies $I(y_t)\subseteq I_A\subseteq\set i$, so $\mu(y_t)=\mu_i(y_t)$ by Step~2.
Hence $R(\mu(\bar y))=R(\mu_i(\bar y))\in\psi_i(D_q)$, so there is an atom
$R(z_1,\dots,z_k)$ of $q$ with $\psi_i(z_t)=\mu_i(y_t)$ for all $t$. By
Lemma~\ref{lem:realizer}, every $z_t$ lies in $S\cup C_i$, where $\psi_i$ agrees with
$\theta$. Therefore $\mu(y_t)=\psi_i(z_t)=\theta(z_t)$ for all $t$, and
$R(\mu(\bar y))=R(\theta(\bar z))\in\theta(D_q)$.

\looseness=-1 \emph{Step 4: comparison atoms.} Since $q'\in\vcclass$, every comparison atom of $q'$
has the form $y\comp c$ with $c$ a constant. By Step~2, $\mu(y)=\mu_i(y)$ for some $i$,
and $\mu_i$ satisfies $y\comp c$. Hence $\mu(y)\comp c$ holds. This is the only place
where $q'\in\vcclass$ is used: a comparison between two variables could involve two
different indices $i$.

\looseness=-1 \emph{Step 5: head.} Let $x_j$ be a head variable, which occurs in both $q$ and $q'$.
Choose $i$ with $I(x_j)\subseteq\set i$. By Step~2 and head preservation,
$\mu(x_j)=\mu_i(x_j)=\psi_i(x_j)$. It remains to show that the variable $x_j$ of $q$
lies in $S\cup C_i$, since then $\psi_i(x_j)=\theta(x_j)$. Let $a=\psi_i(x_j)$. Then
$a\in c(x_j)$, so $x_j\in\Vval a$. Lemma~\ref{lem:realizer} applied to $(\psi_i,\mu_i)$
and an atom of $q'$ containing $x_j$ shows that $a=\mu_i(x_j)$ is a witness value for
$x_j$. By Observation~\ref{obs:cohort}, $\Vval a\subseteq\cmctox{x_j}$, so
$x_j\in\cmctox{x_j}$. By Lemma~\ref{lem:atom-locality} and $I(x_j)\subseteq\set i$,
$\cmctox{x_j}\subseteq S\cup C_i$. Hence $x_j\in S\cup C_i$ and $\mu(x_j)=\theta(x_j)$.

By Steps 3 to 5, $\mu$ satisfies $\body(q')$ over $\theta(D_q)$ and
$\mu(\bar x)=\theta(\bar x)$. Thus $\mu$ witnesses $\theta(\bar x)\in q'(\theta(D_q))$,
as required.
\end{proof}

\begin{proof}[Proof of \cref{lem:s-exhaustive-size}]
\looseness=-1 Recall that a canonical assignment chooses one canonical value for each variable. For a set
of variables $W\subseteq\vars(q)$, write $A(W)=\prod_{u\in W}c(u)$ for the set of all
assignments of canonical values to $W$, so $|A(W)|=\prod_{u\in W}|c(u)|$ and
$\Theta_q=A(\vars(q))$. Since $S,C_1,\dots,C_m$ partition $\vars(q)$, an assignment
$\theta\in\Theta_q$ is the same as the tuple of its restrictions
$(\theta|_S,\theta|_{C_1},\dots,\theta|_{C_m})\in A(S)\times A(C_1)\times\dots\times A(C_m)$,
and every such tuple arises from exactly one $\theta$.

\looseness=-1 Let $n_i=|A(C_i)|$ and $N=\max_{i\in[m]}n_i$, and fix $i^*$ with $n_{i^*}=N$. The claimed
size is $|A(S)|\cdot N$. We show that every $S$-exhaustive set has at least this size, and
then construct one that has at most this size.

\looseness=-1 \emph{Lower bound.} Let $\Psi\subseteq\Theta_q$ be $S$-exhaustive. Applying the definition
with $i=i^*$, for every $\theta\in\Theta_q$ some $\psi\in\Psi$ has
$\psi|_{S\cup C_{i^*}}=\theta|_{S\cup C_{i^*}}$. As $\theta$ ranges over $\Theta_q$, the
restriction $\theta|_{S\cup C_{i^*}}$ ranges over all of $A(S\cup C_{i^*})$. Hence the map
$\psi\mapsto\psi|_{S\cup C_{i^*}}$ from $\Psi$ to $A(S\cup C_{i^*})$ is surjective, and
$|\Psi|\ge|A(S\cup C_{i^*})|=|A(S)|\cdot n_{i^*}=|A(S)|\cdot N$.

\emph{Construction.} For each $i\in[m]$, since $n_i\le N$, we can list the elements of
$A(C_i)$ as a sequence $\gamma_i^1,\dots,\gamma_i^N$ of length $N$ in which every element
of $A(C_i)$ appears at least once. For $\sigma\in A(S)$ and $t\in[N]$, let
$\psi_{\sigma,t}\in\Theta_q$ be the assignment with $\psi_{\sigma,t}|_S=\sigma$ and
$\psi_{\sigma,t}|_{C_i}=\gamma_i^t$ for every $i\in[m]$, and let
$\Psi_q=\{\psi_{\sigma,t}\mid \sigma\in A(S),\ t\in[N]\}$.
There are $|A(S)|\cdot N$ pairs $(\sigma,t)$, so $|\Psi_q|\le|A(S)|\cdot N$.

\looseness=-1 It remains to check that $\Psi_q$ is $S$-exhaustive. Let $\theta\in\Theta_q$ and
$i\in[m]$. Let $\sigma=\theta|_S$, and choose $t\in[N]$ with $\gamma_i^t=\theta|_{C_i}$,
which exists because every element of $A(C_i)$ appears in the sequence. Then
$\psi_{\sigma,t}$ agrees with $\theta$ on $S$ and on $C_i$, that is, on $S\cup C_i$.

\looseness=-1 By the lower bound, $|\Psi_q|\ge|A(S)|\cdot N$ as well, so $|\Psi_q|=|A(S)|\cdot N$ and
this is the size of the smallest $S$-exhaustive set. Since
$|A(S)|\cdot N=\bigl(\prod_{v\in S}|c(v)|\bigr)\cdot\max_{i\in[m]}\prod_{u\in C_i}|c(u)|$,
the lemma follows.
\end{proof}

\begin{proof}[Proof of \cref{thm:separator-nphard}]
\emph{Membership in NP.} Guess $S\subseteq\vars(q)$, check that it is legal, and compare
the size given by Lemma~\ref{lem:s-exhaustive-size} with $t$. This is polynomial, since
the witness sets, and with them the counts $|c(v)|$ and the covering-match relation, can
be computed by a sweep over the endpoints of the effective domains.

\looseness=-1 \emph{Reduction to a graph problem.} When $|c(v)|=2$ for every variable $v$, the size in
Lemma~\ref{lem:s-exhaustive-size} is $2^{|S|}\cdot 2^{\max_i|C_i|}=2^{|S|+\max_i|C_i|}$.
So if every separator of $\RG$ is legal, asking for a legal separator with
$|\Psi_q|\le 2^t$ is the same as asking for a set $S$ with $|S|+\max_i|C_i|\le t$. This
is the following problem. \textsc{Min-Separator}: \emph{given an undirected graph
$G=(V,E)$ and an integer $t$, is there $S\subseteq V$ with $|S|+\max_i|C_i(G-S)|\le t$,
where $C_1,\dots,C_m$ are the connected components of $G-S$?} We show that
\textsc{Min-Separator} is NP-hard, and then that every graph is the relational graph of
a query pair in which every variable has two canonical values and every separator is
legal.

\emph{\textsc{Min-Separator} is NP-hard.} We reduce from \textsc{Vertex Cover}. Let
$(G_0=(V_0,E_0),k)$ be an instance with $n=|V_0|$. We may assume $k\ge 1$, since for
$k=0$ the answer is yes if and only if $G_0$ has no edges. We may assume $k<n$, since
otherwise $V_0$ itself is a cover of size at most $k$. Adding isolated vertices to $G_0$
changes neither its vertex covers nor their sizes, so we may also assume $n\ge 2k+1$. Let
$G$ be $G_0$ with $k$ new \emph{leaf} vertices attached to each $v\in V_0$, each adjacent
only to $v$, and let $t=2k+1$. We claim that $G_0$ has a vertex cover of size at most $k$
if and only if $G$ has a set $S$ with $|S|+\max_i|C_i(G-S)|\le t$.

\looseness=-1 $(\Rightarrow)$ Let $X$ be a vertex cover of $G_0$ with $|X|\le k$, and let $S=X$. Since
$X$ covers every edge of $G_0$, no two original vertices outside $S$ are adjacent. So
every original vertex $v\notin S$ forms a component together with its $k$ leaves, of size
$k+1$, and every leaf of a vertex in $S$ is a component by itself. As $|X|\le k<n$, some
original vertex lies outside $S$, so $\max_i|C_i|=k+1$ and
$|S|+\max_i|C_i|\le k+(k+1)=t$.

\looseness=-1 $(\Leftarrow)$ Let $S$ satisfy $|S|+\max_i|C_i(G-S)|\le t$. We may assume $S\subseteq V_0$:
removing a leaf from $S$ lowers $|S|$ by one, and the leaf either joins a component,
raising its size by one, or forms a new component of size one, so the maximum rises by at
most one and the sum does not increase. We may also assume $S\ne V_0$, since for $S=V_0$
the graph $G-S$ consists of the $nk\ge 1$ isolated leaves and
$|S|+\max_i|C_i|=n+1\ge 2k+2>t$. Now $S$ is a vertex cover of $G_0$: if some edge
$(u,v)\in E_0$ had both endpoints outside $S$, then $u$, $v$ and their $2k$ leaves would
lie in one component of size at least $2k+2>t$. Hence every original vertex outside $S$
forms a component of size exactly $k+1$, so $\max_i|C_i|=k+1$ and $|S|\le t-(k+1)=k$.
Thus $G_0$ has a vertex cover of size at most $k$.

The construction of $G$ is polynomial, so \textsc{Min-Separator} is NP-hard.

\emph{Every graph is a relational graph.} Let $G=(V,E)$ be a graph. Let $q$ have a
variable $x_v$ for each $v\in V$ and no comparison atoms, so $\feasible{x_v}=\Q$. For
each edge $(u,v)\in E$, add to $q$ the atom $R_{uv}(x_u,x_v)$ over a fresh relation
symbol $R_{uv}$. For each $v\in V$, add to $q$ the atom $U_v(x_v)$ over a fresh unary
symbol $U_v$, and add to $q'$ two variables $y_v,y'_v$ with the atoms $U_v(y_v)$,
$U_v(y'_v)$ and the comparisons $y_v<0$ and $y'_v\ge 0$.

\looseness=-1 Each of $y_v,y'_v$ occurs only in the position $(U_v,1)$, which in $q$ is occupied by
$x_v$ alone, so $\tox{x_v}=\set{y_v,y'_v}$. The sets $\set{y_v}$ and $\set{y'_v}$ are
infinite-contradictable by $x_v$, with $E_{x_v}(\set{y_v})=[0,\infty)$ and
$E_{x_v}(\set{y'_v})=(-\infty,0)$, while $E_{x_v}(\set{y_v,y'_v})=\emptyset$. So these
are the two maximally infinite-contradictable sets, no finite witness arises, and
$|c(x_v)|=2$. Both canonical values of $x_v$ are constants chosen for infinite witnesses,
so by the choice of canonical constants they differ from the constant $0$ and from every
canonical value of every other variable of $q$.

\looseness=-1 We now show $\RG=G$. The atoms $R_{uv}$ produce exactly the atom edges of $G$. Let $a$ be
a witness value for $y_v$. Its cohort $\Vval a$ must cover the only position $(U_v,1)$ of
$y_v$, which $x_v$ alone occupies, so $x_v\in\Vval a$ and $a\in c(x_v)$. By the previous
paragraph no other variable of $q$ carries $a$, so $\Vval a=\set{x_v}$. Hence
$\cmctox{y_v}=\set{x_v}$, and likewise $\cmctox{y'_v}=\set{x_v}$. Thus no variable of $q'$
covering-matches two distinct variables of $q$, so every separator is legal. Each
relational atom of $q'$ contains a single variable, whose covering-matches form a
singleton, so no atom of $q'$ covering-matches two distinct variables and there are no
co-match edges. Therefore $\RG=G$.

Combining the two reductions, deciding whether some legal separator achieves
$|\Psi_q|\le 2^t$ is NP-hard, which completes the proof.
\end{proof}

\subsection{Proofs for Section~\ref{sec:cmp}}\label{app:proofs-cmp}

The following example, referenced in Section~\ref{sec:cmp}, shows that for $\cmpclass$ queries the witness sets alone are insufficient, in either direction. In the first pair below the query $q$ contains a variable--variable comparison. In the second pair, only $q'$ does.

\begin{example}\label{exm:witness_set_not_feasible}
\emph{$q$ carries the comparison.} Let $q \in \cmpclass$ and $q' \in \vcclass$ be the Boolean queries
\[
\begin{aligned}
q() \leftarrow  & R(x_1),\ S(x_2),\ 10<x_1\leq 25,\ 10\leq x_2<25,\ x_1 > x_2\\
q'() \leftarrow & R(y_1),\ R(y_2),\ R(y_3),\ S(y_4),\ S(y_5),\ 10\leq y_1\leq 25,\\
                & 10\leq y_2\leq 15,\ 20\leq y_3\leq 25,\ 10\leq y_4\leq 25,\ 10\leq y_5\leq 20
\end{aligned}
\]
\looseness=-1 For $x_1$, the subsets induced by $y_1,y_2,y_3$ are $E_1=\emptyset$, $E_2=(15,25]$ and $E_3=(10,20)$, so $\set{y_2,y_3}$ is the only maximally infinite-contradictable set and no finite witness arises. Symmetrically, only $\set{y_5}$ is maximally infinite-contradictable by $x_2$. Hence $\witset_{x_1}=\set{(15,20)}$ and $\witset_{x_2}=\set{(20,25)}$, so canonical values taken from the witness sets cannot satisfy $x_1 > x_2$. Yet this does not mean that $q \subseteq q'$, as $D=\set{R(21),S(11)}$ gives $q(D)=\text{true}$ and $q'(D)=\text{false}$.

\begin{figure}[h]
\begin{minipage}[h]{0.47\textwidth}
\centering
\begin{tikzpicture}[
  x=0.38cm, y=0.7cm,
  interval/.style={black, line width=0.8pt},
  axis/.style={black, line width=1.0pt},
  vline/.style={blue, dashed, line width=0.8pt},
  settxt/.style={font=\scriptsize, fill=white, inner sep=1.5pt},
  eventtxt/.style={font=\scriptsize, fill=white, inner sep=1.5pt},
  ylab/.style={font=\small, anchor=west, inner sep=2pt},
  hi/.style={draw=red, line width=0.8pt, fill=white, inner sep=1.5pt}
]

\foreach \x in {10,15,20,25} {
  \draw[vline] (\x,-1.75) -- (\x,-0.5);
  \draw[vline] (\x,0.05) -- (\x,2.3);
}

\draw[interval] (10,0.7) -- (25,0.7);
\node[ylab] at (25.05,0.7) {$y_1$};

\draw[interval] (10,1.4) -- (15,1.4);
\node[ylab] at (15.05,1.4) {$y_2$};

\draw[interval] (20,1.4) -- (25,1.4);
\node[ylab] at (25.05,1.4) {$y_3$};

\draw[axis] (10,0) -- (25,0);

\foreach \x/\lab in {10/10,15/15,20/20,25/25} {
  \node[anchor=north, inner sep=2pt, fill=white] at (\x,-0.08) {\scriptsize \lab};
}
\node[ylab] at (25.05,0) {$x_1$};

\node[settxt]    at (12.5,-1.05) {$\left\{y_3\right\}$};

\node[settxt,hi]    at (17.5,-1.05) {$\left\{y_2,y_3\right\}$};

\node[settxt] at (22.5,-1.05) {$\left\{y_2\right\}$};

\node[eventtxt, text=blue] at (15,-2){$-y_2$};
\node[eventtxt, text=blue] at (20,-2){$+y_3$};

\end{tikzpicture}
\end{minipage}
\hfill
\begin{minipage}[h]{0.47\textwidth}
\centering
\begin{tikzpicture}[
  x=0.38cm, y=0.7cm,
  interval/.style={black, line width=0.8pt},
  axis/.style={black, line width=1.0pt},
  vline/.style={blue, dashed, line width=0.8pt},
  settxt/.style={font=\scriptsize, fill=white, inner sep=1.5pt},
  eventtxt/.style={font=\scriptsize, fill=white, inner sep=1.5pt},
  ylab/.style={font=\small, anchor=west, inner sep=2pt},
  hi/.style={draw=red, line width=0.8pt, fill=white, inner sep=1.5pt}
]

\foreach \x in {10,20,25} {
  \draw[vline] (\x,-1.75) -- (\x,-0.5);
  \draw[vline] (\x,0.05) -- (\x,2.3);
}

\draw[interval] (10,0.7) -- (25,0.7);
\node[ylab] at (25.05,0.7) {$y_4$};

\draw[interval] (10,1.4) -- (20,1.4);
\node[ylab] at (20.05,1.4) {$y_5$};

\draw[axis] (10,0) -- (25,0);

\foreach \x/\lab in {10/10,20/20,25/25} {
  \node[anchor=north, inner sep=2pt, fill=white] at (\x,-0.08) {\scriptsize \lab};
}
\node[ylab] at (25.05,0) {$x_2$};

\node[settxt,hi]    at (22.5,-1.05) {$\left\{y_5\right\}$};

\node[settxt]    at (15.0,-1.05) {$\emptyset$};

\node[eventtxt, text=blue] at (20,-2){$-y_5$};

\end{tikzpicture}
\end{minipage}
\caption{Witness set that does not yield feasible canonical values.}\label{fig:witness_interval_not_hold}
\end{figure}
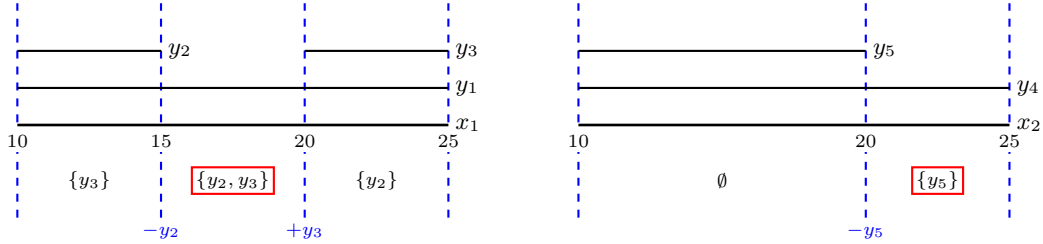
\smallskip\noindent\emph{$q'$ carries the comparison.} Let now $\hat q \in \vcclass$ be
\[
\begin{aligned}
\hat q() \leftarrow & R(x_1),\ R(x_3),\ R(x_4),\ S(x_2),\ S(x_5),\ T(x_1),\ T(x_2),\ 10<x_1\leq 25,\\
                    & 10\leq x_2<25,\ 10\leq x_3\leq 15,\ 20\leq x_4\leq 25,\ 10\leq x_5\leq 20
\end{aligned}
\]
\looseness=-1 and let $\hat q' \in \cmpclass$ be obtained from $q'$ by adding the atoms $T(y_1)$, $T(y_4)$ and the comparison $y_1<y_4$. The variables $x_3,x_4,x_5$ have no contradictable sets, so the witness set of each is its effective domain alone and each has a single canonical value, while $x_1$ and $x_2$ have the same witness sets as above, as illustrated in Figure~\ref{fig:witness_interval_not_hold}. The witness sets thus determine a single canonical assignment $\theta$, and $\hat q(D_\theta) \subseteq \hat q'(D_\theta)$. Still, $\hat q \not\subseteq \hat q'$: since $y_1$ can be matched only to $x_1$ and $y_4$ can be matched only to $x_2$, one can choose values satisfying all atoms of $\hat q$ whose images violate the comparison $y_1 < y_4$.
\end{example}

The proof of \cref{thm:cmp-criterion} rests on the componentwise-relaxation machinery, which we now define formally.

For an assignment $\mu : \vars(q)\to \mathbb{Q}$ and an isolated variable $x\in\vars(q)$, define $\forbidden {\mu}{x} = \set{y \in \tox{x} \mid \mu(x) \notin \fullfeasible{q'}{y}}$. Intuitively, $\forbidden {\mu}{x}$ contains the variables of $q'$ that can match $x$ at the relational level, but cannot be assigned the value $\mu(x)$ because this value violates their effective domains.

Let $\mu_1,\mu_2 : \vars(q)\to \mathbb{Q}$ be two assignments.
We say that $\mu_1$ is a {\em componentwise relaxation\/} of $\mu_2$ and write
$\mu_1 \ceqord \mu_2$, if the following conditions hold.
\begin{enumerate}
   \item The assignments agree on the ordering within each connected component: for every component $C$ of $\CG$ and every pair $z_1,z_2 \in C$, both $\mu_1(z_1)<\mu_1(z_2) \iff \mu_2(z_1)<\mu_2(z_2)$ and $\mu_1(z_1)=\mu_1(z_2) \iff \mu_2(z_1)=\mu_2(z_2)$.
    \item The assignments agree on the intervals: for every nontrivial component $C$ of $\CG$, variable $z \in C$, and interval $I \in \allintervals{C}$, we have $\mu_1(z) \in I \iff \mu_2(z) \in I$.
    \item For every isolated variable $x\in\vars(q)$, every variable of $q'$ forbidden under $\mu_1$ is also forbidden under $\mu_2$, that is, $\forbidden{\mu_1}{x}\subseteq \forbidden{\mu_2}{x}$.
\end{enumerate}

The relation $\ceqord$ captures exactly the information about a valuation that the containment test must preserve: the internal order and equalities within each component, the interval of each component variable, and the forbidden sets of the isolated variables. The next lemma shows that this information can always be preserved by a canonical assignment.

\begin{lemma}\label{lem:component_equality_order_eq}
For every database instance $D$ and every satisfying valuation $\mu$ of $q$ over $D$, there exists a canonical assignment $\theta \in \Theta_q$ such that $\mu \ceqord \theta$.
\end{lemma}

\begin{proof}[Proof of \cref{lem:component_equality_order_eq}]
Let $D$ be a database instance, and let $\mu$ be a satisfying valuation of $q$ over $D$.
We construct a canonical assignment $\theta \in \Theta_q$ as follows.

\looseness=-1 First, we handle isolated variables as in the proof of Theorem~\ref{thm:vc-criterion}. Let $x$ be an isolated variable, and recall that $\varout(\mu(x))$ is the set of variables in $\tox x$ whose effective domains do not contain $\mu(x)$, so that $\forbidden{\mu}{x}=\varout(\mu(x))$. Recall also that, for any set $S\subseteq\tox{x}$, $E_x(S)=\bigcap_{y\in S}\bigl(\fullfeasible{q}{x}\setminus\fullfeasible{q'}{y}\bigr)$. Since $\mu(x)\in\fullfeasible{q}{x}$ (as $\mu$ satisfies $q$) and $\mu(x)\notin\fullfeasible{q'}{y}$ for every $y\in \varout(\mu(x))$, we have $\mu(x)\in E_x(\varout(\mu(x)))$. Hence $\varout(\mu(x))$ is contradictable by $x$.

\looseness=-1 If $E_x(\varout(\mu(x)))$ is finite, then $\mu(x)\in \finvals x$, so $\set{\mu(x)}\in\witset_x$. We set $\theta(x)=\mu(x)$, which is a canonical value since $c(x,\set{\mu(x)})=\set{\mu(x)}$, and then $\forbidden{\theta}{x}=\forbidden{\mu}{x}$.

\looseness=-1 If $E_x(\varout(\mu(x)))$ is infinite, then $\varout(\mu(x))$ is infinite-contradictable, so we extend it to a maximally infinite-contradictable set $S_x\supseteq\varout(\mu(x))$ and set $\theta(x)=c(x,E_x(S_x))$, the constant chosen for the infinite witness $E_x(S_x)$. Since $\theta(x)\in E_x(S_x)\subseteq\fullfeasible{q}{x}\setminus\fullfeasible{q'}{y}$ for every $y\in S_x$, we have $S_x\subseteq\forbidden{\theta}{x}$, and therefore $\forbidden{\mu}{x}=\varout(\mu(x))\subseteq\forbidden{\theta}{x}$.
In either case, $\forbidden{\mu}{x}\subseteq\forbidden{\theta}{x}$.

\looseness=-1 For every nontrivial connected component $C \in \compCG$ and every interval
$I \in \allintervals{C}$, define $X_{C,I}=\{x\in \actvars{C}{I}\mid \mu(x)\in I\}$.Since the intervals in $\allintervals{C}$ are pairwise disjoint, for every variable $x$ there is at most one interval $I \in \allintervals{C}$ such that $x \in X_{C,I}$.
We define $\theta$ separately on each set $X_{C,I}$.

\looseness=-1  If $I=\{k\}$ is a singleton interval, set $\theta(x)=k$ for every $x\in X_{C,I}$.
If $I$ is an open interval, let $B_1,\dots,B_t$ be the equivalence classes of $X_{C,I}$
under equality of $\mu$-values, ordered so that $i<j$ implies $\mu(x)<\mu(y)$ for every
$x\in B_i$ and $y\in B_j$. Since $t\le |X_{C,I}| \le |\actvars{C}{I}|$, and $\repval{C}{I}=\{\alpha_1^{C,I},\dots,\alpha_\ell^{C,I}\}$ where $\ell=|\actvars{C}{I}|$,
we define $\theta(x)=\alpha_i^{C,I}$ for every $x\in B_i$.

\looseness=-1 Repeating this construction for every nontrivial connected component $C$ and every interval
$I\in\allintervals{C}$ defines $\theta$ on all non-isolated variables of $q$. Together with
the construction above for isolated variables, this defines $\theta$ on all variables of $q$.

\looseness=-1 We check that $\theta(x)\in c(x)$ for every non-isolated $x$. Suppose $\theta(x)=\alpha_i^{C,I}$,
so $x\in B_i$ and $i$ is the number of equality-classes of $X_{C,I}$ up to and including that of
$x$. The classes $B_1,\dots,B_{i-1}$ consist of variables of $X_{C,I}$ that are strictly smaller
than $x$ under $\mu$. No variable $y$ with $q\models x\le y$ lies in them, since $\mu$ satisfies
$q$ and hence $\mu(x)\le\mu(y)$. Thus the $i-1$ classes below $x$ are formed from variables in
$\actvars{C}{I}\setminus(\set{x}\cup\set{y\mid q\models x\le y})$, so
$i-1\le |\actvars{C}{I}|-1-|\set{y\in\actvars{C}{I}\setminus\set{x}\mid q\models x\le y}|$, that
is, $i\le\rankmax{x}{I}$. Hence $\alpha_i^{C,I}\in\repvalx{C}{I}{x}\subseteq c(x)$. (For a
singleton interval the claim is immediate, as $\repvalx{C}{I}{x}=\set{k}$.)
By construction, $\theta(x)\in c(x)$ for every $x\in\vars(q)$. Moreover, every non-isolated
variable is mapped to a value in the same interval of $\allintervals{C}$ as under $\mu$, and
within each interval both equality and strict order are preserved. Hence $\theta$ satisfies
all comparison atoms of $q$, and therefore $\theta\in\Theta_q$.

\looseness=-1 It remains to prove that $\mu \ceqord \theta$. First, let $C$ be a nontrivial connected component of $\CG$, and let $z_1,z_2\in C$.
If $\mu(z_1)$ and $\mu(z_2)$ lie in the same interval of $\allintervals{C}$, then by
construction $\mu(z_1)=\mu(z_2)\iff \theta(z_1)=\theta(z_2)$ and
$\mu(z_1)<\mu(z_2)\iff \theta(z_1)<\theta(z_2)$.

\looseness=-1 If $\mu(z_1)$ and $\mu(z_2)$ lie in two distinct intervals of $\allintervals{C}$, then,
by the interval-membership property proved above, $\theta(z_1)$ and $\theta(z_2)$ lie in
the same two intervals, respectively. Since $\allintervals{C}$ is an ordered partition of
$\mathbb{Q}$, the relative order between the two values is the same under both $\mu$ and
$\theta$, and since the intervals are disjoint, equality is impossible under both assignments.

\looseness=-1 Finally, let $x$ be an isolated variable. By the construction for isolated variables, we have
~$\forbidden{\mu}{x}\subseteq\forbidden{\theta}{x}$.
Therefore, all three conditions in the definition of $\ceqord$ hold, and hence
$\mu \ceqord \theta$.
\end{proof}

\begin{proof}[Proof of \cref{thm:cmp-criterion}]
As before, the \enquote{only if} direction is immediate, so it remains to prove the
\enquote{if} direction.

\looseness=-1 Assume that $q(\hat{D})\subseteq q'(\hat{D})$ for all canonical databases
$\hat{D} \in \mathcal D(\Theta_q)$, and let $D$ be an arbitrary database instance.
As before, we assume that $q(D)\neq\emptyset$, and fix a tuple $\bar a \in q(D)$
together with a valuation $\nu:\vars(q)\to \adom(D)$ such that
$\nu(\bar x)=\bar a$ and $D\models \nu(\body(q))$.

\looseness=-1 By Lemma~\ref{lem:component_equality_order_eq}, there exists a canonical assignment
$\theta\in\Theta_q$ such that $\nu \ceqord \theta$. We choose such a $\theta$ as constructed in the proof of the lemma.

\looseness=-1 Let $D_\theta=\theta(D_q)$. Since $\theta(\bar x)\in q(D_\theta)$ and, by assumption,
$q(D_\theta)\subseteq q'(D_\theta)$, there exists a valuation
$\theta':\vars(q')\to \adom(D_\theta)$ such that
$D_\theta\models \theta'(\body(q'))$ and $\theta'(\bar x)=\theta(\bar x)$.

\looseness=-1 We define the mapping $\nu' : \vars(q') \to D $ as follows. For $ y \in \vars(q')$, we set $\nu'(y) = \nu(x)$ for some variable $x$ such that $\theta'(y) = \theta(x)$.

\looseness=-1 There may be several variables $x$ that $\theta$ maps to the same value $c=\theta'(y)$. We check that they all have the same value $\nu(x)$, so that $\nu'(y)$ is well defined. If $c$ comes from a finite witness or is a singleton-interval value, then by construction $\theta(x)=\nu(x)=c$ for every such $x$. If $c=\alpha_i^{C,I}$ is a representative of an open interval, then every variable $x$ with $\theta(x)=c$ lies in the same equality class $B_i$ of the construction in the proof of Lemma~\ref{lem:component_equality_order_eq}, and the variables of $B_i$ share a common $\nu$-value, so $\nu(x)$ is the same for all of them. Finally, a constant chosen for an infinite witness is, by construction, distinct from all other canonical values, so it is taken by a unique variable. In every case all possible choices of $x$ give the same value $\nu(x)$, so $\nu'(y)$ is well defined.

\looseness=-1 It remains to show that $D\models \nu'(\body(q'))$.
First consider a relational atom $R(\bar y)$ of $q'$. Since
$D_\theta\models \theta'(R(\bar y))$, we have
$R(\theta'(\bar y))\in D_\theta$. By the definition of $D_\theta$, there exists a
relational atom $R(\bar x)$ in the body of $q$ such that $\theta(\bar x)=\theta'(\bar y)$. Moreover, since $D\models \nu(\body(q))$, we get
$R(\nu(\bar x))\in D$. By the definition of $\nu'$, we have $\nu'(\bar y)=\nu(\bar x)$, and therefore $R(\nu'(\bar y))\in D$. Hence all relational atoms of $q'$ are satisfied in $D$ under $\nu'$.

\looseness=-1 \textcolor{black}{For the comparison atoms we use the following observation: for every $y\in\vars(q')$ we may take the variable $x$ chosen for $\nu'(y)$ to be partially matched by $y$. Indeed, $y$ occurs in some relational atom $R(\bar y)$ of $q'$, and as above there is an atom $R(\bar x)$ of $q$ with $\theta(\bar x)=\theta'(\bar y)$. The variable $x$ of $\bar x$ at the position of $y$ occupies a position of $y$ and satisfies $\theta(x)=\theta'(y)$. Since $\theta'$ satisfies the comparison atoms of $q'$ and $\theta\in\Theta_q$, this common value lies in $\fullfeasible{q'}{y}\cap\fullfeasible{q}{x}$, so $y$ partially matches $x$. By the well-definedness shown above, $\nu'(y)=\nu(x)$ for this $x$ as well, so we assume from now on that $y$ partially matches the chosen $x$.}

\looseness=-1 Consider first a variable $y$ with an atom in $\compvc{q'}$, and let $x$ be the variable chosen for it, so that $\theta(x)=\theta'(y)\in \fullfeasible{q'}{y}$ and $\nu'(y)=\nu(x)$.

\looseness=-1 If $x$ belongs to a nontrivial connected component $C$, then, \textcolor{black}{since $y$ partially matches $x$ and $x\in C$, the constants bounding $\fullfeasible{q'}{y}$ belong to $\compCons{C}$ and hence are among those inducing $\allintervals{C}$.} Let $I_x\in\allintervals{C}$ be the unique interval such that $\nu(x)\in I_x$. Since $\nu\ceqord\theta$, we also have $\theta(x)\in I_x$, and the interval $I_x$ is either contained in $\fullfeasible{q'}{y}$ or disjoint from it. As $\theta(x)\in I_x\cap \fullfeasible{q'}{y}$, it follows that $I_x\subseteq \fullfeasible{q'}{y}$, and therefore $\nu'(y)=\nu(x)\in \fullfeasible{q'}{y}$.

\looseness=-1 \textcolor{black}{If $x$ is isolated, then $\theta(x)$ was chosen in the proof of Lemma~\ref{lem:component_equality_order_eq} either as the constant of an infinite witness or as a finite-witness value. In the second case $\theta(x)=\nu(x)$ by construction, so $\nu'(y)=\theta'(y)\in\fullfeasible{q'}{y}$. In the first case $\theta(x)$ is taken by no other variable of $q$, so, exactly as in the proof of Theorem~\ref{thm:vc-criterion}, every position of $y$ is a position of $x$ and $y\in\tox{x}$. Then $\theta(x)\in \fullfeasible{q'}{y}$ gives
$y\notin \forbidden{\theta}{x}$. Since $\nu\ceqord\theta$, we have $\forbidden{\nu}{x}\subseteq\forbidden{\theta}{x}$. Hence $y\notin \forbidden{\nu}{x}$, and therefore $\nu'(y)=\nu(x)\in \fullfeasible{q'}{y}$.}

\looseness=-1 Now let $y\comp y' \in \compvv{q'}$, and let $x,x'\in\vars(q)$ be the variables chosen
in the definition of $\nu'$ for $y$ and $y'$, so that $\theta'(y)=\theta(x)$, $\theta'(y')=\theta(x')$, $\nu'(y)=\nu(x)$ and $\nu'(y')=\nu(x')$.

\looseness=-1 Since \textcolor{black}{$y$ partially matches $x$, $y'$ partially matches $x'$}, and $y\comp y'\in\compvv{q'}$, we show that $x$ and $x'$ lie in the same connected component of $\CG$ \textcolor{black}{(if $x=x'$ there is nothing to show)}. If the induced inequality $x\comp x'$ is consistent with $\compvv{q}$, then it belongs to $\fullmcomp{q}{q'}$ and hence to $\fullrelcomp{q}{q'}$, giving an edge $\set{x,x'}$. Otherwise it contradicts $\compvv{q}$, so $\compvv{q}$ itself constrains the order of $x$ and $x'$, and they are connected through edges of $\compvv{q}$. In either case $x$ and $x'$ belong to the same nontrivial connected component $C$.
Since $\theta'$ satisfies $q'$, $\theta(x)\comp\theta(x')$ holds, and condition (1) of $\nu\ceqord\theta$ on the component $C$ transfers it to $\nu(x)\comp\nu(x')$, that is, $\nu'(y)\comp\nu'(y')$.

\looseness=-1 We have shown that $D\models \nu'(\body(q'))$. Moreover, since $\theta'(\bar x)=\theta(\bar x)$, the definition of $\nu'$ gives
$\nu'(\bar x)=\nu(\bar x)=\bar a$. Thus $\bar a\in q'(D)$. Since both $D$ and $\bar a\in q(D)$ were arbitrary, $q(D)\subseteq q'(D)$ for every database instance $D$. Therefore, $q\subseteq q'$.
\end{proof}

\looseness=-1 The proof of \cref{thm:cmp-criterion} yields the following stronger fact, used by the proofs of Section~\ref{sec:addable}: when containment fails, it fails on a tuple produced by a canonical assignment itself.

\begin{corollary}\label{cor:canonical-violating}
Let $q(\bar{x}),q'(\bar{x})$ be $\cmpclass$ queries. If $q \not\subseteq q'$, then
there exists a canonical assignment $\theta\in\Theta_q$ such that
$\theta(\bar x)\in q(D_\theta)$ and $\theta(\bar x)\notin q'(D_\theta)$,
where $D_\theta=\theta(D_q)$.
\end{corollary}

\begin{proof}[Proof of \cref{cor:canonical-violating}]
Assume $q \not\subseteq q'$, and fix a database $D$ and a tuple
$\bar a \in q(D) \setminus q'(D)$, together with a valuation $\nu$ such that
$\nu(\bar x)=\bar a$ and $D\models \nu(\body(q))$. Let $\theta\in\Theta_q$ be the
canonical assignment with $\nu\ceqord\theta$ constructed in the proof of
Theorem~\ref{thm:cmp-criterion}, and suppose that
$\theta(\bar x)\in q'(D_\theta)$. Then the argument in that proof transports a
witnessing valuation for $q'$ over $D_\theta$ back to $D$ and yields
$\bar a\in q'(D)$, contradicting the choice of $\bar a$. Hence
$\theta(\bar x)\notin q'(D_\theta)$, while $\theta(\bar x)\in q(D_\theta)$ holds
by the definition of $D_\theta$.
\end{proof}

\subsection{Proofs for Section~\ref{sec:addable}}\label{app:proofs-addable}

\looseness=-1 Throughout this subsection we fix a nontrivial component $C$ of $\CG$ and
write $V=\vars(C)\cup\compCons{C}$.  Assignments are extended to
$\compCons{C}$ by $\theta(k)=k$.  An edge $u\overset{\comp}{\to}v$ is
\emph{satisfied} by an assignment $\theta$ if $\theta(u)\comp\theta(v)$, and
is \emph{violated} otherwise.  Thus a reverse edge is violated exactly when
the induced comparison whose violation it encodes is satisfied.  Every
canonical assignment of $q$ satisfies every edge of $E_q$.

A \emph{produced query} is obtained from $q$ by adding zero or more
comparison atoms whose sides lie in $V$, possibly after merging variables of
$C$ that the atoms entail equal.  A merged vertex stands for its members.
Every produced query $p$ is read with the component $C$, the constants
$\compCons{C}$, the intervals $\allintervals{C}$, and the representative
values of $(q,q')$.  When only $t$ merged variables of $p$ are active on an
open interval that has $\ell\ge t$ representatives, $p$ uses the lowest $t$
of them.  The opposite graph $\OG_p=(V,E_p\cup E_{\tox p})$ of $p$ is
defined as for $q$.

This convention is sound.  The proofs of
Lemma~\ref{lem:component_equality_order_eq} and
Theorem~\ref{thm:cmp-criterion} use the component data only through three
facts, and each remains true for $p$.
(F1)~No comparison atom of $p$ joins variables of two different components
of $\CG$, since the added atoms have both sides in $V$.
(F2)~Every constant that bounds the effective domain of a variable of $C$
lies in $\compCons{C}$, since no new constant is introduced.
(F3)~If a variable of $q'$ partially matches a vertex of $p$, then it
partially matches a member of that vertex in $(q,q')$, since adding atoms
and merging only shrink effective domains.  The proof of
Lemma~\ref{lem:component_equality_order_eq} needs only $t$ ordered
representatives on an interval with $t$ active variables, and the proof of
Theorem~\ref{thm:cmp-criterion} places the two realizers of a
variable--variable atom of $q'$ in one component using only that they are
partially matched in $(q,q')$, which (F3) provides.  Hence
Theorem~\ref{thm:cmp-criterion} and Corollary~\ref{cor:canonical-violating}
hold for every produced query, as do all subsequent results in this subsection that rely on them.

To prove Theorem~\ref{thm:addable-edge-containment}, a canonical assignment
witnessing $q\not\subseteq q'$ must be reordered so that it also satisfies
the added comparison.  The next lemma performs the reordering along any graph
extending $E_q$ and shows that the failure of $q'$ survives it.

\begin{lemma}[Reordering and transport]
\label{lem:reordering-transport}
Let $\theta\in\Theta_q$, and let $G=(V,E)$ be a labeled directed graph with
$E_q\subseteq E$.  Assume no directed cycle of $G$ contains a strict edge
and that $\theta$ is constant on every strongly connected component of $G$.
Then some assignment $\rho$ satisfies:
\begin{enumerate}
  \item \emph{(Locality)} $\rho=\theta$ outside $\vars(C)$ and $\rho(k)=k$ for every
  $k\in\compCons{C}$,
  \item \emph{(Order)} $\rho$ satisfies every edge of $G$, and satisfies it strictly whenever
  its endpoints lie in distinct strongly connected components of $G$,
  \item \emph{(Validity)} $\rho$ satisfies the comparison atoms of $q$, and
  \item \emph{(Inj)} for all $u,v\in\vars(q)$,
  $\rho(u)=\rho(v)\Longrightarrow\theta(u)=\theta(v)$.
\end{enumerate}
If, in addition, every edge of $E_{\tox q}$ satisfied by $\theta$
belongs to $E$, then
$\theta(\bar x)\notin q'(\theta(D_q))\Longrightarrow\rho(\bar x)\notin q'(\rho(D_q))$. \hfill(Tr)
\end{lemma}

\begin{proof}
We contract the strongly connected components of $G$. The resulting condensation is a DAG. We choose a topological order $o$ of its vertices. Since $E_q$ contains the constant edges, the components containing constants appear in $o$ in increasing numerical order. No component contains two distinct constants, since the strict constant edges between them would then lie on a directed cycle.

We now construct $\rho$. Outside $\vars(C)$, set $\rho=\theta$.
For a component $S$ containing a constant $k$, set $\rho(x)=k$ for every
$x\in S$. Since $\theta$ is constant on $S$ and $\theta(k)=k$, we also have
$\theta(x)=k$ for every $x\in S$. In particular, $\rho(k)=k$ for every
$k\in\compCons{C}$. This proves \emph{(Locality)}.

Now let $S$ contain variables but no constant. Let $k$ and $k'$ be,
respectively, the closest constant components before and after $S$ in $o$,
with an infinite endpoint when one does not exist. Then
$I_S=(k,k')\in\allintervals{C}$. The endpoint edges of $E_q$ imply
$I_S\subseteq\feasible{x}$ for every $x\in S$. For all variable components with the same interval $I_S$, choose distinct rational values inside $I_S$ in the order given by $o$. Choose these values also to be distinct from the finitely many values taken by $\theta$ outside $C$. This is possible because every open interval of $\Q$ is infinite.

\looseness=-1 We next verify \emph{(Order)}. The resulting values strictly increase between
distinct components in $o$. Indeed, two variable components with no constant
component between them in $o$ have the same closest constants and hence the
same interval, where we chose the values in the order $o$. If two components
have different intervals, some constant component lies between them. Values
before that constant are smaller than it, unless the component is the constant
itself, and values after it are larger. Distinct constant components appear in
their numerical order. Hence every edge of $G$ is satisfied, with strict inequality between distinct components. This proves \emph{(Order)}.

\looseness=-1 We now verify \emph{(Validity)}. By \emph{(Order)}, all comparison-atom edges of
$q$ are satisfied, and the endpoint edges ensure that every variable remains
in its effective domain. Thus $\rho$ satisfies the comparison atoms of $q$.
This proves \emph{(Validity)}.

It remains to verify \emph{(Inj)}. Two variables of $C$ receive the same value
only when they lie in the same strongly connected component, where $\theta$
is constant. A value shared by a variable of $C$ and a variable outside $C$
can only be a constant $k$, in which case both variables also receive $k$
under $\theta$. Thus
$\rho(u)=\rho(v)\quad\Longrightarrow\quad\theta(u)=\theta(v)$,
which proves \emph{(Inj)}.

\looseness=-1 It remains to prove \emph{(Tr)}. Suppose toward a contradiction that some valuation
$h$ satisfies $q'$ over $\rho(D_q)$ and
$h(\bar x')=\rho(\bar x)$. For every $y\in\vars(q')$, choose a relational atom
of $q'$ containing $y$. Since its image under $h$ belongs to $\rho(D_q)$,
it is the $\rho$-image of a relational atom of $q$. Let $z_y$ be the variable
of that atom at the position of $y$. Then $h(y)=\rho(z_y)$,
and $y$ partially matches $z_y$, so the common value belongs to both relevant
effective domains. Define
$h_0(y)=\theta(z_y)$.
If another variable $z'_y$ is chosen, then
$\rho(z_y)=\rho(z'_y)$, so \emph{(Inj)} gives
$\theta(z_y)=\theta(z'_y)$. Hence $h_0$ is well defined.

\looseness=-1 The same argument applies to relational atoms, position by position. Indeed, for a
relational atom $R(\bar y)$ of $q'$, the tuple $h(\bar y)$ is the $\rho$-image
of some atom $R(\bar z)$ of $q$. At each position, the globally chosen
realizer of the corresponding variable of $q'$ has the same $\rho$-value as
the variable of $\bar z$ at that position, hence the same $\theta$-value by
\emph{(Inj)}. Therefore $h_0(\bar y)=\theta(\bar z)$ and
$R(h_0(\bar y))\in\theta(D_q)$. Similarly, from
$h(x'_i)=\rho(x_i)$ and \emph{(Inj)} we obtain
$h_0(x'_i)=\theta(x_i)$ for every head position, so
$h_0(\bar x')=\theta(\bar x)$.

Since $\theta(\bar x)\notin q'(\theta(D_q))$, some comparison atom of $q'$
must fail under $h_0$. Write this atom as $u_1\comp u_2$, and replace each
variable $u_j$ by its chosen variable $w_j=z_{u_j}$, while a constant stands
for itself. Thus the induced comparison $w_1\comp w_2$ fails under $\theta$.
Since $h(u_j)=\rho(w_j)$ for each $j$ (with constants fixed by $\rho(k)=k$), it suffices to show that the induced comparison fails under $\rho$. Then the original comparison atom also fails under $h$.

If both sides are constants, the comparison has the same fixed truth value
under every assignment, so we may assume that at least one of $w_1,w_2$ is a variable.

If some variable among $w_1,w_2$ lies outside $C$, the induced comparison has
the same truth value under $\theta$ and $\rho$. Constants are fixed, and for
a variable--variable comparison the two corresponding variables lie in the same
component of $\CG$, as shown in the proof of
Theorem~\ref{thm:cmp-criterion}. Hence if one lies outside $C$, so does the
other. We may therefore assume that every variable among $w_1,w_2$ lies in
$C$. Any constant among them then belongs to $\compCons{C}$.

If the comparison atoms of $q$, together with the order of the constants, imply $w_1\comp w_2$, then it cannot fail under $\theta$. If they imply the negation of $w_1\comp w_2$, then $w_1\comp w_2$ also fails under $\rho$, since \emph{(Validity)} says that $\rho$ satisfies the comparison atoms of $q$. Otherwise the induced comparison is undecided by $q$. By the definition of the opposite graph, the reverse edge encoding the negation of $w_1\comp w_2$ lies in $E_{\tox q}$. Since the induced comparison fails under $\theta$, that reverse edge is satisfied by $\theta$, and hence belongs to $E$ by hypothesis. By \emph{(Order)}, it is therefore satisfied by $\rho$, so the induced comparison fails under $\rho$ as well.

Thus some comparison atom of $q'$ fails under $h$, contradicting the assumption that $h$ satisfies $q'$, which proves \emph{(Tr)}.
\end{proof}

\begin{proof}[Proof of Theorem~\ref{thm:addable-edge-containment}.]
Since $q^*$ is obtained from $q$ by adding $x_1<x_2$, we have $q^*\subseteq q$.
Thus $q\subseteq q'$ implies $q^*\subseteq q'$.

For the converse, we prove the contrapositive.  Suppose $q\not\subseteq q'$.
By Corollary~\ref{cor:canonical-violating} there is a canonical assignment
$\theta\in\Theta_q$ such that $\theta(\bar x)\in q(D)\setminus q'(D)$
for $D=\theta(D_q)$.

Let $ F=\{e\in E_{\tox q}:\theta\text{ satisfies }e\}$ 
and let $  G=(V,E_q\cup F\cup\{x_1\overset{<}{\to}x_2\})$.
Every edge of $E_q\cup F$ is satisfied by $\theta$.  Hence, no directed cycle in this subgraph contains a strict edge, and $\theta$ is constant on each of its strongly connected components. Since $x_1<x_2$ is addable, $\OG$
contains no path from $x_2$ to $x_1$.  In particular none
exists in $E_q\cup F$, so the new strict edge creates no directed cycle and the
strongly connected components are unchanged.

Apply Lemma~\ref{lem:reordering-transport}. By \emph{(Order)}, it gives an assignment $\rho$ that satisfies $G$, and hence $\rho(x_1)<\rho(x_2)$. By \emph{(Validity)}, $\rho$ satisfies all comparison atoms of $q$, and by \emph{(Tr)}, it preserves the failure of $q'$. Therefore, $\rho(\bar x)\in q^*(\rho(D_q))\setminus q'(\rho(D_q))$,
so $q^*\not\subseteq q'$. This proves the equivalence.

The proof used only the absence of a directed path from $x_2$ to $x_1$ and
kept constants fixed. Therefore, the same argument applies when one endpoint is a constant of $\compCons{C}$.
\end{proof}

In order to prove Theorem~\ref{thm:trichotomy} we use three lemmas.
Lemma~\ref{lem:path-characterization} characterizes the relations entailed by
a produced query in terms of paths in its opposite graph.
Lemma~\ref{lem:pinned-pairs} shows that once every cycle reverse edge is
decided and the variables that are entailed to be equal are merged, the only directed cycles
left in the opposite graph pin a variable to a constant, so no reverse edge
lies on a cycle any more and every remaining pair can be ordered by an
addable comparison.  Lemma~\ref{lem:cycle-completion} then uses these two lemmas to add such
comparisons to a trichotomy query until only one local assignment remains.
\begin{lemma}[Path characterization]
\label{lem:path-characterization}
Let $r$ be a produced query with satisfiable comparison atoms, and let $s,t$
be vertices of its opposite graph $\OG_r=(V,E_r\cup E_{\tox r})$.  The
comparison atoms of $r$ entail $s\le t$ if and only if $E_r$ contains a
directed path from $s$ to $t$, and entail $s<t$ if and only if at least one such
path contains a strict edge.
\end{lemma}

\begin{proof}
First let $E_r$ contain a directed path from $s$ to $t$.  Every
edge of $E_r$ records a relation that follows from the comparison atoms of
$r$ or from the numerical order of the constants.  An atom edge is given by an atom of
$r$, an endpoint edge records the tightest bound that the atoms of $r$ place
on its variable, and a constant edge is a true inequality between two
constants.  Composing these relations along a path from $s$ to $t$ gives
$s\le t$, and gives $s<t$ when at least one edge on the path is strict.

\looseness=-1 For the other direction, suppose that $E_r$ has no path from $s$ to $t$, or
that no path of $E_r$ from $s$ to $t$ contains a strict edge.  We build a
satisfying valuation of $r$ with $t<s$ in the first case and with $t\le s$ in
the second.  Add to $E_r$ the edge
$t\to s$, labeled $<$ in the first case and $\le$ in the second.  A directed
cycle through the new edge closes a path of $E_r$ from $s$ to $t$. Thus, in the
first case there is no such cycle, and in the second case every such cycle
consists only of non-strict edges.  A directed cycle of $E_r$ that contains a
strict edge would give $u<u$ for any vertex $u$ on it, by the direction just
proved. This is impossible, since the comparison atoms of $r$ are
satisfiable.  Hence the extended graph
has no directed cycle containing a strict edge.

As in the proof of Lemma~\ref{lem:reordering-transport}, contract the strongly
connected components of the extended graph and fix a topological order of
the resulting DAG.  Since the constant edges are strict and belong to $E_r$,
no component contains two distinct constants, and the components containing constants
occur in numerical order.  Assign to each such component its constant, and to
each remaining component a value in the open interval between the nearest
constant components before and after it, choosing these values increasing
along the order.  The resulting assignment $\rho$ fixes every constant and
satisfies every edge of the extended graph, strictly between distinct
components.  In particular $\rho$ satisfies every atom edge of $E_r$, and
hence every comparison atom of $r$ whose variables lie in $C$.
No comparison atom of $r$ joins a variable of $C$ to a variable
outside $C$, because such an atom of $q$ is an edge of $\CG$ and every added
atom has both sides in $V$.  Extending $\rho$ outside $C$ by any satisfying
valuation of $r$, which exists by satisfiability, therefore yields a satisfying
valuation of $r$.  It satisfies the added edge, so it has $t<s$ in the first
case and $t\le s$ in the second, as required.
\end{proof}

\begin{lemma}[Pinned pairs]
\label{lem:pinned-pairs}
Let $p$ be obtained from $q$ by adding for every cycle reverse edge
$e=(u\overset{\comp}{\to}v)\in\cedge$, comparison atoms between $u$ and $v$
that decide the induced comparison encoded by $e$, assume that the comparison
atoms of $p$ are satisfiable, and let $\widetilde p$ be obtained from $p$ by
merging every maximal class of variables of $C$ that its atoms entail to be equal.
Call a variable $x$ \emph{pinned} to a constant $k$ if $\widetilde p$ entails
$x\leq k$ and $k\leq x$.  Then:
\begin{enumerate}
  \item if the comparison atoms of $\widetilde p$ entail $s\leq t$ for vertices
  $s$ and $t$, then $\OG$ contains a directed path from every member of $s$ to
  every member of $t$, and
  \item every strongly connected component of $\OG_{\widetilde p}$ is a
  singleton or a pinned variable--constant pair.
\end{enumerate}
\end{lemma}

\begin{proof}
\looseness=-1 Call a directed path in $\OG$ from a member of $a$ to a member of $b$ an
\emph{expansion} of an edge $a\to b$ of $E_p$ or $\OG_{\widetilde p}$, where
a merged vertex stands for its members. We claim that every such edge has an
expansion.

\looseness=-1 For edges inherited from $q$, this is immediate: every atom, endpoint, or
constant edge is already an edge of $\OG$ between the corresponding members.
Now consider an edge introduced by an added comparison between $u$ and $v$,
where $e=(u\to v)\in\cedge$. If the added edge has the direction of $e$, then
$e$ itself is an expansion. If it has the opposite direction, the rest of a
directed cycle of $\OG$ containing $e$ gives the required path. If an added
variable--constant comparison tightens an endpoint edge, the tightened edge is
simply the edge of that added atom.

It remains to consider reverse edges of $\OG_{\widetilde p}$. Such an edge has
a member-level counterpart in $E_{\tox q}\setminus\cedge$. Indeed, partial
matching can only shrink when passing to $\widetilde p$, so the corresponding
induced comparison already exists for $(q,q')$. Since this comparison is not
decided by $\widetilde p$, it is not decided by $q$, and its reverse edge
therefore belongs to $E_{\tox q}$. It cannot belong to $\cedge$, because $p$
decides the comparison encoded by every edge in $\cedge$.

\looseness=-1 Expansions of consecutive edges of $E_p$ can be concatenated directly, since
$p$ has no merged vertices. For $\OG_{\widetilde p}$, two consecutive expansions
may meet at different members of the same merged vertex. Let $x$ and $y$ be
such members. Since they were merged, $p$ entails both $x\leq y$ and $y\leq x$.
By Lemma~\ref{lem:path-characterization}, $E_p$ therefore contains paths from
$x$ to $y$ and from $y$ to $x$. Expanding these paths gives directed paths in
$\OG$ in both directions. Thus we can connect any two expansions that meet at
the same merged vertex.

It follows that the expansions of all edges on a path in
$\OG_{\widetilde p}$ from $s$ to $t$ can be concatenated into a directed path
in $\OG$, and by the same connections it can start at any member of $s$ and
end at any member of $t$. Claim (1) now follows from
Lemma~\ref{lem:path-characterization} applied to $\widetilde p$.

No reverse edge of $\OG_{\widetilde p}$ can lie on a directed cycle. Suppose
that a reverse edge $a\to b$ did lie on such a cycle. Then the rest of the
cycle gives a directed path from $b$ back to $a$. By the previous argument,
this path expands to a directed path in $\OG$ from a member of $b$ to a member
of $a$. The reverse edge itself has a member-level counterpart in
$E_{\tox q}\setminus\cedge$, so together they form a directed cycle in $\OG$
containing that counterpart. This is impossible, since any such edge lying on
a directed cycle would belong to $\cedge$.

\looseness=-1 For claim (2), suppose that a strongly connected component of
$\OG_{\widetilde p}$ contains two distinct vertices $s$ and $t$. There are
directed paths from $s$ to $t$ and from $t$ to $s$. Every edge on these paths
lies on a directed cycle, so by the previous paragraph none of them is a
reverse edge. Thus all of them are atom, endpoint, or constant edges, and
therefore record relations entailed by $\widetilde p$ or by the order of the
constants.

\looseness=-1 Neither path can contain a strict edge. Otherwise, one path would imply $s<t$
while the other implies $t\leq s$, or vice versa, contradicting
satisfiability. Hence $\widetilde p$ entails both $s\leq t$ and $t\leq s$, so
$s$ and $t$ are entailed equal.

By construction of $\widetilde p$, two distinct merged variables cannot be
entailed equal, and neither can two distinct constants. Therefore a
non-singleton strongly connected component contains exactly one variable $x$
and one constant $k$. Since they are entailed equal, $\widetilde p$ entails
both $x\leq k$ and $k\leq x$, so $x$ is pinned to $k$.
\end{proof}

\begin{lemma}[Completion after the cycle decisions]
\label{lem:cycle-completion}
Let $p,\widetilde p$ be as in Lemma~\ref{lem:pinned-pairs}. One can, in polynomial time, add $O(|V|^2)$ comparison atoms on $V$ to
$\widetilde p$ and obtain $\widehat p$ with:
\begin{enumerate}
  \item $p\subseteq q'$ if and only if $\widehat p\subseteq q'$,
  \item the canonical assignments of $\widehat p$ induce a single local
  assignment on $C$, and
  \item this local assignment satisfies every edge of
  $E_{\tox q}\setminus\cedge$.
\end{enumerate}
\end{lemma}

\begin{proof}
  \looseness=-1 Merging identifies only variables that $p$ entails to be equal, so
$p\subseteq q'$ if and only if $\widetilde p\subseteq q'$. Throughout the
construction, the current query is produced and satisfies (F1)--(F3). Hence
Theorem~\ref{thm:addable-edge-containment} applies at every stage, and by
Lemma~\ref{lem:path-characterization}, every entailed relation is witnessed
by a path.

\emph{Protection.}
Let $e=(u\overset{\comp}{\to}v)\in E_{\tox q}\setminus\cedge$. If the current
atoms already imply $u\comp v$, no modification is needed. They cannot imply
its negation, since the corresponding reverse path would expand, as
in Lemma~\ref{lem:pinned-pairs}, to a path from $v$ to $u$ in $\OG$, placing
$e$ on a directed cycle. Thus, whenever $u$ and $v$ are not yet strictly
ordered in the direction of $e$, the comparison $u<v$ is a candidate. It is
addable for the same reason: a reverse path in the current opposite graph
would expand to one in $\OG$. We therefore add such comparisons for all edges
in $E_{\tox q}\setminus\cedge$. Since added edges are parallel to the
corresponding protected edges and reverse edges can only disappear, the
argument remains valid throughout. Afterwards, every valuation of the current
query satisfies every edge of $E_{\tox q}\setminus\cedge$.

\emph{Completion.}
We maintain the invariant that every strongly connected component of the
current opposite graph is either a singleton or a pinned variable--constant
pair. This holds initially by Lemma~\ref{lem:pinned-pairs}(2) and is preserved
by protection, since addable edges create no directed cycle.

We repeatedly apply the following two operations. First
(\emph{materialization}), if an entailed
variable--constant bound is not yet an atom, we add it. By
Lemma~\ref{lem:path-characterization}, the new edge is parallel to an existing
path, so any cycle created by it stays within an existing pinned component.
Second (\emph{saturation}), if two vertices are neither strictly ordered nor
entailed to be equal,
at least one strict orientation between them is addable. Otherwise, they would
lie in the same strongly connected component and hence be entailed to be
equal. We add such an orientation and repeat.

\looseness=-1 These operations preserve the invariant and satisfiability. They also
introduce no new equality. Any new equality using an added strict edge would
create a strict directed cycle. Each operation adds a new atom, so after
$O(|V|^2)$ additions the process terminates with a query $\widehat p$. By
Theorem~\ref{thm:addable-edge-containment}, containment equivalence with $q'$
is preserved at every strict addition, while materialization adds only
entailed atoms. Hence $p\subseteq q' \Longleftrightarrow \widehat p\subseteq q'$,
which proves (1). All required tests are reachability tests in graphs of
polynomial size, so $\widehat p$ is computable in polynomial time.

\emph{One local assignment.}
In $\widehat p$, every two vertices are either strictly ordered or entailed to
be equal, and every entailed variable--constant bound is an atom. Thus every
pinned variable is fixed to its constant, while every unpinned variable lies
in a unique interval $I\in\allintervals{C}$. The merged variables active on
$I$ form a strict total order $x_1<\cdots<x_t$.

For $x_j$, exactly $x_{j+1},\ldots,x_t$ are forced above it, so its rank bound
is $j$. Hence its canonical set is $\{\alpha^{C,I}_1,\ldots,\alpha^{C,I}_j\}$, and strict increase forces
$x_j\mapsto\alpha^{C,I}_j$ by induction on $j$.

It remains only to check that these values are canonical for $q$. Let
$\ell=|\actvars{C}{I}|$, and let $x$ belong to the class of $x_j$. No lower
class contains a variable $y$ with $q\models x\le y$, since $\widehat p$ is
satisfiable and extends $q$, and every lower class remains active on $I$ for
$q$. Choosing one member from each lower class gives $j-1
\le
\ell-1-
\left|
\{y\in\actvars{C}{I}\setminus\{x\}:q\models x\le y\}
\right|$,
so $j\le\rankmax{x}{I}$ and
$\alpha^{C,I}_j\in\repvalx{C}{I}{x}$.
Pinned variables similarly take the canonical value of their singleton
interval. Hence $\widehat p$ induces a single local assignment on $C$,
proving (2).

Finally, (3) follows from the protection step: every valuation of
$\widehat p$ satisfies every edge of $E_{\tox q}\setminus\cedge$.
\end{proof}

\begin{proof}[Proof of Theorem~\ref{thm:trichotomy}]
  Every trichotomy query $q_\sigma$ is obtained from $q$ by adding comparison
atoms, so $q_\sigma\subseteq q$.  Hence $q\subseteq q'$ implies
$q_\sigma\subseteq q'$ for every $\sigma$.

\looseness=-1 Conversely, suppose $q\not\subseteq q'$.  Choose a database $D$, a tuple
$\bar a\in q(D)\setminus q'(D)$, and a satisfying valuation $\nu$ of $q$ with
$\nu(\bar x)=\bar a$, extended to constants by $\nu(k)=k$.  For every
$e=(u\to v)\in\cedge$, exactly one of $\nu(u)<\nu(v)$, $\nu(u)=\nu(v)$, or $\nu(v)<\nu(u)$
holds.  Let $\sigma$ record these relations.  Then $\nu$ satisfies every atom
added to $q_\sigma$, and therefore
$ \bar a\in q_\sigma(D)\setminus q'(D)$. Thus $q_\sigma\not\subseteq q'$, proving $q\subseteq q'$ if and only if $q_\sigma\subseteq q'$ for every trichotomy query $q_\sigma$.

For the local-assignment claim, fix $\sigma$.  If $q_\sigma$ is unsatisfiable,
its canonical test is vacuous.  Otherwise $q_\sigma$ decides the induced
comparison encoded by every edge of $\cedge$, so
Lemma~\ref{lem:cycle-completion} applies.  It yields a completed query
$\widehat q_\sigma$ with $q_\sigma\subseteq q' \Longleftrightarrow \widehat q_\sigma\subseteq q'$, whose canonical assignments induce one local assignment on $C$.  Thus each
trichotomy case is tested over at most one local assignment, and the component
contributes at most $3^{|\cedge|}$ local assignments.
\end{proof}

\begin{proof}[Proof of Theorem~\ref{thm:feedback-query-containment}]
  \looseness=-1 Assume $\OG$ has no non-strict cycle.  Every feedback query $q_i$
adds comparisons to $q$, so $q_i\subseteq q$.  Hence
$q\subseteq q'$ implies $q_i\subseteq q'$ for every minimal $\mfac$ $E_i$.

For the converse, suppose $q\not\subseteq q'$.  By Corollary~\ref{cor:canonical-violating} there is a
canonical assignment $\theta\in\Theta_q$ such that $  \theta(\bar x)\in q(\theta(D_q))\setminus q'(\theta(D_q))$.

\looseness=-1 Let $U=\{e\in\cedge:\theta\text{ violates }e\}$.
We claim that $U$ is an $\mfac$.  Otherwise $ (V,E_q\cup(E_{\tox q}\setminus U))$
contains a directed cycle $\gamma$.  Every reverse edge on $\gamma$ lies on a
directed cycle of the original $\OG$, and hence belongs to
$\cedge\setminus U$.  It is therefore satisfied by $\theta$.  The edges of
$E_q$ are satisfied as well.  Thus $\theta$ is nondecreasing around $\gamma$,
which forces all vertices of $\gamma$ to receive the same value. However, by the assumption of the theorem, $\gamma$ contains a strict edge, a contradiction.

\looseness=-1 Choose a minimal $\mfac$ $E\subseteq U$, and set $ H=(V,E_q\cup(E_{\tox q}\setminus E))$.
Then $H$ is a DAG.  Every reverse edge satisfied by $\theta$ belongs to $H$. Indeed,
a satisfied cycle reverse edge is not in $U$ and hence not in $E$, while a
reverse edge outside $\cedge$ is never removed.  Apply
Lemma~\ref{lem:reordering-transport} to $H$.  Since $H$ is acyclic, the
resulting assignment $\rho$ is strictly increasing along every edge of $H$ by
(Order), and $\rho(\bar x)\notin q'(\rho(D_q))$ by (Tr).

Let $q_E$ be the feedback query associated with $E$.  Its original comparison
atoms hold under $\rho$.  If
$e=(u\overset{\comp}{\to}v)\in\cedge\setminus E$, then $e\in H$, so $\rho$
satisfies the comparison $u\comp v$ asserted by $q_E$.  If
$e=(u\overset{\comp}{\to}v)\in E$, the minimality of $E$ implies that adding
$e$ back to $H$ creates a directed cycle.  Hence $H$ contains a directed path
from $v$ to $u$.  Since $\rho$ is strict along $H$,
$\rho(v)<\rho(u)$, and therefore $\rho$ satisfies the negation
$v\altcomp u$ added by $q_E$.  Thus
$\rho(\bar x)\in q_E(\rho(D_q))$.

Together with $\rho(\bar x)\notin q'(\rho(D_q))$, this gives $q_E\not\subseteq q'$. This proves the characterization.

\emph{One local assignment.}
If the comparison atoms of $q$ are unsatisfiable, which can be checked in
polynomial time, then $q$ and every $q_E$ are contained in $q'$ and have no
canonical assignment, so the claims hold trivially.  Assume therefore that
they are satisfiable.  Then every feedback query $q_E$ is satisfiable.  Fix
any $\theta\in\Theta_q$, which exists by
Lemma~\ref{lem:component_equality_order_eq}.  For a minimal $\mfac$ $E$, the graph $H$ above is a DAG, so
$\theta$ is trivially constant on each of its strongly connected components.
Lemma~\ref{lem:reordering-transport}~(Order) therefore
realizes every edge of $H$ strictly, and the reverse path supplied by minimality
satisfies the negation added for every edge of $E$.  The query $q_E$ decides the induced comparison
encoded by every edge of $\cedge$, so Lemma~\ref{lem:cycle-completion} applies
and yields one local canonical assignment on $C$.

\emph{Polynomial-delay generation.}
We may assume that $E_q$ is
acyclic, since otherwise every directed cycle of $E_q$ contains a strict edge
by the assumption of the theorem, and no $\mfac$ exists.
\looseness=-1 Define a
digraph
$T$ whose vertices are the edges of $\cedge$.  Put an edge $e\to f$ whenever the
head of $e$ reaches the tail of $f$ by a path consisting only of $E_q$-edges,
where the path may be empty when the two vertices coincide.

For $E\subseteq\cedge$, write $  \OG_E=(V,E_q\cup(E_{\tox q}\setminus E))$.
We show that $E$ is an $\mfac$ of $\OG$ if and only if $E$ is a feedback
vertex set of $T$.
If $\OG_E$ has a directed cycle, every reverse edge on that cycle belongs to
$\cedge\setminus E$, and because $E_q$ is acyclic the cycle contains at least
one such edge.  These reverse edges are pairwise distinct, and reading them in
cyclic order, the $E_q$-paths between consecutive ones witness the edges of a
directed cycle of $T-E$, a self-loop when there is only one.

Conversely, let $  e_1\to e_2\to\cdots\to e_m\to e_1$
be a directed cycle in $T-E$.  For each $i$, follow the reverse edge $e_i$ and
then the $E_q$-path witnessing the edge from $e_i$ to $e_{i+1}$ (indices modulo
$m$).  Following these pieces in order leads from the head of $e_1$ back to
the tail of $e_1$, so $\OG_E$ contains a directed path between them.  Together
with $e_1$ it forms a directed cycle.  This proves the equivalence.  The
correspondence preserves inclusion, so the minimal $\mfac$s are exactly the
minimal feedback vertex sets of $T$.

The graph $T$ is computable in polynomial time by reachability in $E_q$.  A
vertex of $T$ carrying a self-loop belongs to every feedback vertex set.  Remove
all such vertices and enumerate the minimal feedback vertex sets of the
remaining loop-free digraph with polynomial delay using the algorithm of
Schwikowski and Speckenmeyer~\cite{schwikowski2002enumerating}.  Adding the forced loop vertices back gives
exactly the minimal $\mfac$s, once each.

For each output $E$, the feedback query $q_E$ can be constructed in polynomial
time. By the satisfiability argument above, Lemma~\ref{lem:cycle-completion}
computes its unique local assignment in polynomial time. Indeed, each materialization or saturation step requires only a reachability computation in a graph of polynomial size, and at most $O(|V|^2)$ comparison atoms are added.  Hence the feedback
queries, together with their local assignments, are generated with polynomial
delay.
\end{proof}

\subsection{Proofs for Section~\ref{sec:combined}}\label{app:proofs-combined}
\looseness=-1 A value is \emph{private} under a combined assignment $\theta$ if exactly one
variable of $q$ is assigned that value.

\begin{observation}
\label{obs:private-combined}
\looseness=-1 Let $\theta$ be a combined assignment, let $h$ be a satisfying valuation of
$q'$ over $\theta(D_q)$, and let $y\in\vars(q')$ satisfy
$h(y)=\theta(z)$, where $\theta(z)$ is private. Then $ \POS{q'}{y}\subseteq \POS{q}{z}$.

\end{observation}

\begin{proof}
Let $(R,j)\in\POS{q'}{y}$. Choose an $R$-atom of $q'$ in which $y$ occurs
at position $j$. Since $h$ satisfies $q'$ over $\theta(D_q)$, the image of
this atom under $h$ belongs to $\theta(D_q)$, and hence is the $\theta$-image
of some $R$-atom of $q$. The variable occurring at position $j$ in that atom
is mapped by $\theta$ to $h(y)=\theta(z)$. Since $\theta(z)$ is private, that
variable must be $z$. Thus $(R,j)\in\POS{q}{z}$, which proves the claim.
\end{proof}

\begin{proof}[Proof of Theorem~\ref{thm:comb-criterion}.]
Every $D\in\Dcomb(q,q')$ is a database, so the only-if direction is immediate.
For the converse, assume that $q(D)\subseteq q'(D)$ for every $D\in\Dcomb(q,q')$.
Let $D_0$ be a database, let $\nu$ satisfy $q$ over $D_0$, and set
$\bar a=\nu(\bar x)$. We show that $\bar a\in q'(D_0)$.

\emph{Head positions.}
Let $x_i$ be a head variable in no comparison atom. We claim that
$\POS{q'}{x'_i}\subseteq\POS{q}{x_i}$.
We choose a combined assignment $\theta_1$ that assigns a private value to $x_i$. If $x_i$ is isolated, assign it the fresh constant from an infinite witness in $\witset_{x_i}$, which exists since $\feasible{x_i}=\Q$.
If $x_i$
lies in a nontrivial component $C$, then by normalization the comparison atoms
of $q$ admit a satisfying assignment that is injective on $C$ and places $x_i$
in an open interval. The construction of
Lemma~\ref{lem:component_equality_order_eq} turns it into a canonical
assignment that is still injective on $C$, and we null $V_n$. Since
$x_i$ is a head variable, $x_i\notin V_n$, so its value is a private
open-interval representative of $C$. Since $\theta_1(\bar x)\in q(\theta_1(D_q))$, the assumption gives a satisfying valuation $h_1$ of $q'$ over
$\theta_1(D_q)$ with $h_1(x'_i)=\theta_1(x_i)$, and
Observation~\ref{obs:private-combined} gives the claim.

\emph{The assignment $\theta$.}
Let $\theta(x)=\bot$ for $x\in V_n$ and for $x\in V_t$ with $\nu(x)=\bot$.
Every other $x$ receives a canonical value as follows.
\begin{enumerate}
  \item[(a)] If $\nu(x)\neq\bot$, take the value that mimics $\nu(x)$ in the
  proof of Theorem~\ref{thm:vc-criterion} when $x$ is isolated. If $x$
  lies in a nontrivial component, use the representative produced by
  Lemma~\ref{lem:component_equality_order_eq} from the non-null part of $\nu$.
  \item[(b)] If $\nu(x)=\bot$, then $x\in V_f$ and $x$ is uncompared. If $x$
  is isolated, choose the fresh constant of an infinite witness in $\witset_x$.
  If $x$ lies in a nontrivial component $C$, then $x$ is active on every
  interval and every representative is canonical for $x$. On each open
  interval $I$, rule (a) uses one representative for each $\nu$-equality class
  whose value lies in $I$. Hence enough representatives remain to assign the
  variables of rule (b) in $C$ distinct unused representatives.
\end{enumerate}
\looseness=-1 Compared variables lie in $\nonnull(q)$, so they are non-null under $\nu$ and
fall under rule (a). Hence $\nu\ceqord\theta$ on the variables that both
assignments leave non-null, so $\theta$ satisfies the variable--variable
comparison atoms of $q$, and $\theta(x)\in\feasible{x}$ gives the
variable--constant ones. Thus $\theta$ is a combined assignment. Moreover, a
value chosen by rule (b) is taken by no other variable, so it is private.

\emph{The valuation $\nu'$.}
Since $\theta(\bar x)\in q(\theta(D_q))$, the assumption gives a satisfying valuation $h$ of $q'$ over $\theta(D_q)$ with
$h(\bar x')=\theta(\bar x)$. For each relational atom $A=R(\bar y)$ of $q'$,
fix an atom $R(\bar z_A)$ of $q$ with $\theta(\bar z_A)=h(\bar y)$. We define
$\nu'$ by two rules: (1) If $h(y)\neq\bot$, set $\nu'(y)=\nu(z)$ for any $z$ with
  $\theta(z)=h(y)$. (2) If $h(y)=\bot$, then $y\notin\nonnull(q')$ occurs once, in some
  atom $A$, and we set $\nu'(y)=\nu(z)$ for the variable $z$ of $\bar z_A$ at
  the position of $y$.

\looseness=-1 Rule (1) is well defined. If $\theta(z_1)=\theta(z_2)\neq\bot$, then the
value is private and $z_1=z_2$, or it is an open-interval representative of
rule (a) and $z_1,z_2$ share their $\nu$-value, or it is a finite-witness or
singleton-interval value and $\theta(z_j)=\nu(z_j)$ for both $j$. Hence $\nu(z_1)=\nu(z_2)$ in all cases.

\emph{Relational atoms.}
For $A=R(\bar y)$, rules (1) and (2) give $\nu'(\bar y)=\nu(\bar z_A)$, and
$R(\nu(\bar z_A))\in D_0$ because $\nu$ satisfies $q$.

\emph{Non-null variables of $q'$.}
\looseness=-1 Let $w\in\nonnull(q')$, so $h(w)\neq\bot$ and $\nu'(w)=\nu(z)$ for some $z$
with $\theta(z)=h(w)$. Suppose that $\nu(z)=\bot$. Then $z$ falls under
rule (b), so $\theta(z)$ is private and Observation~\ref{obs:private-combined}
gives $\POS{q'}{w}\subseteq\POS{q}{z}$. As $z$ is uncompared and non-join,
$z\in V_f\setminus\nonnull(q)=V_{CN}\cup(\head(q)\setminus V_{CJ})$ and
$\POS{q}{z}$ is a singleton, so $\POS{q'}{w}=\POS{q}{z}$. Now $z$ is covered.
This is immediate for $z\in V_{CN}$, and for $z=x_i\in\head(q)$ it follows
from the head-position claim, which gives $\POS{q'}{x'_i}=\POS{q}{z}$. Hence $w$ puts $z$ into
$V_{CJ}$, contradicting $z\notin V_{CJ}$. So $\nu'(w)\neq\bot$.

\emph{Comparison atoms.}
\looseness=-1 For a compared variable $y$ of $q'$, pick an atom $A$ containing $y$ and let
$z_y$ be the variable of $\bar z_A$ at the position of $y$. Then
$\theta(z_y)=h(y)$ lies in both effective domains, so $y$ partially matches
$z_y$, and $\nu'(y)=\nu(z_y)$ by rule (1). By the previous paragraph
$\nu(z_y)\neq\bot$, so $z_y$ falls under rule (a).

\looseness=-1 Let $y\comp k$ be an atom of $q'$. If $z_y$ is isolated, either
$\theta(z_y)=\nu(z_y)$, so
$\nu'(y)=h(y)\in\fullfeasible{q'}{y}$, or $\theta(z_y)$ is the private
constant of an infinite witness. Then
Observation~\ref{obs:private-combined} gives $y\in\tox{z_y}$. Since
$\theta(z_y)\in\fullfeasible{q'}{y}$, we have
$y\notin\forbidden{\theta}{z_y}
\supseteq
\forbidden{\nu}{z_y}$,
and hence $\nu(z_y)\in\fullfeasible{q'}{y}$.

If $z_y$ lies in a nontrivial component $C$, then the bounds of
$\fullfeasible{q'}{y}$ lie in $\compCons{C}$, while $\nu(z_y)$ and
$\theta(z_y)$ lie in the same interval of $\allintervals{C}$. This interval
is either contained in or disjoint from $\fullfeasible{q'}{y}$. Since
$\theta(z_y)\in\fullfeasible{q'}{y}$, also
$\nu(z_y)\in\fullfeasible{q'}{y}$.

Let $y\comp y'$ be an atom of $q'$. As in the proof of
Theorem~\ref{thm:cmp-criterion}, $z_y$ and $z_{y'}$ lie in one nontrivial
component, and $\nu\ceqord\theta$ on them transfers
$\theta(z_y)\comp\theta(z_{y'})$ to $\nu(z_y)\comp\nu(z_{y'})$.

\emph{Head.}
\looseness=-1 Let $i$ be a head position, so $h(x'_i)=\theta(x_i)$. If
$\theta(x_i)\neq\bot$, then $h(x'_i)\neq\bot$, and rule (1) with
$z=x_i$ gives $\nu'(x'_i)=\nu(x_i)$.
Otherwise $x_i\in V_t$ with $\nu(x_i)=\bot$, so $x_i$ is uncompared, and $h$
nulls $x'_i$, so $\nu'(x'_i)=\nu(z)$ for the variable $z$ at the position of
$x'_i$ in the fixed atom, where $\theta(z)=\bot$. If $\nu(z)\neq\bot$, then
$z\in V_n$, so $z$ is non-join and $\POS{q}{z}$ is that single position. By
the head-position claim, and since $x_i\in V_t\subseteq V_C$ is non-join,
$\POS{q}{z}=\POS{q'}{x'_i}=\POS{q}{x_i}$. Coveredness depends only on the
position set and $x_i$ is covered, so $z$ is covered, contradicting
$z\in V_n$. Hence $\nu(z)=\bot=\nu(x_i)$.

\looseness=-1 Thus $\nu'$ satisfies $q'$ over $D_0$ with $\nu'(\bar x')=\bar a$, so
$\bar a\in q'(D_0)$. Since $D_0$ and $\nu$ were arbitrary,
$q\subseteq_\bot q'$.
\end{proof}

As in Corollary~\ref{cor:canonical-violating}, the proof of
Theorem~\ref{thm:comb-criterion} yields the following stronger
fact: when containment fails, it fails at the tuple of a combined assignment
itself.

\begin{corollary}
\label{cor:combined-counterexample}
Let $q,q'\in\cmpcombclass$.  If $q\not\subseteq_\bot q'$, then
$ \theta(\bar x)\notin q'(\theta(D_q))$
for some combined assignment $\theta$.
\end{corollary}

\begin{proof}
The converse direction of Theorem~\ref{thm:comb-criterion} uses its hypothesis only at tuples
$\theta(\bar x)$.  Hence, if
$\theta(\bar x)\in q'(\theta(D_q))$ for every combined assignment $\theta$,
that proof turns an arbitrary satisfying valuation of $q$ over an arbitrary
database into a satisfying valuation of $q'$ with the same output.  Thus
$q\subseteq_\bot q'$.
\end{proof}

\begin{lemma}[Component pinning under nulls]
\label{lem:component-pinning-nulls}
Let $q,q'\in\cmpcombclass$, let $C$ be a nontrivial component of $\CG$, and
let $P$ be a set with $V_n\subseteq P\subseteq V_n\cup V_t$, that is, a
possible null pattern of a combined assignment.
Let $\OG_P$ be the opposite graph of $C$ restricted to the vertices outside
$P$, let $\cedge_P(C)$ be its cycle reverse edges, and let
$\sigma\colon\cedge_P(C)\to\{<,=,>\}$ be a choice whose trichotomy query
$q_\sigma$ is satisfiable. Then there is a single local assignment
$\theta_{P,\sigma}$ of the variables of $C$ outside $P$ with the following
property. Let $\theta$ be a combined assignment with null pattern $P$ that
gives the endpoints of every edge of $\cedge_P(C)$ the relation recorded by
$\sigma$, and let $\theta^*$ be obtained from $\theta$ by replacing its values
on the variables of $C$ outside $P$ with those of $\theta_{P,\sigma}$. If
$\theta(\bar x)\notin q'(\theta(D_q))$, then
$\theta^*(\bar x)\notin q'(\theta^*(D_q))$.
\end{lemma}

\begin{proof}
\looseness=-1 A nulled variable is uncompared, so it has effective domain $\Q$ and carries
only reverse edges. Hence $\OG_P$ keeps the forced edges $E_q$ of $C$ on the
surviving vertices, with the original constants, intervals, and
representatives, and the atoms of $q_\sigma$ relate surviving vertices and
decide the comparison encoded by every edge of $\cedge_P(C)$. Lemmas~\ref{lem:pinned-pairs}
and~\ref{lem:cycle-completion} therefore apply to $q_\sigma$ read on $\OG_P$.
Merging the variables that $q_\sigma$ entails equal and completing the order
yields a single local assignment $\theta_{P,\sigma}$ on the surviving
variables, which depends only on $C$, $P\cap\vars(C)$, and $\sigma$. Its
values are canonical values of $q$: the rank check in the proof of
Lemma~\ref{lem:cycle-completion} counts $\ell=|\actvars{C}{I}|$ for $q$ itself
and uses only that no lower class contains a variable $y$ with
$q\models x\le y$, which holds here as well, and a variable pinned to a
constant $k$ takes $k\in\feasible{x}$. By Lemma~\ref{lem:cycle-completion}(3)
and the atoms of $q_\sigma$, the assignment $\theta_{P,\sigma}$ satisfies every
reverse edge of $\OG_P$ that lies on no directed cycle and realizes $\sigma$
on $\cedge_P(C)$. \hfill(P)

\looseness=-1 Now let $\theta$ be as stated. Its non-null part on $C$ satisfies the
comparison atoms of $q$ and realizes $\sigma$, so it satisfies every atom of
$q_\sigma$. Suppose that a satisfying valuation $h$ of $q'$ over
$\theta^*(D_q)$ has $h(\bar x')=\theta^*(\bar x)$. Set $h'(y)=\bot$ if
$h(y)=\bot$, and otherwise $h'(y)=\theta(z)$ for any $z$ with
$\theta^*(z)=h(y)$.

\looseness=-1 This is well defined. Suppose that $\theta^*(z_1)=\theta^*(z_2)\neq\bot$. If neither
$z_j$ is a surviving variable of $C$, then $\theta$ and $\theta^*$ agree on
both. If both are, they lie in one merged class, since $\theta_{P,\sigma}$
separates distinct classes, so $q_\sigma$ entails $z_1=z_2$ and $\theta$
equates them. If exactly one is, the shared value is a constant $k$ to which
that variable is pinned, because representatives of $C$ are taken by no
variable outside $C$, and then $\theta$ gives it $k$ as well. Thus
$\theta(z_1)=\theta(z_2)$ in every case.

\looseness=-1 Consequently $h'$ sends every relational atom of $q'$ into $\theta(D_q)$,
position by position. It nulls exactly the variables that $h$ nulls, hence no
variable of $\nonnull(q')$, and $h'(\bar x')=\theta(\bar x)$: at a head
position with $h(x'_i)\neq\bot$ we have $\theta^*(z)=\theta^*(x_i)$ for the
chosen $z$, so $\theta(z)=\theta(x_i)$, and otherwise both assignments null
$x_i$.

Let $\alpha$ be a comparison atom of $q'$, satisfied by $h$. Choose its
realizers positionally as in the proof of
Lemma~\ref{lem:reordering-transport}, and write the induced condition as
$w_1\comp w_2$, a constant standing for itself. Realizers are non-null, so a
realizer in $C$ survives, and for a variable--variable atom both realizers
lie in one component, as in the proof of Theorem~\ref{thm:cmp-criterion}. If
no realizer lies in $C$, then $\theta$ and $\theta^*$ agree on $w_1,w_2$.
Otherwise $w_1,w_2$ are vertices of $\OG_P$, and there are three cases. If
the comparison atoms of $q$ and the order of the constants decide
$w_1\comp w_2$, its truth value is the same under $\theta$ and $\theta^*$. If
its reverse edge lies in $\cedge_P(C)$, the same holds, since $\sigma$ records
the relation under $\theta$ and $\theta_{P,\sigma}$ realizes $\sigma$.
Otherwise its reverse edge lies on no directed cycle of $\OG_P$, so by (P) it
is satisfied by $\theta^*$, and $w_1\comp w_2$ fails under $\theta^*$,
contradicting that $h$ satisfies $\alpha$. Hence $w_1\comp w_2$ holds under
$\theta$, and $h'$ satisfies $\alpha$.

Thus $h'$ witnesses $\theta(\bar x)\in q'(\theta(D_q))$, contradicting the
assumption on $\theta$.
\end{proof}

\begin{proof}[Proof of Theorem~\ref{thm:comb-full-mult}.]
  \looseness=-1 By Theorem~\ref{thm:comb-criterion} containment is decided by the combined canonical family, and by
Corollary~\ref{cor:combined-counterexample} a failure appears at the tuple
$\theta(\bar x)$ of some combined assignment.  Fix a null pattern $P$ and
apply Lemma~\ref{lem:component-pinning-nulls} to the nontrivial components one
at a time.  A violating assignment is thereby replaced by one that agrees with
$\theta_{P,\sigma}$ on every nontrivial component and still violates, so only
such assignments are needed. Since
$ \cedge_P(C)\subseteq\cedge(C)$
for every $P$, indexing the cases by maps on the original cycle-edge sets
leaves at most $3^{|\cedge(q)|}$ cases.

Within a fixed case of the trichotomy, a variable in a nontrivial component
takes the value assigned by $\theta_{P,\sigma}$, unless it is mapped to
$\bot$. Thus, a variable in $V_t$ has two possible choices, while every other
variable in such a component has only one. Variables outside nontrivial
components retain the choices defined in Section~\ref{sec:combined}.

\looseness=-1 Now consider a legal separator $S$. Since $\RG^+$ contains a clique on every
nontrivial component, each such component is contained in $S\cup C_i$ for
some component $C_i$ of $\RG^+-S$. The choices made for the variables of a
nontrivial component must be considered together, since they determine a
single assignment for that component. In particular, if the component meets
both $S$ and $C_i$, the choices on its two parts are not independent. We
therefore use one choice per variable and apply the construction of
Lemma~\ref{lem:s-exhaustive-size} with these choices in place of the
canonical values.

\looseness=-1 This construction pairs every choice vector on $S$ with a choice vector on
each $C_i$. Not every such pair need be realizable in the fixed case. Indeed,
if a nontrivial component meets both $S$ and some $C_j$, the combined choices
may make its trichotomy query unsatisfiable. Whenever a member is
unrealizable on some $C_j$, we replace
the choices on $C_j$ by those of any realizable assignment that agrees with
the given choices on $S$. If no such assignment exists, we simply omit that
member.

\looseness=-1 For the pair $(\theta,i)$ used to establish $S$-exhaustiveness in
Lemma~\ref{lem:s-exhaustive-size}, the member is never omitted, since
$\theta$ itself is a realizable assignment with the same choices on $S$, and
this replacement can affect only
components outside $S\cup C_i$. It therefore does not change the required
agreement on $S\cup C_i$. Hence the resulting family is still $S$-exhaustive,
and its size remains bounded by $\left(\prod_{v\in S}\valcnt(v)\right)
\cdot
\max_i\prod_{u\in C_i}\valcnt(u)$.

\looseness=-1 The stitching of Theorem~\ref{thm:s-exhaustive-fixed} applies to this subfamily with two additions.
First, covering-matches are read with $\bot$ as a canonical value of every
variable of $V_n\cup V_t$.  The cohort of $\bot$ is
$V_\bot=V_n\cup V_t$, and $\bot$ is a witness value for $y$ whenever this
cohort covers $\POS{q'}{y}$.  The effective-domain condition is vacuous because
a satisfying valuation nulls only variables of $q'$ occurring in no comparison
atom.  Lemmas~\ref{lem:atom-locality} and~\ref{lem:realizer} are unaffected, since they use only that the value
placed at a variable of $q'$ is a witness value for it.  The head step is
unaffected for the same reason: a head variable sent to $\bot$ lies in
$V_\bot$.

Second, the stitching step that in Section~\ref{sec:vc-components} was restricted to
variable--constant atoms of $q'$ is repaired by the cohort edges of $\RG^+$.
For a variable--variable atom $y\comp y'$ of $q'$, these edges place $\cmctox{y}\cup\cmctox{y'}$
outside the separator in one component $C_i$.  Hence the stitched valuation
takes both values from the witness $\mu_i$ of that component: directly when a
variable is served from $C_i$, through agreement of the witnesses on $S$ when
one is served from the separator, and from any witness when both cohorts lie in
$S$.  Since $\mu_i$ satisfies the comparison, so does the stitched valuation.

Containment therefore holds if and only if every member of every case
passes the test $\psi(\bar x)\in q'(\psi(D_q))$.  A violating member witnesses
the failure of containment directly, and a failure of containment yields a
violating combined assignment, then a violating assignment of some case, and
then, by the stitching, a violating member of that case.  Each case
contributes at most
$\Bigl(\prod_{v\in S}\valcnt(v)\Bigr)
  \cdot
  \max_i\prod_{u\in C_i}\valcnt(u)$
databases.  With at most $3^{|\cedge(q)|}$ cases this gives the stated bound.
The numeric estimate follows from
$\valcnt(u)\leq 2|\tox{u}|+2\leq 2\Delta+2$ outside nontrivial components, by Proposition~\ref{prop:witness_set_size}, and from
$\valcnt(u)\leq2$ inside them.
\end{proof}

\begin{lemma}[Feedback pinning under nulls]
\label{lem:feedback-pinning-nulls}
\looseness=-1 Let $q,q'\in\cmpcombclass$, let $C$ be a nontrivial component of $\CG$, and
assume that the opposite graph of $C$ has no non-strict cycle.  Let $P$ be a
set with $V_n\subseteq P\subseteq V_n\cup V_t$, that is, a possible null
pattern of a combined assignment, let $\OG_P$ be as in
Lemma~\ref{lem:component-pinning-nulls}, and let
$\theta$ be a combined assignment with null pattern $P$ such that $\theta(\bar x)\notin q'(\theta(D_q))$. Then:
\begin{enumerate}
  \item The set $U$ of cycle reverse edges of $\OG_P$ violated by $\theta$ is
  an $\mfac$ of $\OG_P$.

  \item Every minimal $\mfac$ of $\OG_P$ is the set of surviving edges of a
  minimal $\mfac$ of the opposite graph of $C$.  In particular, if $\ell_C$
  is the number of minimal $\mfac$s of the original opposite graph, then
  $\OG_P$ has at most $\ell_C$ minimal $\mfac$s.

  \item For every minimal $\mfac$ $E$ of $\OG_P$ there is a single local
  assignment $\theta_{P,E}$ of the variables of $C$ outside $P$ such that, if
  $E\subseteq U$, the combined assignment $\theta^*$ that agrees with
  $\theta_{P,E}$ on the variables of $C$ outside $P$ and with $\theta$
  everywhere else satisfies $  \theta^*(\bar x)\notin q'(\theta^*(D_q))$.
  
\end{enumerate}
\end{lemma}

\begin{proof}
\looseness=-1 Every directed cycle of $\OG_P$ is a directed cycle of the opposite graph of
$C$ and therefore contains a strict edge. Write $V^P$ for the vertices of
$\OG_P$ and $E^P_{\tox q}$ for its reverse edges, and call an edge of the
opposite graph of $C$ \emph{surviving} if both its endpoints lie outside $P$.

\emph{Part (1).} Suppose that $(V^P,E_q\cup(E^P_{\tox q}\setminus U))$ has a
directed cycle $\gamma$. Every reverse edge on $\gamma$ lies on a directed
cycle of $\OG_P$ and is not in $U$, so $\theta$ satisfies it, and $\theta$
satisfies every edge of $E_q$. Hence $\theta$ is nondecreasing around
$\gamma$ and constant on it, contradicting the strict edge of $\gamma$. So
$U$ is an $\mfac$ of $\OG_P$, and it contains a minimal one.

\emph{Part (2).} Let $E$ be a minimal $\mfac$ of $\OG_P$, and let $F$ be the
set of cycle reverse edges of the opposite graph of $C$ with an endpoint in
$P$. A vertex in $P$ carries only reverse edges, so every directed cycle
through such a vertex uses an edge of $F$, and every other directed cycle
lies in $\OG_P$ and uses an edge of $E$. Hence $E\cup F$ is an $\mfac$ of the
opposite graph of $C$ and contains a minimal one, say $E''$. The surviving
edges of $E''$ form a subset of $E$ that still breaks every cycle of $\OG_P$,
since a cycle of $\OG_P$ avoiding them avoids all of $E''$. By minimality of
$E$ this subset is $E$ itself. Thus $E$ is the surviving part of $E''$, and
distinct $E$ have distinct $E''$, so $\OG_P$ has at most $\ell_C$ minimal
$\mfac$s.

\looseness=-1 \emph{Part (3), the assignment.} Let $E$ be a minimal $\mfac$ of $\OG_P$ and
let $H=(V^P,E_q\cup(E^P_{\tox q}\setminus E))$, a DAG. The feedback query
$q_E$ adds the relation asserted by every cycle reverse edge of $\OG_P$
outside $E$ and the negation of every edge of $E$. Its atoms are satisfiable:
Lemma~\ref{lem:reordering-transport} applied to $H$ gives an assignment that
is strict along every edge of $H$, and the path from $v$ to $u$ that $H$
contains for each $e=(u\overset{\comp}{\to}v)\in E$, by minimality of $E$,
makes it satisfy the added negation $v\altcomp u$.

\looseness=-1 The opposite graph of $q_E$, read on $\OG_P$, is acyclic, because each of its
edges is realized by a directed path of $H$. An edge of $E_q$ or a kept cycle
reverse edge is an edge of $H$. The negation added for $e\in E$ is realized
by the path from $v$ to $u$ just mentioned. A tightened endpoint edge
coincides with the added variable--constant atom. A reverse edge of $q_E$
comes from a reverse edge of $\OG_P$ that $q_E$ leaves undecided, since
adding atoms only shrinks partial matches, so it lies outside
$\cedge_P(C)$ and belongs to $H$. A directed cycle would therefore expand to
a directed cycle of the DAG $H$.

\looseness=-1 Lemma~\ref{lem:cycle-completion} applied to $q_E$ on $\OG_P$ now yields a
single local assignment $\theta_{P,E}$ of the surviving variables, with
values canonical for $q$ by the rank argument in the proof of
Lemma~\ref{lem:component-pinning-nulls}. By Lemma~\ref{lem:cycle-completion}(3)
and the atoms of $q_E$, it satisfies every edge of $H$ and violates every
edge of $E$. Since the opposite graph of $q_E$ is acyclic, no two surviving
variables are entailed equal and no variable is pinned to a constant, so all
values of $\theta_{P,E}$ are open-interval representatives, distinct for
distinct variables and taken by no variable outside $C$.

\emph{Part (3), the transport.} Assume $E\subseteq U$ and let $\theta^*$ be
as stated. Suppose that a satisfying valuation $h$ of $q'$ over
$\theta^*(D_q)$ has $h(\bar x')=\theta^*(\bar x)$, and define $h'$ from $h$
as in the proof of Lemma~\ref{lem:component-pinning-nulls}. Well-definedness
follows from the injectivity just noted, and relational atoms, the null
pattern, and the head transport as there. Let $\alpha$ be a comparison atom
of $q'$ satisfied by $h$, with realizers chosen positionally and induced
condition $w_1\comp w_2$. If no realizer lies in $C$, then $\theta$ and
$\theta^*$ agree on $w_1,w_2$. Otherwise $w_1,w_2$ are vertices of $\OG_P$
and there are three cases. If the atoms of $q$ and the order of the constants
decide the condition, its truth value is the same under $\theta$ and
$\theta^*$. If its reverse edge lies in $E$, then $E\subseteq U$ says that
$\theta$ violates it, and $\theta_{P,E}$ violates it too, so the condition
holds under both assignments. Otherwise its reverse edge lies in $H$ and is
satisfied by $\theta^*$, so the condition fails under $\theta^*$,
contradicting that $h$ satisfies $\alpha$. Hence $h'$ satisfies every
comparison atom of $q'$ and witnesses $\theta(\bar x)\in q'(\theta(D_q))$,
a contradiction.
\end{proof}

\begin{proof}[Proof of the complexity claims of Corollary~\ref{cor:fpt-full}.]
 \looseness=-1 By Corollary~\ref{cor:combined-counterexample} and the proof of
Theorem~\ref{thm:comb-full-mult}, containment holds if and only if
$\psi(\bar x)\in q'(\psi(D_q))$ for every combined assignment $\psi$ in the
families of the at most $3^{|\cedge(q)|}$ cases.  Assume that $\Delta$, $w$,
and $|\cedge(q)|$ are constant.  Some legal separator $S$ of $\RG^{+}$ has
$|S|+\max_i|C_i|\le w$, and in particular $|S|\le w$, so such a separator is
found by enumerating the polynomially many subsets of size at most $w$ and
checking legality and width directly.  The witness sets, the canonical values,
the graph $\RG^{+}$, and the local assignments of
Lemma~\ref{lem:cycle-completion} are all computable in polynomial time, and
each case contributes at most $(2\Delta+2)^{w}$ assignments, so the whole
family consists of constantly many databases, each of size polynomial in the
queries.  Testing $\psi(\bar x)\in q'(\psi(D_q))$ amounts to guessing a
valuation of $q'$ and verifying its atoms, so the conjunction of constantly
many such tests lies in \textnormal{NP}.  When $q'$ admits polynomial
evaluation, each test runs in polynomial time, and containment is in
\textnormal{P}.
\end{proof}

\begin{proof}[Proof of the feedback claim of Corollary~\ref{cor:fpt-full}.]
 \looseness=-1 Replace Lemma~\ref{lem:component-pinning-nulls} by
Lemma~\ref{lem:feedback-pinning-nulls} in the proof of Theorem~\ref{thm:comb-full-mult}.  On each
nontrivial component, a violating assignment with null pattern $P$ violates a
set of cycle reverse edges that is an $\mfac$ of $\OG_P$ and contains a
minimal one.  Forcing the components one at a time preserves the violation.  A
case is now a minimal $\mfac$ of the original opposite graph of $C$, and under
$P$ it contributes the local assignment of its surviving part exactly when
that surviving part is a minimal $\mfac$ of $\OG_P$, and nothing otherwise.
By part~(2) of Lemma~\ref{lem:feedback-pinning-nulls}, every local assignment
that is needed arises this way.  Thus a component contributes at most $\ell_C$
cases, and the count within a case is unchanged.  The factor
$3^{|\cedge(q)|}$ is therefore replaced by $ \prod_C \ell_C$.
\end{proof}

\begin{example}[The combined test, worked]
\label{exm:combined-worked}
A genealogy records each child's birth year in
$\mathsf{Person}(\mathit{parent},\mathit{child},\mathit{birthYr})$, where the
parent may be unknown, and maintenance events are timestamped by the hour in
$\mathsf{Backup}$, $\mathsf{Audit}$, $\mathsf{Archive}$, and $\mathsf{Purge}$.
Consider
\[
\begin{aligned}
q(p_1,p_2)\leftarrow{}&
 \mathsf{Person}(p_1,p_2,b_2),\ \mathsf{Person}(p_2,p_3,b_3),\
 \mathsf{Backup}(x_1),\ \mathsf{Audit}(x_2),\\
&\mathsf{Archive}(x_3),\ \mathsf{Purge}(x_4),\quad 0<x_1,x_2,x_3,x_4<10,\\
q'(u_1,u_2)\leftarrow{}&
 \mathsf{Person}(u_1,u_2,c_2),\ \mathsf{Person}(u_1,u_3,c_3),\
 \mathsf{Backup}(s_1),\ \mathsf{Audit}(t_1),\\
&\mathsf{Archive}(g_1),\ \mathsf{Purge}(h_1),\
 \mathsf{Backup}(s_2),\ \mathsf{Audit}(t_2),\\
& s_1\le t_1,\ t_1\le g_1,\ g_1\le h_1,\ t_2\le s_2.
\end{aligned}
\]
\looseness=-1 Query $q$ asks for a grandparent--parent pair and four  events in
the first ten hours. Query $q'$ asks for a parent with two children, and for
events in which a backup precedes an audit, an audit an archive, and an
archive a purge, while a second audit follows a second backup.

\looseness=-1 The join variables are
$\join(q)=\{p_2\}$ and $\join(q')=\{u_1\}$, and the compared variables are
$\compvars(q)=\{x_1,\dots,x_4\}$ and $\compvars(q')=\{s_1,s_2,t_1,t_2,g_1,h_1\}$.
The non-join variables
$p_1,p_3$ are covered, by $u_1$ and $u_2$, and only $p_1$ shares its position
set with a variable of $\nonnull(q')$, namely $u_1$. Hence $V_t=\{p_1\}$,
$V_n=\{b_2,b_3\}$, and $V_f$ holds the other six variables. The birth years
are nulled outright, and only the unknown grandparent is toggled.

\looseness=-1 
 We compute the three parameters of Theorem~\ref{thm:comb-full-mult}.
The match sets are $\tox{p_1}=\{u_1\}$, $\tox{p_2}=\{u_1,u_2,u_3\}$,
$\tox{p_3}=\{u_2,u_3\}$, $\tox{b_2}=\tox{b_3}=\{c_2,c_3\}$,
$\tox{x_1}=\{s_1,s_2\}$, $\tox{x_2}=\{t_1,t_2\}$, $\tox{x_3}=\{g_1\}$, and
$\tox{x_4}=\{h_1\}$, so $\Delta=3$. The variable--variable atoms of $q'$
induce $x_1\le x_2$, $x_2\le x_3$, $x_3\le x_4$, and $x_2\le x_1$, none
decided by $q$, so $\CG$ has the single nontrivial component
$C=\{x_1,\dots,x_4\}$, with $\compCons C=\{0,10\}$. These comparisons
contribute the strict reverse edges $x_2\to x_1$, $x_3\to x_2$,
$x_4\to x_3$, and $x_1\to x_2$, whose only directed cycle is
$x_1\to x_2\to x_1$. Hence $|\cedge(q)|=2$, and the two singleton edge sets
are the minimal $\mfac$s, so $\ell_C=2$. For the width, note that each
isolated variable has one canonical value, because the variables of $q'$
matching it carry no comparison. Hence $\valcnt(p_1)=2$, as $p_1$ may also be nulled, while
$\valcnt(p_2)=\valcnt(p_3)=\valcnt(b_2)=\valcnt(b_3)=1$, and
$\valcnt(x_i)=1$ because a case pins $C$ to one value and $x_i\in V_f$.
Moreover, $\RG^{+}$ is the disjoint union of a clique on
$\{p_1,p_2,p_3,b_2,b_3\}$, from the two $\mathsf{Person}$ atoms and their
covering-matches, and a clique on $C$. So the empty separator is legal and
$w=5$.

\looseness=-1 All four variables of $C$
are active on $(0,10)$ and $q$ entails no order
among them, so each keeps all four representatives and $C$ has $4^4=256$
local assignments, which the two toggles of $p_1$ double to
$|\Dcomb(q,q')|=512$. Theorem~\ref{thm:comb-full-mult} tests instead
$3^{|\cedge(q)|}\cdot\max\{2,1\}=18$ canonical databases, and the feedback
improvement of Corollary~\ref{cor:fpt-full} only $\ell_C\cdot 2=4$. Adding
further maintenance events lengthens the induced chain while leaving
$\cedge(q)$ and the value counts unchanged.
\end{example}

\end{document}